\documentclass[prodmode,acmtocl]{acmsmall}
\newcommand{\qedhere}{\hfill\qed}

\usepackage{amsmath}
\usepackage{array}
\usepackage{cite}

\def\proofmove{\proof}\def\endproofmove{\endproof}\def\fixstatement#1{}

\usepackage{tikz}
\usetikzlibrary{shapes,fit}
\usetikzlibrary{positioning,chains,matrix,arrows,decorations.pathmorphing,calc}
\usetikzlibrary{decorations.markings}
\usepackage{wrapfig}

\usepackage{xcolor}
\definecolor{vgreen}{rgb}{.1,.5,0}
\definecolor{vred}{rgb}{.7,0,0}
\definecolor{vblue}{rgb}{.1,.15,.62}
\pgfdeclarelayer{background}
\pgfsetlayers{background,main}

\usepackage{hyperref}

\usepackage[bbsets,Dfprime]{math}
\usepackage[modernsign,substopindex,shortmquant,mquantifiertype,mconnectiveformal,bracketinterpret,modifopindex,seqarrow,seqoptional,sidenotecalculus,abbrseqcontext,shortterms,nosigmaterms,novarterms]{logic}
\usepackage[pretest,nocommandblocks]{progreg}
\usepackage[bracketinterpret]{dL}

\usepackage{prettyref}
\newcommand{\rref}[2][]{\prettyref{#2}}
\newrefformat{sec}{Section\,\ref{#1}}
\newrefformat{app}{Appendix\,\ref{#1}}
\newrefformat{def}{Def.\,\ref{#1}}
\newrefformat{thm}{Theorem\,\ref{#1}}
\newrefformat{prop}{Proposition\,\ref{#1}}
\newrefformat{lem}{Lemma\,\ref{#1}}
\newrefformat{cor}{Corollary\,\ref{#1}}
\newrefformat{rem}{Remark\,\ref{#1}}
\newrefformat{ex}{Example\,\ref{#1}}
\newrefformat{tab}{Table\,\ref{#1}}
\newrefformat{fig}{Fig.\,\ref{#1}}
\newrefformat{case}{case\,\ref{#1}}
\newrefformat{foot}{Footnote\,\ref{#1}}

\renewcommand{\usubst}[3][]{\subst[#1]{#2(\cdot)}{#3(\cdot)}}%
\renewcommand{\preusubst}[2][]{#1}%

\DeclareMathOperator{\leaf}{leaf}

\newcommand{\inflop}[2][]{\tau^{#1}(#2)}
\newcommand{\inflopstrat}[1][]{\def\inflopstratargI{{#1}}\inflopstratRelay}
\newcommand{\inflopstratRelay}[2][]{\varsigma_{#1}^{\inflopstratargI}(#2)}

\renewcommand{\ltrue}{\top}
\renewcommand{\lfalse}{\bot}

\newcommand{\dGLordinal}{\omega_1^{\text{CK}}}%

    \tikzstyle{gamelabel}=[vgreen,above,label distance=0mm]
    \tikzstyle{and}=[draw,dashed,shape=diamond,inner sep=1pt]
    \tikzstyle{or}=[draw,solid,shape=rectangle,inner sep=2pt]
    \tikzstyle{straight action}=[fill=none,above,draw=none,shape=rectangle,set style={{edge from parent}=[]},set style={{every node}=[]}]
    \tikzstyle{action}=[sloped,straight action]
    \tikzstyle{subgame}=[left=-4pt,draw=none,shape=rectangle]
    \tikzstyle{subgamea}=[solid,black,decorate,decoration={saw,amplitude=2pt}]
    \tikzstyle{psubgame}=[set style={{edge from parent}+=[subgamea]}]
    \tikzstyle{someone}=[draw,solid,shape=circle,black]
    \tikzstyle{diamond}=[and,vred]
    \tikzstyle{box}=[or,vblue,minimum width=0.4cm,minimum height=0.4cm]
    \tikzstyle{diamonda}=[decorate=false,dashed,vred]
    \tikzstyle{boxa}=[decorate=false,solid,vblue]
    \tikzstyle{diamondastrategy}=[diamonda,thick]
    \tikzstyle{boxastrategy}=[boxa,thick]
    \tikzstyle{pdia}=[set style={{edge from parent}+=[diamonda]}]
    \tikzstyle{pbox}=[set style={{edge from parent}+=[boxa]}]
    \tikzstyle{pdiastrategy}=[set style={{edge from parent}+=[diamondastrategy]}]
    \tikzstyle{pboxstrategy}=[set style={{edge from parent}+=[boxastrategy]}]
    \tikzstyle{diawon}=[label={[vred]below:$\boldsymbol{\diamond}$}]
    \tikzstyle{boxwon}=[label={[vblue]below:\rotatebox{45}{$\boldsymbol{\diamond}$}}]%
    \tikzstyle{diawonstrategy}=[label={[vred,draw,solid,inner sep=0pt,circle]below:$\boldsymbol{\diamond}$}]
    \tikzstyle{boxwonstrategy}=[label={[vblue,draw,solid,inner sep=-1pt,circle]below:\rotatebox{45}{$\boldsymbol{\diamond}$}}]
    \tikzstyle{etc}=[label={[black]below:$\vdots$}]
    \tikzstyle{backedge}=[double,double distance=2pt]
    \tikzstyle{endgame}=[draw=none,fill=gray!10,rounded corners]

    \tikzstyle{box}+=[fill=vblue!10]
    \tikzstyle{diamond}+=[fill=vred!10]

\newsavebox{\tmpdiawon}
\sbox{\tmpdiawon}{\textcolor{vred}{$\boldsymbol{\diamond}$}}
\newsavebox{\tmpdiawonstrategy}
\sbox{\tmpdiawonstrategy}{\tikz{\node[vred,draw,solid,inner sep=0pt,circle]{$\boldsymbol{\diamond}$};}}
\newsavebox{\tmpboxwon}
\sbox{\tmpboxwon}{\textcolor{vblue}{\rotatebox{45}{$\boldsymbol{\diamond}$}}}
\newsavebox{\tmpboxwonstrategy}
\sbox{\tmpboxwonstrategy}{\tikz{\node[vblue,draw,solid,inner sep=-1pt,circle]{\rotatebox{45}{$\boldsymbol{\diamond}$}};}}
\newsavebox{\closurallI}
\sbox{\closurallI}{$\forall$I}
    
\renewcommand{\linterpretationsname}{\mathcal{S}}
\newcommand{\stdI}{\dLint[const=I,state=s]}
\newcommand{\I}{\dLint[const=I,state=s]}
\newcommand{\It}{\dLint[const=I,state=t]}
\newcommand{\Iu}{\dLint[const=I,state=u]}
\newcommand{\J}{\dLint[const=J,state=s]}

\renewcommand{\mapply}[3][]{#2(#3)}

\newcommand*{\strategyfor}[2][]{\varsigma_{#1}(#2)}
\newcommand*{\dstrategyfor}[2][]{\delta_{#1}(#2)}

\newcommand{\sol}{x}%
\newcommand{\solf}{y}%

\newcommand{\stime}{\sol_0}%
\newcommand{\ignore}[1]{}
\newcommand{\bebecomes}{\mathrel{::=}}
\newcommand{\alternative}{~|~}
\renewcommand{\with}{:}

\newcommand*{\genDE}[1]{\theta}%
\newcommand{\ivr}{\psi}
\newcommand{\inv}{\varphi}
\newcommand{\var}{\varphi}

\newcommand{\FOD}{\text{FOD}\xspace}%
\newcommand{\muD}{\ensuremath{L_{\mu\text{D}}}\xspace}%
\newcommand{\DL}{\text{DL}\xspace}%

\newcommand{\reduct}[1]{#1^\flat}%

\newcommand*{\pair}[2]{\lbanana#1,#2\rbanana}%
\newcommand*{\tupleempty}{\lbanana\rbanana}%
\newcommand{\W}{WALL$\boldsymbol{\cdot}$E\xspace}%
\newcommand{\E}{EVE\xspace}%

\usepackage{ifpdf}
\ifpdf
\fi

\setcopyright{rightsretained}

\doi{0000001.0000001}

\issn{1234-56789}

\begin{document}

\markboth{A. Platzer}{Differential Game Logic}

\title{Differential Game Logic}
\author{ANDR\'E PLATZER
\affil{Carnegie Mellon University}}

\begin{abstract}
\emph{Differential game logic} (\dGL) is a logic for specifying and verifying properties of \emph{hybrid games}, i.e.\ games that combine discrete, continuous, and adversarial dynamics.
Unlike hybrid systems, hybrid games allow choices in the system dynamics to be resolved adversarially by different players with different objectives.
The logic \dGL can be used to study the existence of winning strategies for such hybrid games, i.e.\ ways of resolving the player's choices in \emph{some} way so that he wins by achieving his objective for \emph{all} choices of the opponent.
Hybrid games are determined, i.e.\ from each state, one player has a winning strategy, yet computing their winning regions may take transfinitely many steps.
The logic \dGL, nevertheless, has a sound and complete axiomatization relative to any expressive logic.
Separating axioms are identified that distinguish hybrid games from hybrid systems.
Finally, \dGL is proved to be strictly more expressive than the corresponding logic of hybrid systems by characterizing the expressiveness of both.
\end{abstract}



\acmformat{Platzer, A.  YYYY. Differential game logic.}

\begin{bottomstuff}
This material is based upon work supported by the National Science Foundation under
NSF CAREER Award CNS-1054246.
A preliminary version has appeared as a report \cite{Platzer12:dGL,Platzer13:dGL}.

Author's address: A. Platzer, Computer Science Department,
Carnegie Mellon University.
\end{bottomstuff}

\maketitle

\section{Introduction}
\irlabel{MP|MP}%
\newsavebox{\Rval}%
\sbox{\Rval}{$\scriptstyle\mathbb{R}$}%
\irlabel{qear|\usebox{\Rval}}%
\irlabel{RCFp|\usebox{\Rval}}%
\irlabel{gena|$\forall$}%

\emph{Hybrid systems} \cite{DBLP:conf/hybrid/NerodeK92a,DBLP:journals/tcs/AlurCHHHNOSY95,DBLP:journals/tac/BranickyBM98,DavorenNerode_2000} are dynamical systems combining discrete dynamics and continuous dynamics. They are widely important, e.g., for modeling how computers control physical systems such as cars \cite{DBLP:conf/hybrid/DeshpandeGV96}, aircraft \cite{DBLP:conf/hybrid/UmenoL07} and other cyber-physical systems.
Hybrid systems combine difference equations (or discrete assignments) and differential equations with conditional switching, nondeterministic choices, and repetition \cite{Platzer10}.
Hybrid systems are not semidecidable \cite{DBLP:conf/stoc/HenzingerKPV95}, %
but nevertheless studied by many successful verification approaches \cite{DBLP:conf/lics/Platzer12a,DoyenFPP16}. They have a complete axiomatization relative to differential equations in \emph{differential dynamic logic} (\dL) \cite{DBLP:journals/jar/Platzer08,DBLP:conf/lics/Platzer12b}, which extends Pratt's dynamic logic of conventional discrete programs \cite{DBLP:conf/focs/Pratt76} to hybrid systems by adding differential equations and a reachability relation semantics on the real Euclidean space.
\emph{Hybrid games} \cite{DBLP:conf/hybrid/NerodeRY96,DBLP:journals/tac/TomlinPS:98,DBLP:conf/concur/HenzingerHM99,DBLP:journals/IEEE/TomlinLS00,DBLP:journals/jopttapp/DharmattiDR06,DBLP:journals/corr/abs-0911-4833,DBLP:journals/tcs/VladimerouPVD11} are games of two players on a hybrid system.
Hybrid games add an adversarial dynamics to hybrid systems, i.e.\ an adversarial way of resolving the choices in the system dynamics.
Both players can make their respective choices arbitrarily. They are not assumed to cooperate towards a common goal but may compete.
The prototypical example of a hybrid game is RoboCup, where two (teams of) robots move continuously on a soccer field subject to the discrete decisions of their respective control programs, and they resolve their choices adversarially in active competition for scoring goals.
Worst-case verification of many other situations leads to hybrid games. Two robots may already end up in a hybrid game if they simply do not know each other's objectives, because worst-case analysis assumes they might interfere, which makes them compete accidentally rather than on purpose.
The former situation is \emph{true competition}, the latter \emph{analytic competition}, because possible competition was assumed for the sake of a worst-case analysis.
Aircraft separation provides further natural scenarios for both true \cite{Isaacs:DiffGames} and analytic \cite{DBLP:journals/tac/TomlinPS:98} competition.
Hybrid games are also fundamental for security questions about hybrid systems, which intrinsically involve adversarial situations with more than one player.
Many different variations of hybrid games are interesting for applications \cite{DBLP:journals/tac/TomlinPS:98,DBLP:journals/IEEE/TomlinLS00,DBLP:conf/concur/HenzingerHM99,DBLP:conf/ecc/PrandiniHP01,DBLP:journals/corr/abs-0911-4833,DBLP:journals/tcs/VladimerouPVD11,DBLP:conf/cade/QueselP12}, including games between controller and plant for control synthesis, hybrid pursuit-evader games, or hybrid games for verification of robot controllers against an uncertain environment or an external disturbance.

This article does not focus on one such fixed pattern of game interaction in hybrid systems, but considers a more general framework for hybrid game interactions of two players.
While the results of this article show promise in practice, the focus of this article is on fundamental logical considerations for hybrid games.
It develops a compositional programming language for hybrid games and a logic for hybrid games along with its fundamental compositional proof principles.
The article analyzes hybrid games and contrasts them with hybrid systems in terms of their analytic complexity, axiomatizations, and expressiveness.

\paragraph{Approach}

This article studies a compositional model of hybrid games obtained as a programming language from a compositional model of hybrid systems \cite{DBLP:journals/jar/Platzer08} by simply adding the duality operator $\pdual{}$ for passing control between the players.
The dual game $\pdual{\alpha}$ is the same as the hybrid game $\alpha$ with the roles of the players swapped, much like what happens when turning a chessboard around by $180^\circ$ so that players black and white swap sides.
Hybrid games without $\pdual{}$ are single-player, like hybrid systems are, because $\pdual{}$ is the only operator where control passes to the other player.
Hybrid games with $\pdual{}$ give both players control over their respective choices (indicated by $\pdual{}$). They can play in reaction to the outcome that the previous choices by the players have had on the state of the system.
The fact that $\pdual{}$ is an operator on hybrid games makes them fully symmetric. That is, they allow arbitrary combinations of all operators at arbitrary nesting depths to define the game, not just a single fixed pattern like, e.g., the separation into a single loop of a continuous plant player and a discrete controller player that has been predominant in other approaches.

Hybrid games are game-theoretically reasonably tame sequential, non-cooperative, zero-sum, two-player games of perfect information with payoffs $\pm1$.\footnote{Draws, coalitions, rewards, and payoffs different from $\pm1$ etc.\ are expressible in the logic developed in this article, which gains simplicity and elegance by focusing on the most fundamental case of hybrid games.}
What makes them challenging is that they are played on hybrid systems, which causes reachability computations and the canonical game solution technique of backwards induction for winning regions to take infinitely many iterations (${\geq}\omega_1^{\text{CK}}$) to terminate.

One of the most fundamental questions about a hybrid game is whether the player of interest has a \emph{winning strategy},
i.e.\ a way of resolving his choices that will lead to a state in which that player wins, no matter how the opponent player resolves his respective choices.\footnote{\label{foot:exhibit-winning-strategy}%
A closely related question is about ways to exhibit that winning strategy, for which existence is a prerequisite and a constructive proof is a representation of that winning strategy.
A proof is a certificate witnessing the existence of a winning strategy.
As soon as one knows from which states a winning strategy exists, local search in the action space would be enough. But search may still be challenging in dense action spaces.}
If the player has such a winning strategy, he can achieve his objectives no matter what the opponent does, otherwise he needs his opponent to cooperate.
This article introduces a logic and proof calculus for hybrid games and thereby decouples the questions of truth (existence of winning strategies) and proof (winning strategy certificates) and proof search (automatic construction of winning strategies).
It studies provability (existence of proofs) and the proof theory of hybrid games and identifies what the right proof rules for hybrid games are (soundness \& completeness).

This article presents \emph{differential game logic} (\dGL) and its axiomatization for studying the existence of winning strategies for hybrid games.
It generalizes hybrid systems to hybrid games by adding the duality operator $\pdual{}$ and a winning strategy semantics on the real Euclidean space.
Hybrid games simultaneously generalize hybrid systems \cite{DBLP:conf/hybrid/NerodeK92a,DBLP:journals/tcs/AlurCHHHNOSY95} and discrete games \cite{DBLP:conf/focs/Parikh83,DBLP:journals/anndiscrmath/Parikh85,DBLP:journals/sLogica/PaulyP03a}. %
Similarly, \dGL simultaneously generalizes logics of hybrid systems and logics of discrete games.
The logic \dGL generalizes differential dynamic logic (\dL) \cite{DBLP:journals/jar/Platzer08,DBLP:conf/lics/Platzer12b} from hybrid systems to hybrid games with their adversarial dynamics and, simultaneously, generalizes Parikh's propositional game logic \cite{DBLP:conf/focs/Parikh83,DBLP:journals/anndiscrmath/Parikh85,DBLP:journals/sLogica/PaulyP03a} from games on finite-state discrete systems to games on hybrid systems with their differential equations, uncountable state spaces, uncountably many possible moves, and interacting discrete and continuous dynamics.

\paragraph{Contributions}
Every particular play of a hybrid game has exactly one winner (\rref{sec:dGL}).
From each state exactly one player has a winning strategy no matter how the opponent reacts (determinacy, \rref{sec:MetaProperties}).
The \dGL proof calculus can be used to find out which of the two players it is that has a winning strategy from which state (\rref{sec:dGL-axiomatization}).
The logic \dGL for hybrid games is proved to be fundamentally more expressive than the logic \dL for hybrid systems by characterizing the expressiveness of both (\rref{sec:dGL-dL-expressiveness}).

The primary contributions of this article are as follows.
The logic \dGL identifies the logical essence of hybrid games and their game combinators.\footnote{Hybrid games only lead to a minor syntactic change compared to hybrid systems (the addition of $\pdual{}$), yet one that entails pervasive semantical reconsiderations, because the semantic basis for assigning meaning to operators changes in the presence of adversarial interactions. This change leads to more expressiveness. It is a sign of logical robustness that this results in a surprisingly small change in the axiomatization.
Overall, the changes induced by dualities are in some ways radical, yet, in other ways surprisingly smooth.
}
It identifies a simple, algebraic, compositional model of hybrid games as a programming language for hybrid games by adding the control switching operator $\pdual{}$ to a programming language for hybrid systems \cite{DBLP:conf/lics/Platzer12b} and reinterpreting its compositional operators as operators on hybrid games.
This article introduces differential game logic for hybrid games with a simple modal semantics and a simple compositional proof calculus, which is proved to be a sound and complete axiomatization relative to any expressive logic.
Completeness for game logics is a subtle problem. Completeness of propositional discrete game logic has been an open problem for 30 years \cite{DBLP:conf/focs/Parikh83}. This article focuses on more general hybrid games and proves a generalization of Parikh's calculus to be relatively complete for hybrid games.
The completeness proof is constructive and identifies a fixpoint-style proof technique, which can be considered a modal analogue of characterizations in the Calculus of Constructions \cite{DBLP:journals/iandc/CoquandH88}.
This technique is practical for hybrid games, and also easier for hybrid systems than previous proof techniques.
These results suggest hybrid game versions of influential views of understanding program invariants as fixpoints \cite{DBLP:conf/popl/CousotC77,Clarke79}.
Harel's convergence rule \cite{DBLP:conf/stoc/HarelMP77}, which poses practical challenges for hybrid systems verification, now turns out to be unnecessary for hybrid games, hybrid systems, and programs.
All separating axioms are identified that capture the logical difference of hybrid systems versus hybrid games.
Hybrid games are proved to be determined, i.e.\ in every state, exactly one player has a winning strategy, which is the basis for assigning classical truth to logical formulas that refer to winning strategies of hybrid games.
Winning regions of hybrid games are characterized by fixpoints of a monotone operator, which can be obtained by iteration but that iteration may only stop after ${\geq}\omega_1^{\text{CK}}$ many steps.
Hybrid games are shown to be a fundamental extension of hybrid systems by proving that the logic \dGL for hybrid games is fundamentally more expressive than the corresponding logic \dL for hybrid systems, which is related to long-standing unsolved questions in the propositional case \cite{DBLP:journals/anndiscrmath/Parikh85,DBLP:journals/mst/BerwangerGL07}.
This separation also characterizes the expressiveness of \dL and of \dGL.

\paragraph{Structure of this Article}

The syntax and denotational semantics of \dGL are introduced in \rref{sec:dGL}.
\rref{sec:MetaProperties} establishes meta-properties, including determinacy (\rref{sec:Determinacy}), 
equivalences of hybrid games and reductions eliminating evolution domains (\rref{sec:HG-equivalences}),
and an analysis proving that winning regions of hybrid games can be characterized by iteration of a monotone image computation operator until a fixpoint, which may only stop after ${\geq}\omega_1^{\text{CK}}$ many steps (\rref{sec:ClosureOrdinals}) unlike hybrid systems whose corresponding closure ordinal is $\omega$.
\rref{sec:dGL-axiomatization} presents an axiomatization of \dGL as a Hilbert-type proof calculus (\rref{sec:dGL-calculus}) that is proved sound (\rref{sec:dGL-sound}) and complete (\rref{sec:dGL-complete}) relative to any differentially expressive logic, which are exemplified subsequently (\rref{sec:differentially-expressive}).
The axiomatic separation between hybrid systems and hybrid games is identified in \rref{sec:separating-axioms} and proved in \rref{app:separating-axioms}.
Hybrid games are proved more expressive than hybrid systems in \rref{sec:dGL-dL-expressiveness}, where the expressiveness of the hybrid games logic \dGL and the hybrid systems logic \dL are characterized.
Related work is discussed in \rref{sec:RelatedWork}, concluding in \rref{sec:Conclusion}.

Example proofs in the \dGL calculus are shown in \rref{app:ExampleProofs}.
\rref{app:separating-axioms} proves the axiomatic separation from \rref{sec:separating-axioms}.
For reference and to support the interactive intuition of game play, an operational semantics of hybrid games is shown in \rref{app:operational-HG-semantics}.
Alternative semantics for repetitions of hybrid games are contrasted in \rref{app:alternative-semantics}.
Finally, \rref{app:dGL-closure-lower} provides concrete hybrid games to support the computational intuition for the higher closure ordinals proved generically
 in \rref{sec:ClosureOrdinals}.

\section{Differential Game Logic} \label{sec:dGL}

A robot is a canonical example of a hybrid system.
Suppose a robot, \W, is running around on a planet collecting trash.
His dynamics is that of a hybrid system, because his continuous dynamics comes from the differential equations describing his continuous physical motion in space, while his discrete dynamics comes from his computer-based control decisions about when to move in which direction and when to stop moving in order to gather trash.
As soon as \W meets another robot, \E, however, her presence changes everything for him.
If \W neither knows how \E is programmed nor exactly what her goal is, then the only safe thing he can assume about her is that she might do anything within her physical capabilities.
It takes the study of a hybrid game to find out whether or not \W can use his choices in some way to reach his goal, say, collecting trash and avoiding collisions with \E, regardless of how \E chooses her actions.

The hybrid games considered here have no draws.\footnote{
For applications with \emph{draws}, it is easy to follow Zermelo \citeyear{Zermelo13} and compare two games, one for each player, that attribute draws pessimistically as losses.
Draws in the original game result from those states from which both players would lose their games when considering draws pessimistically as their respective losses.
Draws come from states where no player can make sure to win-and-not-draw but only to win-or-draw.
Hybrid games with draws are represented using two modalities of the logic \dGL developed in this article.
}
For any particular play of the \W and \E game, for example, either \W achieves his objective or he does not. There is no in between.
When a hybrid game expects a player to move, but the rules of the game do not permit any of his moves from the current state, then that player loses right away (he \emph{deadlocks}).
If the game completes without deadlock, the player who reaches one of his winning states wins.
Thus, exactly one player wins each (completed) game play for complementary winning states, because the games are uncountably-infinite state but of arbitrary unbounded but finite duration.
So reaching a winning state in the limit after infinitely many steps is not enough to win the game, but the players have to win in any arbitrary unbounded finite amount of time and after an arbitrary unbounded finite number of steps.
The games are \emph{zero-sum} games, i.e.\ if one player wins, the other one loses, with player payoffs $\pm1$.
Losses or victories of different payoff are not considered explicitly, because they are representable with extra variables that track payoffs.
The two players are classically called \emph{Angel}\ignore{(or $\ldiamond{}$-player)} and \emph{Demon}\ignore{(or $\lbox{}$-player)}.
By considering aggregate players, these results generalize in the usual way to the case where Angel and Demon represent coalitions of agents that work together to achieve a common goal such as the aggregate Angel player for one team and the aggregate Demon player for the other team in the case of RoboCup.

Hybrid games are non-cooperative and sequential games.
In non-cooperative games, the players can choose to act arbitrarily according to the rules represented in the game.\footnote{
Applications with \emph{cooperative games}, where players form coalitions or negotiate binding contracts, are representable in the rules of the game to track coalitions and limit player's choices according to the contracts.
}
Sequential (or dynamic) games are games that proceed in a series of steps, where, at each step, exactly one of the players can choose an action based on the outcome of the game so far. %
Concurrent games, where both players choose actions simultaneously, as well as their equivalent\footnote{\emph{Concurrent games} in which the players choose actions simultaneously can be converted into sequential games of imperfect information in which the players choose sequentially yet without information about the opponent's choice.} games of imperfect information, are interesting but not considered explicitly here \cite{DBLP:journals/jacm/AlurHK02,DBLP:conf/icla/BerwangerP09}.
Imperfect information games lead to Henkin quantifiers, not first-order quantifiers.

\subsection{Syntax} \label{sec:dGL-syntax}
Differential game logic (\dGL) is a logic for studying properties of hybrid games.
The idea is to describe the game form, i.e.\ rules, dynamics, and choices of the particular hybrid game of interest, using a program notation and to then study its properties by proving the validity of logical formulas that refer to the existence of winning strategies for objectives of those hybrid games.
Even though hybrid game forms only describe the game \emph{form} with its dynamics and rules and choices, not the objective, they are still called hybrid games.
Hybrid game forms represent the rules of the game, so during the game play, players can never win but only lose (prematurely) by violating these rules.
The actual objective for a hybrid game is defined in the modal logical formula that refers to that hybrid game form and is evaluated at the end of the game play.
Hybrid games and differential game logic formulas are defined by simultaneous induction focusing on polynomial terms for simplicity.

\begin{definition}[Hybrid games] \label{def:dGL-HG}
The \emph{hybrid games of differential game logic {\dGL}} are defined by the following grammar ($\alpha,\beta$ are hybrid games, $x$ a variable, $\theta$ a (polynomial) term, $\ivr$ is a \dGL formula):
\[
  \alpha,\beta ~\bebecomes~
  \pupdate{\pumod{x}{\theta}}
  \alternative
  \pevolvein{\D{x}=\genDE{x}}{\ivr}
  \alternative
  \ptest{\ivr}
  \alternative
  \alpha\cup\beta
  \alternative
  \alpha;\beta
  \alternative
  \prepeat{\alpha}
  \alternative
  \pdual{\alpha}
\]
\end{definition}
\begin{definition}[\dGL formulas] \label{def:dGL-formula}
The \emph{formulas of differential game logic {\dGL}} are defined by the following grammar ($\phi,\psi$ are \dGL formulas, $p$ is a predicate symbol of arity $k$, $\theta_i$ are (polynomial) terms, $x$ a variable, and $\alpha$ is a hybrid game):
  \[
  \phi,\psi ~\bebecomes~
  p(\theta_1,\dots,\theta_k) \alternative
  \theta_1\geq\theta_2 \alternative
  \lnot \phi \alternative
  \phi \land \psi \alternative
  \lexists{x}{\phi} \alternative 
  \ddiamond{\alpha}{\phi}
  \alternative \dbox{\alpha}{\phi}
  \]
\end{definition}
Other operators $>,=,\leq,<,\lor,\limply,\lbisubjunct,\forall{x}{}$ or $\ltrue$ and $\lfalse$ for true and false can be defined as usual, e.g., \m{\lforall{x}{\phi} \mequiv \lnot\lexists{x}{\lnot\phi}}, including all program operators \cite{Platzer10}.
The modal formula \m{\ddiamond{\alpha}{\phi}} expresses that Angel has a winning strategy\footnote{A strategy for a player can be thought of as a function that selects one option whenever that player has a choice during the game play.
A winning strategy for a player is a way of resolving choices that will lead to a state in which that player wins, no matter how the opponent player resolves his respective choices.
The semantics of \dGL is a denotational semantics based on winning regions.
A formal definition for strategies and winning strategies is, thus, unnecessary, yet shown in \rref{app:operational-HG-semantics} for reference.}
to achieve $\phi$ in hybrid game $\alpha$, i.e.\ Angel has a strategy to reach a state satisfying \dGL formula $\phi$ when playing hybrid game $\alpha$, no matter what strategy Demon chooses.
That is, \m{\ddiamond{\alpha}{\phi}} expresses that Angel can guarantee to reach into the set of states satisfying $\phi$ but she cannot usually predict which of the states satisfying $\phi$ she will reach, because that depends on Demon's choices.
The modal formula \m{\dbox{\alpha}{\phi}} expresses that Demon has a winning strategy to achieve $\phi$ in hybrid game $\alpha$, i.e.\ a strategy to reach a state satisfying $\phi$, no matter what strategy Angel chooses.
Note that the same game is played in \m{\dbox{\alpha}{\phi}} as in \m{\ddiamond{\alpha}{\phi}} with the same choices resolved by the same players. The difference between both \dGL formulas is the player whose winning strategy they refer to.
Both use the set of states where \dGL formula $\phi$ is true as the winning states for that player.
The winning condition is defined by the modal formula, $\alpha$ only defines the hybrid game form, not when the game is won, which is what $\phi$ does.
Hybrid game $\alpha$ defines the rules of the game, including conditions on state variables that, if violated, cause the present player to lose for violation of the rules of the game.
The \dGL formulas \m{\ddiamond{\alpha}{\phi}} and \m{\dbox{\alpha}{\lnot\phi}} consider complementary winning conditions for Angel and Demon.

The \emph{atomic games} of \dGL are assignments, continuous evolutions, and tests.
The \emph{discrete assignment game} \m{\pupdate{\pumod{x}{\theta}}} instantly
 changes the value of variable $x$ to that of $\theta$ by a discrete jump without any choices to resolve.
In the \emph{continuous evolution game} \m{\pevolvein{\D{x}=\genDE{x}}{\ivr}}, the system follows the differential equation \m{\pevolve{\D{x}=\genDE{x}}} where the duration is Angel's choice, but Angel is not allowed to choose a duration that would, at any time, take the state outside the region where formula $\ivr$ holds.
In particular, Angel is deadlocked and loses immediately if $\ivr$ does not hold in the current state, because she cannot even evolve for duration 0 then without going outside $\ivr$.\footnote{\label{footnote:pevolvein}
The most common case for $\ivr$ is a formula of first-order real arithmetic. In \rref{sec:HG-equivalences}, evolution domain constraints $\ivr$ turn out to be unnecessary, because they can be defined using hybrid games.
In the ordinary differential equation \m{\pevolve{\D{x}=\genDE{x}}}, the term $\D{x}$ denotes the time-derivative of $x$ and $\genDE{x}$ is a polynomial term that is allowed to mention $x$ and other variables.
Systems of differential equations are considered vectorially.
More general forms of differential equations are possible \cite{DBLP:journals/logcom/Platzer10}, but will not be considered explicitly.
}
The \emph{test game} or \emph{challenge} \m{\ptest{\ivr}} has no effect on the state, except that Angel loses the game immediately if \dGL formula $\ivr$ does not hold in the current state.
If Angel passes the challenge \m{\ptest{\ivr}}, the game continues from the same state, otherwise she loses immediately.

The \emph{compound games} of \dGL are sequential, choice, repetition, and duals.
The \emph{sequential game} \m{\alpha;\beta} is the hybrid game that first plays hybrid game $\alpha$ and, when hybrid game $\alpha$ terminates without a player having lost already (so no challenge in $\alpha$ failed), continues by playing game $\beta$.
When playing the \emph{choice game} \m{\pchoice{\alpha}{\beta}}, Angel chooses whether to play hybrid game $\alpha$ or play hybrid game $\beta$.
Like all the other choices, this choice is dynamic, i.e.\ every time \m{\pchoice{\alpha}{\beta}} is played, Angel gets to choose again whether she wants to play $\alpha$ or $\beta$ this time.
The \emph{repeated game} \m{\prepeat{\alpha}} plays hybrid game $\alpha$ repeatedly and Angel chooses, after each play of $\alpha$ that terminates without a player having lost already, whether to play the game again or not, albeit she cannot choose to play indefinitely but has to stop repeating ultimately.
Angel is also allowed to stop \m{\prepeat{\alpha}} right away after zero iterations of $\alpha$.

Most importantly, the \emph{dual game} \m{\pdual{\alpha}} is the same as playing the hybrid game $\alpha$ with the roles of the players swapped during game play. That is Demon decides all choices in $\pdual{\alpha}$ that Angel has in $\alpha$, and Angel decides all choices in $\pdual{\alpha}$ that Demon has in $\alpha$.
Players who are supposed to move but deadlock lose, hence Demon loses in the dual game \m{\pdual{\alpha}} when he deadlocks because those correspond situations where Angel is supposed to move but deadlocks in $\alpha$ and vice versa.
Thus, while the test game \m{\ptest{\ivr}} causes Angel to lose if formula $\ivr$ does not hold, the \emph{dual test game} (or \emph{dual challenge}) \m{\pdual{(\ptest{\ivr})}} causes Demon to lose if $\ivr$ does not hold.
It is exactly the same formula $\ivr$ whose truth-value decides about the fate of the game in both cases since $\pdual{}$ does not affect the meaning of formulas within tests, but the player who loses changes. In \m{\ptest{\ivr}}, it is Angel who loses when $\ivr$ does not hold because she was supposed to move. In \m{\pdual{(\ptest{\ivr})}}, it is Demon who loses when $\ivr$ does not hold as he was supposed to move.
For example, if $\alpha$ describes the game of chess, then \m{\pdual{\alpha}} is chess where the players switch to control the other side.
Recall that hybrid games are game forms so $\pdual{}$ only affects the actions (and premature losses if no action is possible), not the objective of the game, which is what the modal formulas define.
So, Angel has the same goal $\phi$ in \(\ddiamond{\pdual{\alpha}}{\phi}\) and \(\ddiamond{\alpha}{\phi}\) but her actions in $\pdual{\alpha}$ switched to what used to be Demon's actions in $\alpha$ and vice versa, which is why it will turn out that \(\ddiamond{\pdual{\alpha}}{\phi}\) and \(\dbox{\alpha}{\phi}\) are equivalent (\rref{sec:dGL-axiomatization}).
The dual operator $\pdual{}$ is the only syntactic difference of \dGL for hybrid games compared to \dL for hybrid systems \cite{DBLP:journals/jar/Platzer08,DBLP:conf/lics/Platzer12b}, but a fundamental one, because it is the only operator where control passes from Angel to Demon or back.
Without $\pdual{}$ all choices are resolved uniformly by Angel without interaction.
The presence of $\pdual{}$ requires a thorough semantic generalization throughout the logic, though.

The logic \dGL only provides logically essential operators.
Many other game interactions for games can be defined from the elementary operators that \dGL provides.
\emph{Demonic choice} between hybrid game $\alpha$ and $\beta$ is \m{\dchoice{\alpha}{\beta}}, defined by \m{\pdual{(\pchoice{\pdual{\alpha}}{\pdual{\beta}})}}, in which either the hybrid game $\alpha$ or the hybrid game $\beta$ is played, by Demon's choice.
\emph{Demonic repetition} of hybrid game $\alpha$ is \m{\drepeat{\alpha}}, defined by \m{\pdual{(\prepeat{(\pdual{\alpha})})}}, in which $\alpha$ is repeated as often as Demon chooses to.
In \m{\drepeat{\alpha}}, Demon chooses after each play of $\alpha$ whether to repeat the game, but cannot play indefinitely so he has to stop repeating ultimately.
The \emph{dual differential equation} \m{\pdual{(\pevolvein{\D{x}=\theta}{\ivr})}} follows the same dynamics as \m{\pevolvein{\D{x}=\theta}{\ivr}} except that Demon chooses the duration, so he cannot choose a duration during which $\ivr$ stops to hold at any time. Hence he loses when $\ivr$ does not hold in the current state.
Dual assignment \m{\pdual{(\pupdate{\pumod{x}{\theta}})}} is equivalent to \m{\pupdate{\pumod{x}{\theta}}}, because it is deterministic so involves no choices.
Other program operators are also definable \cite{Platzer10}, e.g., nondeterministic assignment \(\prandom{x}\) defined by \(\pevolve{\D{x}=1};\pevolve{\D{x}=-1}\).
Unary operators (including $\prepeat{}, \pdual{}, \forall{x}{}, \dbox{\alpha}{},\ddiamond{\alpha}{}$) bind stronger than binary operators and $;$ binds stronger than $\pchoice{}{}$ and $\dchoice{}{}$, so \m{\pchoice{\alpha;\beta}{\gamma} \mequiv \pchoice{(\alpha;\beta)}{\gamma}}.

Note that, quite unlike in the case of \m{\prepeat{\alpha}} and unlike in differential games \cite{DBLP:journals/corr/Platzer15:dGI}, it is irrelevant whether Angel decides the duration for \m{\pevolvein{\D{x}=\genDE{x}}{\ivr}} before or after that continuous evolution, because initial-value problems for \m{\D{x}=\genDE{x}} have unique solutions by Picard-Lindel\"off as term $\genDE{x}$ is smooth.

Observe that every (completed) play of a game is won or lost by exactly one player. 
Even a play of repeated game \m{\prepeat{\alpha}} has only one winner, because the game stops as soon as one player has won, e.g., because his opponent failed a test.
This is different than the repetition of whole game plays (including winning/losing), where the purpose is for the players to repeat the same game over and over again to completion, win and lose multiple times, and study who wins how often in the long run with mixed strategies.
A hybrid game is played once (even if some part of it constitutes in repeating action choices) and it stops as soon as either Angel or Demon have won.
In applications, the system is already in trouble even if it loses the game only once, because that may entail that a safety-critical property has already been violated.

\def\wallac{u}%
\def\eveac{g}%
\begin{example}[\textup{\sf WALL$\boldsymbol{\cdot}$E} and \textup{\sf EVE}] \label{ex:WE}
  Consider a game of the robots \W and \E moving on a (one-dimensional) planet.
  \begin{equation}
  \begin{aligned}
    (w-e)^2\leq1 \land v=f \limply\, &
    \big\langle\big( 
(\dchoice{\pupdate{\pumod{\wallac}{1}}}{\pupdate{\pumod{\wallac}{-1}}});\\
    & \phantom{\big\langle\big(}     (\pchoice{\pupdate{\pumod{\eveac}{1}}}{\pupdate{\pumod{\eveac}{-1}}});\\
    & \phantom{\big\langle\big(}
    \pupdate{\pumod{t}{0}};
    \devolvein{(\D{w}=v\syssep\D{v}=\wallac\syssep\D{e}=f\syssep\D{f}=\eveac\syssep\D{t}=1}{t\leq1)}
    \drepeat{\big)}
    \\&\big\rangle ~ (w-e)^2\leq1
  \end{aligned}
  \label{eq:EW}
  \end{equation}
  Robot \W is at position $w$ with velocity $v$ and acceleration $\wallac$ and plays the part of Demon.
  Robot \E is at $e$ with velocity $f$ and acceleration $\eveac$ and plays the part of Angel.
  The antecedent of \rref{eq:EW} before the implication assumes that \W and \E start close to one another (distance at most 1) and with identical velocities.
  The objective of \E, who plays Angel's part in \rref{eq:EW}, is to be close to \W (i.e.\ \((w-e)^2\leq1\)) as specified after the $\ddiamond{\cdot}{}$ modality in the succedent.
  The hybrid game proceeds as follows.
  Demon \W controls how often the hybrid game repeats by operator $\drepeat{}$.
  In each iteration, Demon \W first chooses (${\dchoice{}{}}$) to accelerate ($\pupdate{\pumod{\wallac}{1}}$) or brake ($\pupdate{\pumod{\wallac}{-1}}$), then Angel \E chooses ($\cup$) whether to accelerate ($\pupdate{\pumod{\eveac}{1}}$) or brake ($\pupdate{\pumod{\eveac}{-1}}$).
  Every time that the $\drepeat{}$ loop repeats, the players get to make that choice again. They are not bound by what they chose in the previous iterations.
  Yet, depending on the previous choices, the state will have evolved differently, which influences indirectly what moves a player needs to choose to win.
  After this sequence of choices of $\wallac$ and $\eveac$ by Demon and Angel, respectively, a clock variable $t$ is reset to $\pupdate{\pumod{t}{0}}$. 
  Then the game follows a differential equation system such that the time-derivative of \W's position $w$ is his velocity $v$ and the time-derivative of $v$ is acceleration $\wallac$, the time-derivative of \E's position $e$ is her velocity $f$ and the time-derivative of $f$ is acceleration $\eveac$.
  The time-derivative of clock variable $t$ is 1, yet the differential equation is restricted to the evolution domain $t\leq1$.
  Angel controls the duration of differential equations.
  Yet, this differential equation is within a dual game by operator $\pdual{}$, so Demon controls the duration of the continuous evolution.
  Here, both \W and \E evolve continuously but Demon \W decides how long.
  He cannot chose durations $>1$, because that would make him violate the evolution domain constraint $t\leq1$ and lose,
  so the players can change their control after at most one time unit, but Demon decides when exactly.
  Similar games can be studied for robot motion in higher dimensions using \dGL.
\end{example}

The two players in \rref{ex:WE} use various dualities to model a situation where Demon chooses discretely ($\dchoice{}{}$), then Angel chooses discretely ($\cup{}{}$), then Demon chooses the duration for the joint continuous evolution ($\devolvein{(\D{w}=\dots}{t\leq1)}$) and finally Demon decides about repetition ($\drepeat{}$), which gives 6 dual operators with nesting depth 4 (or 3 when discounting dual assignments).
Deeper nesting levels of hybrid game operators, which give rise to longer chains of $\pdual{}$ operators, can be used to describe hybrid games with more levels of interaction (any arbitrary number of nested $\pchoice{}{}, \pdual{}, \prepeat{}$).

\begin{example}[\textup{\sf WALL$\boldsymbol{\cdot}$E} and \textup{\sf EVE} and the world] \label{ex:WE-world}
The game in \rref{eq:EW} accurately reflects the situation when \W is in control of time since the (only) differential equation occurs within an odd number of $\pdual{}$ operators.
But to design the model \rref{eq:EW}, \E may have used a common modeling device to conservatively attribute the control of the differential equation to \W, even if time is really under control of a third player, the external environment.
\E's reason for this model would be that she is not in control of time, so there is no reason to believe why time would help her.
\E, thus, conservatively ceases control of time to Demon, which corresponds to assuming that the third player of the external environment is allowed to collaborate with \W to form an aggregate Demon player consisting of \W and the environment.

When \W wants to analyze his winning strategies with a $\dbox{\cdot}{}$ variation of \rref{eq:EW}, he could use the same modeling device to flip the differential equation over to Angel's control by removing the $\pdual{}$ to conservatively associate the environment to the opponent:
  \begin{equation}
  \begin{aligned}
    (w-e)^2\leq1 \land v=f \limply\, &
    \big[\big( 
(\dchoice{\pupdate{\pumod{\wallac}{1}}}{\pupdate{\pumod{\wallac}{-1}}});\\
    & \phantom{\big\langle\big(}     (\pchoice{\pupdate{\pumod{\eveac}{1}}}{\pupdate{\pumod{\eveac}{-1}}});\\
    & \phantom{\big\langle\big(}
    \pupdate{\pumod{t}{0}};
    \pevolvein{(\D{w}=v\syssep\D{v}=\wallac\syssep\D{e}=f\syssep\D{f}=\eveac\syssep\D{t}=1}{t\leq1)}
    \drepeat{\big)}
    \\&\big] ~ (w-e)^2>1
  \end{aligned}
  \label{eq:WE}
  \end{equation}
If, instead, \W keeps the same hybrid game as in \rref{eq:WE}, just with $\dbox{\cdot}{}$, the game considers the situation when \W has control over time, which would make some part of the game trivial, because, once he reached \((w-e)^2>1\), the Demon \W can then just always evolve for 0 time units.
Observe how a three-player game of \W, \E, and environment can be analyzed by combining the \dGL formulas \rref{eq:EW} and \rref{eq:WE} propositionally, which then analyze the same game from different perspectives of possible collaborations.
The \dGL formula expressing that neither \rref{eq:EW} nor \rref{eq:WE} is true, for example, is true in exactly the states where \W and \E draw, because the disturbance of the external environment can choose the winner by helping either \W or \E.
\end{example}

When the role of such environments is not limited to influencing only differential equation durations, the appropriate placement of $\pdual{}$ operators around additional choices ($\cup$) or repetitions ($\prepeat{}$) or multiple differential equations gives models more flexibility.
Similar phenomena happen when further players are added to the game as long as their role can be described by logical combinations of \dGL formulas investigating the game from a set of aggregated two-player perspectives.
Surprisingly, there is virtually no limit to how far logic around such essentially two-player hybrid games extends to seemingly substantially more general situations by appropriate modeling (\rref{sec:dGL-dL-expressiveness}).

\subsection{Semantics} \label{sec:dGL-semantics}

The logic \dGL has a denotational semantics.
The \dGL semantics defines, for each formula $\phi$, the set $\imodel{\I}{\phi}$ of states in which $\phi$ is true.
For each hybrid game $\alpha$ and each set of winning states $X$, the \dGL semantics defines the set $\strategyfor[\alpha]{X}$ of states from which Angel has a winning strategy to achieve $X$ in hybrid game $\alpha$, as well as the set $\dstrategyfor[\alpha]{X}$ of states from which Demon has a winning strategy to achieve $X$ in $\alpha$.

A \emph{state} $\iget[state]{\I}$ is a mapping from variables to $\reals$.
An \emph{interpretation} $\iget[const]{\I}$ assigns a relation \m{\iget[const]{\I}(p) \subseteq \reals^k} to each predicate symbol $p$ of arity $k$.
The interpretation further determines the set of states $\linterpretations{\Sigma}{V}$, which is isomorphic to a Euclidean space $\reals^n$ when $n$ is the number of relevant variables.
For a subset \m{X\subseteq\linterpretations{\Sigma}{V}} the complement \m{\linterpretations{\Sigma}{V} \setminus X} is denoted $\scomplement{X}$.
Let
\m{\iget[state]{\imodif[state]{\I}{x}{r}}} denote the state that agrees with state~$\iget[state]{\I}$ except for the interpretation of variable~\m{x}, which is changed to~\m{r \in \reals}.
The value of term $\theta$ in state $\iget[state]{\I}$ is denoted by \m{\ivaluation{\I}{\theta}}.
The denotational semantics of \dGL formulas will be defined in \rref{def:dGL-semantics} by simultaneous induction along with the denotational semantics, $\strategyfor[\alpha]{\cdot}$ and $\dstrategyfor[\alpha]{\cdot}$, of hybrid games, defined in \rref{def:HG-semantics}, because \dGL formulas are defined by simultaneous induction with hybrid games.

\begin{definition}[\dGL semantics] \label{def:dGL-semantics}
The \emph{semantics of a \dGL formula} $\phi$ for each interpretation $\iget[const]{\I}$ with a corresponding set of states $\linterpretations{\Sigma}{V}$ is the subset \m{\imodel{\I}{\phi}\subseteq\linterpretations{\Sigma}{V}} of states in which $\phi$ is true.
It is defined inductively as follows
\begin{enumerate}
\item \(\imodel{\I}{p(\theta_1,\dots,\theta_k)} = \{\iportray{\I} \in \linterpretations{\Sigma}{V} \with (\ivaluation{\I}{\theta_1},\dots,\ivaluation{\I}{\theta_k})\in\iget[const]{\I}(p)\}\)
\item \(\imodel{\I}{\theta_1\geq\theta_2} = \{\iportray{\I} \in \linterpretations{\Sigma}{V} \with \ivaluation{\I}{\theta_1}\geq\ivaluation{\I}{\theta_2}\}\)
\item \(\imodel{\I}{\lnot\phi} = \scomplement{(\imodel{\I}{\phi})}\)
\item \(\imodel{\I}{\phi\land\psi} = \imodel{\I}{\phi} \cap \imodel{\I}{\psi}\)
\item
{\def\Im{\imodif[state]{\I}{x}{r}}%
\(\imodel{\I}{\lexists{x}{\phi}} =  \{\iportray{\I} \in \linterpretations{\Sigma}{V} \with \iget[state]{\Im} \in \imodel{\I}{\phi} ~\text{for some}~r\in\reals\}\)
}
\item \(\imodel{\I}{\ddiamond{\alpha}{\phi}} = \strategyfor[\alpha]{\imodel{\I}{\phi}}\)
\item \(\imodel{\I}{\dbox{\alpha}{\phi}} = \dstrategyfor[\alpha]{\imodel{\I}{\phi}}\)
\end{enumerate}
A \dGL formula $\phi$ is \emph{valid in $\iget[const]{\I}$}, written \m{\iget[const]{\I}\models{\phi}}, iff it is true in all states, i.e.\ \m{\imodel{\I}{\phi}=\linterpretations{\Sigma}{V}}.
Formula $\phi$ is \emph{valid}, \m{\entails\phi}, iff \m{\iget[const]{\I}\models{\phi}} for all interpretations $\iget[const]{\I}$.
\end{definition}

\begin{definition}[Semantics of hybrid games] \label{def:HG-semantics}
The \emph{semantics of a hybrid game} $\alpha$ is a function \m{\strategyfor[\alpha]{\cdot}} that, for each interpretation $\iget[const]{\I}$ and each set of Angel's winning states \m{X\subseteq\linterpretations{\Sigma}{V}}, gives the \emph{winning region}, i.e.\ the set of states \m{\strategyfor[\alpha]{X}} from which Angel has a winning strategy to achieve $X$ in $\alpha$ (whatever strategy Demon chooses). It is defined inductively as follows\footnote{
The semantics of a hybrid game is not merely a reachability relation between states as for hybrid systems \cite{DBLP:conf/lics/Platzer12b}, because the adversarial dynamic interactions and nested choices of the players have to be taken into account.}
\begin{enumerate}
\item \(\strategyfor[\pupdate{\pumod{x}{\theta}}]{X} = \{\iportray{\I} \in \linterpretations{\Sigma}{V} \with \modif{\iget[state]{\I}}{x}{\ivaluation{\I}{\theta}} \in X\}\)
\item \(\strategyfor[\pevolvein{\D{x}=\genDE{x}}{\ivr}]{X} = \{\varphi(0) \in \linterpretations{\Sigma}{V} \with 
      \varphi(r)\in X\)
      for some $r\in\reals_{\geq0}$ and (differentiable)
      \m{\varphi:[0,r]\to\linterpretations{\Sigma}{V}}
      such that
      \(\varphi(\zeta)\in\imodel{\I}{\ivr}\)
      and
      \m{\D[t]{\,\varphi(t)(x)} (\zeta) =       %
      \ivaluation{\iconcat[state=\varphi(\zeta)]{\I}}{\theta}}
      for all \(0\leq\zeta\leq r\}\)
\item \(\strategyfor[\ptest{\ivr}]{X} = \imodel{\I}{\ivr}\cap X\)
\item \(\strategyfor[\pchoice{\alpha}{\beta}]{X} = \strategyfor[\alpha]{X}\cup\strategyfor[\beta]{X}\)
\item \(\strategyfor[\alpha;\beta]{X} = \strategyfor[\alpha]{\strategyfor[\beta]{X}}\)
\item \(\strategyfor[\prepeat{\alpha}]{X} = \capfold\{Z\subseteq\linterpretations{\Sigma}{V} \with X\cup\strategyfor[\alpha]{Z}\subseteq Z\}\)

\item \(\strategyfor[\pdual{\alpha}]{X} = \scomplement{(\strategyfor[\alpha]{\scomplement{X}})}\)
\end{enumerate}
The \emph{winning region} of Demon, i.e.\ the set of states \m{\dstrategyfor[\alpha]{X}} from which Demon has a winning strategy to achieve $X$ in $\alpha$ (whatever strategy Angel chooses) is defined inductively as follows
\begin{enumerate}
\item \(\dstrategyfor[\pupdate{\pumod{x}{\theta}}]{X} = \{\iportray{\I} \in \linterpretations{\Sigma}{V} \with \modif{\iget[state]{\I}}{x}{\ivaluation{\I}{\theta}} \in X\}\)
\item \(\dstrategyfor[\pevolvein{\D{x}=\genDE{x}}{\ivr}]{X} = \{\varphi(0) \in \linterpretations{\Sigma}{V} \with 
      \varphi(r)\in X\)
      for all $r\in\reals_{\geq0}$ and (differentiable)
      \m{\varphi:[0,r]\to\linterpretations{\Sigma}{V}}
      such that
      \(\varphi(\zeta)\in\imodel{\I}{\ivr}\)
      and
      \m{\D[t]{\,\varphi(t)(x)} (\zeta) =       %
      \ivaluation{\iconcat[state=\varphi(\zeta)]{\I}}{\theta}}
      for all $0\leq\zeta\leq r\}$
\item \(\dstrategyfor[\ptest{\ivr}]{X} = \scomplement{(\imodel{\I}{\ivr})}\cup X\)
\item \(\dstrategyfor[\pchoice{\alpha}{\beta}]{X} = \dstrategyfor[\alpha]{X}\cap\dstrategyfor[\beta]{X}\)
\item \(\dstrategyfor[\alpha;\beta]{X} = \dstrategyfor[\alpha]{\dstrategyfor[\beta]{X}}\)
\item \(\dstrategyfor[\prepeat{\alpha}]{X} = \cupfold\{Z\subseteq\linterpretations{\Sigma}{V} \with Z\subseteq X\cap\dstrategyfor[\alpha]{Z}\}\)
\item \(\dstrategyfor[\pdual{\alpha}]{X} = \scomplement{(\dstrategyfor[\alpha]{\scomplement{X}})}\)
\end{enumerate}
\end{definition}
The notation uses $\strategyfor[\alpha]{X}$ and $\dstrategyfor[\alpha]{X}$ instead of $\varsigma^{\iget[const]{\I}}_\alpha{(X)}$ and $\delta^{\iget[const]{\I}}_\alpha{(X)}$, because the interpretation $\iget[const]{\I}$ that gives a semantics to predicate symbols in tests and evolution domains is clear from the context.
Strategies do not occur explicitly in the \dGL semantics, because it is based on winning regions, i.e.\ the existence of winning strategies, not on the strategies themselves.
The winning regions for Angel are illustrated in \rref{fig:HG-semantics-illustration}.
The winning region \(\strategyfor[\prepeat{\alpha}]{X}\) is the smallest set $Z$ around $X$ that already contains all states from which one more $\strategyfor[\alpha]{\argholder}$ could win into $Z$, so \(\strategyfor[\alpha]{\strategyfor[\prepeat{\alpha}]{X}} \setminus \strategyfor[\prepeat{\alpha}]{X} = \emptyset\).
It includes all $\kappa$-fold iterations $\inflopstrat[\kappa][\alpha]{X}$ of the winning region construction $\strategyfor[\alpha]{\argholder}$ (\rref{sec:ClosureOrdinals}) for all $\kappa$.

\begin{figure}[tbh]
  \tikzstyle{strategy preimage}=[draw=vblue,fill=vblue!20,shape=ellipse]
  \tikzstyle{winning condition}=[draw=vred,fill=vred!50,text=white,shape=circle]
  \tikzstyle{winning annotation}=[ultra thick,black,->]
\begin{tikzpicture}[thick]
  \begin{scope}[xshift=-2.2cm,yshift=0cm]%
    \node[strategy preimage,shape=circle] (update) at (-1.5,0) {\phantom{$X$}};
    \node[winning condition] at (0,0) {$X$};
    \node (jump) at (-2.0,1.5) {$\strategyfor[\pupdate{\pumod{x}{\theta}}]{X}$};
    \draw[winning annotation] (jump) -- (update.center);
  \end{scope}
  \begin{scope}[xshift=3.5cm,yshift=0cm]%
    \node[strategy preimage,minimum width=4.8cm,minimum height=2.1cm] at (-1.2,0) {};
    \node[winning condition] at (0,0) {$X$};
    \draw[*-*] (-2.9,-0.3) .. controls (-2,-1.2) and (-1.2,1.5) .. (-0.1,0.1) 
    node[above,pos=0.55,sloped] {$\D{x}=\genDE{x}$};
    \node (evolve) at (-3.5,1.5) {$\strategyfor[\pevolve{\D{x}=\genDE{x}}]{X}$};
    \draw[winning annotation] (evolve) -- (-2.8,0.3);
  \end{scope}
  \begin{scope}[xshift=7.5cm,yshift=0cm]%
    \node[winning condition,minimum width=1.5cm] at (0,0) {$X$};
    \node[draw=vgreen,fill=vgreen!20,shape=ellipse,minimum height=3cm,minimum width=1.5cm,opacity=0.5] at (-0.9,0) {};
    \node at (-0.9,1) {$\imodel{\I}{\phi}$};
    \node (test) at (0.3,1.5) {$\strategyfor[\ptest{\phi}]{X}$};
    \draw[winning annotation] (test) -- (-0.5,0.2);
  \end{scope}
  \begin{scope}[xshift=0cm,yshift=-3.4cm]%
    \begin{scope}[rotate=-30]
    \node[strategy preimage,fill=vblue!30,opacity=0.5,minimum width=3.5cm,minimum height=1.7cm,rotate=-30] at (-0.8,0) {};
    \node[rotate=-30] (alpha) at (-2,0) {$\strategyfor[\alpha]{X}$};
    \end{scope}
    \begin{scope}[rotate=30]
    \node[strategy preimage,fill=vblue!30,opacity=0.5,minimum width=3.5cm,minimum height=1.7cm,rotate=30] at (-0.8,0) {};
    \node[rotate=30] (beta) at (-2,0) {$\strategyfor[\beta]{X}$};
    \end{scope}
    \node[winning condition] at (0,0) {$X$};
    \node (choice) at (-4,0) {$\strategyfor[\pchoice{\alpha}{\beta}]{X}$};
    \draw[winning annotation] (choice.east) -- (alpha.south);
    \draw[winning annotation] (choice.east) -- (beta.north);
  \end{scope}
  \begin{scope}[xshift=6cm,yshift=-3.4cm]%
    \node[strategy preimage,fill=vblue!11,minimum width=4.8cm,minimum height=2.1cm] at (-1.2,0) {};
    \node (alphabeta) at (-2.7,0) {$\strategyfor[\alpha]{\strategyfor[\beta]{X}}$};
    \node[strategy preimage,fill=vblue!30,minimum width=2.2cm,minimum height=1.3cm] at (-0.4,0) {};
    \node at (-0.9,0) {$\strategyfor[\beta]{X}$};
    \node[winning condition] at (0,0) {$X$};
    \node (compose) at (-3.5,1.6) {$\strategyfor[\alpha;\beta]{X}$};
    \draw[winning annotation] (compose) -- (alphabeta);
  \end{scope}
  \begin{scope}[xshift=1cm,yshift=-8cm]%
    \node[strategy preimage,fill=vblue!2,shape=rectangle,rounded corners=8pt,minimum height=0.9cm] (nomore) at (-2.4,1.8) {$\strategyfor[\alpha]{\strategyfor[\prepeat{\alpha}]{X}} \setminus \strategyfor[\prepeat{\alpha}]{X}$\strut};
    \node (empty) at (0.4,2.1) {$\emptyset$};
    \draw[winning annotation] (empty) -- (nomore);
    \node[strategy preimage,fill=vblue!4,minimum width=7.8cm,minimum height=3.2cm] at (-1.8,0) {};
    \node[anchor=west] (infinite) at (-5.7,0) {$\inflopstrat[\infty][\alpha]{X} ~\cdots$};
    \node[strategy preimage,fill=vblue!11,minimum width=4.8cm,minimum height=2.1cm] at (-1.2,0) {};
    \node at (-3.1,0) {$\inflopstrat[3][\alpha]{X}$};
    \node[strategy preimage,minimum width=3.5cm,minimum height=1.7cm] at (-0.8,0) {};
    \node at (-2,0) {$\inflopstrat[2][\alpha]{X}$};
    \node[strategy preimage,fill=vblue!30,minimum width=2.2cm,minimum height=1.3cm] at (-0.4,0) {};
    \node at (-0.9,0) {$\inflopstrat[][\alpha]{X}$};
    \node[winning condition] at (0,0) {$X$};
    \node (repeat) at (-5.5,1.5) {$\strategyfor[\prepeat{\alpha}]{X}$};
    \draw[winning annotation] (repeat) -- (infinite);
  \end{scope}
  \begin{scope}[xshift=6cm,yshift=-8cm]%
    \node[strategy preimage,fill=vgreen!10,minimum width=5cm,minimum height=3.5cm,shape=rectangle] at (0,0) {};
    \node at (-2.1,1.5) {$\scomplement{X}$};
      \node[winning condition,shape=circle,minimum width=3cm] at (0,0) {};
    \node[winning condition,draw=none,fill=none] at (0,1) {$X$};
      \node at (0,-0.9) {$\strategyfor[\alpha]{\scomplement{X}}$};
      \node[strategy preimage,shape=ellipse,minimum width=2.5cm,minimum height=1.2cm] at (-0.2,0.1) {$\scomplement{\strategyfor[\alpha]{\scomplement{X}}}$};
    \node (dual) at (2.2,2.5) {$\strategyfor[\pdual{\alpha}]{X}$};
    \draw[winning annotation] (dual) -- (0.5,0);
  \end{scope}
\end{tikzpicture}
  \caption{Illustration of denotational semantics of hybrid games as winning regions}
  \label{fig:HG-semantics-illustration}
\end{figure}
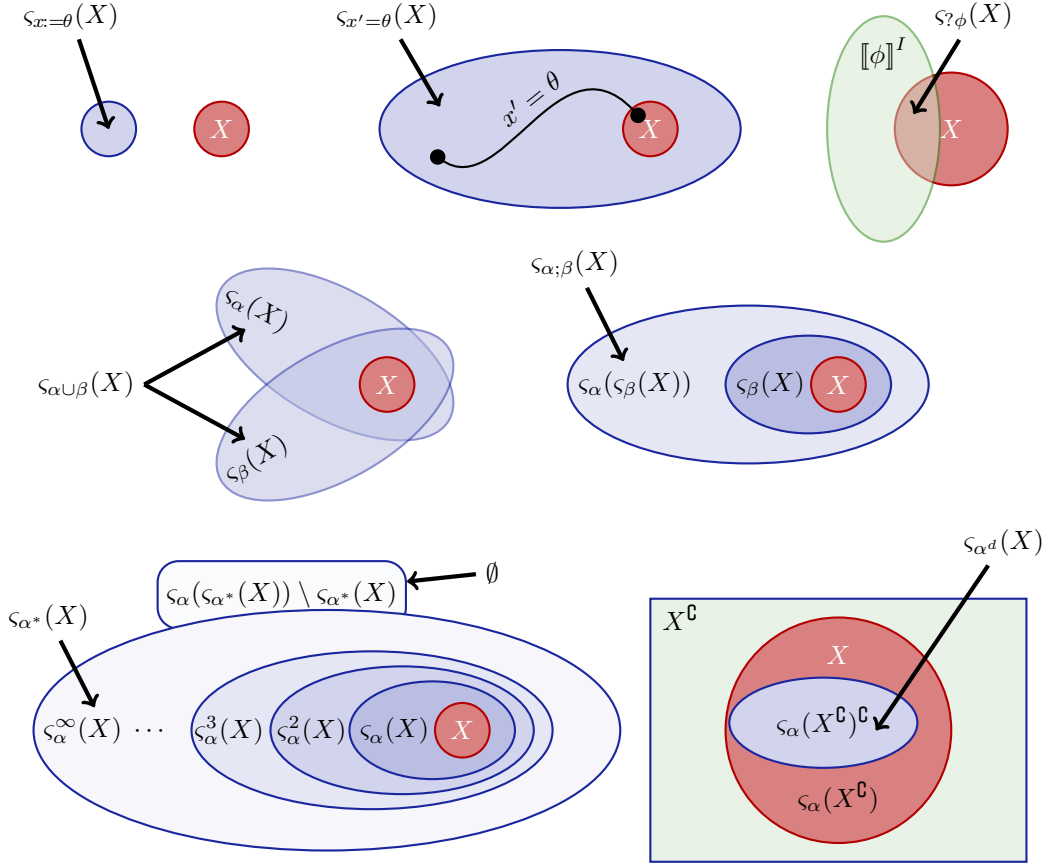

The semantics is \emph{compositional}, i.e.\ the semantics of a compound \dGL formula is a simple function of the semantics of its pieces, and the semantics of a compound hybrid game is a function of the semantics of its pieces.
This guarantees referential transparency and enables a compositional proof calculus.
Furthermore, existence of a strategy in hybrid game $\alpha$ to achieve $X$ is independent of games and \dGL formulas surrounding $\alpha$, but just depends on the remaining game $\alpha$ itself and the goal $X$.
By a simple inductive argument, this shows that one can focus on memoryless strategies, because the existence of strategies does not depend on the context, hence, by working bottom up, the strategy itself cannot depend on past states and choices, only the current state, remaining game, and goal.
This follows from a generalization of a classical result \cite{Zermelo13}, but is directly apparent in a logical setting.
Furthermore, the semantics is monotone, i.e.\ larger sets of winning states induce larger winning regions.

\begin{lemma}[Monotonicity] \label{lem:monotone}%
  The semantics is \emph{monotone}, i.e.\ \m{\strategyfor[\alpha]{X}\subseteq\strategyfor[\alpha]{Y}} and \m{\dstrategyfor[\alpha]{X}\subseteq\dstrategyfor[\alpha]{Y}} for all \m{X\subseteq Y}.
\end{lemma}
\proofmove
A simple check based on the observation that $X$ only occurs with an even number of negations in the semantics.
For example,
\(\strategyfor[\prepeat{\alpha}]{X} = \capfold\{Z\subseteq\linterpretations{\Sigma}{V} \with X\cup\strategyfor[\alpha]{Z}\subseteq Z\} \subseteq \capfold\{Z\subseteq\linterpretations{\Sigma}{V} \with Y\cup\strategyfor[\alpha]{Z}\subseteq Z\} = \strategyfor[\prepeat{\alpha}]{Y}\) if \(X\subseteq Y\).
Likewise, \(X\subseteq Y\) implies \(\scomplement{X}\supseteq\scomplement{Y}\), hence
\(\strategyfor[\alpha]{\scomplement{X}}\supseteq \strategyfor[\alpha]{\scomplement{Y}}\), so
\(\strategyfor[\pdual{\alpha}]{X} = \scomplement{(\strategyfor[\alpha]{\scomplement{X}})} \subseteq \scomplement{(\strategyfor[\alpha]{\scomplement{Y}})} = \strategyfor[\pdual{\alpha}]{Y}\).
\endproofmove
Monotonicity implies that the least fixpoint in \m{\strategyfor[\prepeat{\alpha}]{X}} and the greatest fixpoint in \m{\dstrategyfor[\prepeat{\alpha}]{X}} are well-defined \cite[Lemma 1.7]{Harel_et_al_2000}.
The semantics of $\strategyfor[\prepeat{\alpha}]{X}$ is a least fixpoint, which results in a well-founded repetition of $\alpha$, i.e.\ Angel can repeat any number of times but she ultimately needs to stop at a state in $X$ in order to win.
In particular, Angel cannot play a Zeno strategy with infinitely many steps in finite time.
The semantics of $\dstrategyfor[\prepeat{\alpha}]{X}$ is a greatest fixpoint, instead, for which Demon needs to achieve a state in $X$ after every number of repetitions, because Angel could choose to stop at any time, but Demon still wins if he only postpones $\scomplement{X}$ forever, because Angel ultimately has to stop repeating.
Thus, for the formula $\ddiamond{\prepeat{\alpha}}{\phi}$, Demon already has a winning strategy if he only has a strategy that is not losing by preventing $\phi$ indefinitely, because Angel eventually has to stop repeating anyhow and will then end up in a state not satisfying $\phi$, which makes her lose.
The situation for $\dbox{\prepeat{\alpha}}{\phi}$ is dual.

Hybrid games branch finitely when the players decide which game to play in \m{\pchoice{\alpha}{\beta}} and \m{\dchoice{\alpha}{\beta}}, respectively.
The games \m{\prepeat{\alpha}} and \m{\drepeat{\alpha}} also branch finitely, because, after each repetition of $\alpha$, the respective player (Angel for  \m{\prepeat{\alpha}} and Demon for \m{\drepeat{\alpha}}) may decide whether to repeat again or stop.
Repeated games still lead to infinitely many branches, because a repeated game can be repeated any arbitrary number of times.
The game branches uncountably infinitely, however, when the players decide how long to evolve along differential equations in \m{\pevolvein{\D{x}=\genDE{x}}{\ivr}} and \m{\pdual{(\pevolvein{\D{x}=\genDE{x}}{\ivr})}}, because uncountably many nonnegative real number could be chosen as a duration (unless the system leaves $\ivr$ immediately).
These choices can be made explicit by relating the simple denotational modal semantics of \dGL to an equivalent operational game semantics that is technically much more involved but directly exposes the interactive intuition of game play.
For reference, this approach has been made explicit in \rref{app:operational-HG-semantics}.

\begin{example} \label{ex:ex-explosive-clock}
The following simple \dGL formula
\begin{equation}
\ddiamond{\prepeat{(\pchoice{\pupdate{\pumod{x}{x+1}}; \devolve{(\D{x}=x^2)}}{\pupdate{\pumod{x}{x-1}}})}}{\, (0\leq x<1)}
\label{eq:ex-explosive-clock}
\end{equation}
is true in all states from which there is a winning strategy for Angel to reach [0,1).
It is Angel's choice whether to repeat ($\prepeat{}$) and, every time she does, it is her choice ($\cup$) whether to increase $x$ by 1 and then (after $;$) give Demon control over the duration of the differential equation \m{\pevolve{\D{x}=x^2}} (left game) or whether to instead decrease $x$ by 1 (right game).
Formula \rref{eq:ex-explosive-clock} is valid, because Angel has the winning strategy of choosing the left action \({\pupdate{\pumod{x}{x+1}}; \devolve{(\D{x}=x^2)}}\) until $x\geq0$ followed by the right action \({\pupdate{\pumod{x}{x-1}}}\) until $0\leq x<1$.
By choosing the left action, $x\geq0$ will ultimately happen, because \m{\pupdate{\pumod{x}{x+1}}} increases $x$ by 1 and \m{\pevolve{\D{x}=x^2}} can only make $x$ bigger, because the derivative $x^2$ is nonnegative.
Once $x\geq0$, choosing the right action suitably often will reach the postcondition $0\leq x<1$ to allow Angel to win the game.
Note that Angel also wins immediately with the left action from $x=-1$, since the differential equation is stuck at $x=0$.
The following minor variation, however, is not valid:
\begin{equation*}
\hspace*{-0.3cm}
\ddiamond{\prepeat{(\pchoice{\pupdate{\pumod{x}{x+1}}; \devolve{(\D{x}=x^2)}}{(\dchoice{\pupdate{\pumod{x}{x-1}}}{\pupdate{\pumod{x}{x-2}}})})}}{(0\leq x<1)}
\label{eq:ex-explosive-clock2}
\end{equation*}
because Demon can spoil Angel's efforts by choosing \m{\pupdate{\pumod{x}{x-2}}} in his choice ($\dchoice{}{}$) to make $x$ negative whenever $1\leq x<2$, and then increasing $x$ to 1.5 again via \m{\devolve{(\D{x}=x^2)}} when Angel takes the left choice. Angel will never reach $0\leq x<1$ that way unless this was true initially already.
This phenomenon is examined in \rref{sec:Determinacy} in more detail.
\end{example}

\begin{example}[\textup{\sf WALL$\boldsymbol{\cdot}$E} and \textup{\sf EVE}] \label{ex:WE2}
  The \dGL formula \rref{eq:EW} from \rref{ex:WE} is valid, because Angel \E indeed has a winning strategy to get close to \W by mimicking Demon's choices.
  Recall that Demon \W controls the repetition $\drepeat{}$, so the fact that the hybrid game starts \E off close to \W is not sufficient for \E to win the game.
  The hybrid game in \rref{eq:EW} would be trivial if Angel were to control the repetition (because she would then win just by choosing not to repeat) or were to control the differential equation (because she would then win by always  evolving for duration 0).
  The analysis of \rref{eq:EW} is more difficult if the first two lines in the hybrid game are swapped so that Angel \E chooses $\eveac$ before Demon \W chooses $\wallac$.
\end{example}

\section{Meta-Properties} \label{sec:MetaProperties}

This section analyzes meta-properties and semantical properties of the hybrid games of \dGL, including determinacy of hybrid games, hybrid game equivalences such as reduction of evolution domains, and closure ordinals of hybrid games.

\subsection{Determinacy} \label{sec:Determinacy}

Every particular game play in a hybrid game is won by exactly one player, because hybrid games are zero-sum and there are no draws.
That alone does not imply determinacy, i.e.\ that, from all initial situations, either one of the players always has a winning strategy to force a win, regardless of how the other player chooses to play.

In order to understand the importance of determinacy for classical logics, consider the semantics of repetition, defined as a \emph{least} fixpoint, which is crucial because that gives a well-founded repetition.
Otherwise, the \emph{filibuster formula} would not have a well-defined truth-value:
\begin{equation}
\ddiamond{\prepeat{(\dchoice{\pumod{x}{0}}{\pumod{x}{1}})}}{x=0}
\label{eq:filibuster}
\end{equation}
It is Angel's choice whether to repeat ($\prepeat{}$), but every time Angel repeats, it is Demon's choice ($\cap$) whether to play \m{\pumod{x}{0}} or \m{\pumod{x}{1}}.
The game in this formula never deadlocks, because every player always has a remaining move (here even two).
But, without the least fixpoint, the game would have perpetual checks, because no strategy helps either player win the game but just prevents the other player from winning; see \rref{fig:nondetermined}.%

\begin{figure}[bth]
  \centering
  \begin{tikzpicture}[grow'=down]
    \tikzstyle{level 2}=[sibling distance=+22mm]
    \tikzstyle{level 3}=[sibling distance=+12mm]
    \tikzstyle{level 4}=[sibling distance=+7mm]
    \tikzstyle{level 5}=[sibling distance=+6mm]
    \node[diamond] (PO) {X}
      child[pdia] {node[box] {X}
        child[pbox] {node[diamond] (1) {1}
          child[pdia] {node[box] {1}
            child[pbox] {node[diamond] (back1) {1}}
            child[pbox] {node[diamond] (0over) {0}
              child[pdia] {node[box] {0}
                child[pbox] {node[diamond] (1again) {1}}
                child[pbox] {node[diamond,diawon] (0again) {0}}
                edge from parent node[action] {repeat}
              }
            child[pdia] {node[box,diawon] {0} edge from parent node[action] {stop}}
          }
            edge from parent node[action] {repeat}
          }
          child[pdia] {node[diamond,boxwon] {1} edge from parent node[action] {stop}}
        }
        child[pbox] {node[diamond] (0) {0}
          child[pdia] {node[box] {0}
            child[pbox] {node[diamond] (1over) {1}}
            child[pbox] {node[diamond,diawon] (back0) {0}}
            edge from parent node[action] {repeat}
          }
          child[pdia] {node[diamond,diawon] {0} edge from parent node[action] {stop}}
        }
        edge from parent node[action] {repeat}
      }
      child[pdia] {node[diamond] {X}
        edge from parent node[action] {stop}
      }
      ;
      \draw[backedge] (back1) to[bend right=60] (1);
      \draw[backedge] (back0) to[bend left=90] (0);
      \draw[backedge] (0again) to[bend left=70] (0over);
      \draw[backedge] (1again) to[bend right=90] (1);
      \draw[backedge] (1over) to[bend right=30] (1);
      \draw[boxastrategy] (PO-1) -- (1);
      \draw[diamondastrategy] (1) -- (1-1);
      \draw[boxastrategy] (1-1) -- (back1);
  \end{tikzpicture}
  \caption{The filibuster game formula \(\ddiamond{\prepeat{(\dchoice{\pumod{x}{0}}{\pumod{x}{1}})}}{x=0}\) is false (unless $x=0$ initially), but would be non-determined without least fixpoints (strategies follow thick actions).
  Angel's action choices are illustrated by dashed edges from dashed diamonds, Demon's action choices by solid edges from solid squares, and double lines indicate bisimilar states with the same continuous state and a subgame of the same structure of subsequent choices.
  States where Angel wins are marked \usebox{\tmpdiawon} and states where Demon wins by \usebox{\tmpboxwon}.
}
  \label{fig:nondetermined}
\end{figure}
Demon can move \m{\pumod{x}{1}} and would win, but Angel observes this and decides to repeat, so Demon can again move \m{\pumod{x}{1}}.
Thus (unless Angel is lucky starting from an initial state where she has won already) every strategy that one player has to reach $x=0$ or $x=1$ could be spoiled by the other player so the game would not be determined, i.e.\ no player has a winning strategy.
Every player can let his opponent win, but would not have a strategy to win himself.
Because of the least fixpoint \(\strategyfor[\prepeat{\alpha}]{\argholder}\) in the semantics, however, repetitions are well-founded and, thus, have to stop eventually (after an arbitrary unbounded number of rounds).
Hence, in the example in \rref{fig:nondetermined}, Demon still wins and formula \rref{eq:filibuster} is $\lfalse$, unless $x=0$ holds initially (the unknown initial value is marked $X$ in \rref{fig:nondetermined}).
In other words, the formula in \rref{eq:filibuster} is equivalent to $x=0$.
The same phenomenon happens in \rref{ex:ex-explosive-clock}.
Likewise, the dual filibuster game formula 
\begin{equation}
  x=0\limply\ddiamond{\drepeat{(\pchoice{\pumod{x}{0}}{\pumod{x}{1}})}}{x=0}
  \label{eq:dual-filibuster}
\end{equation}
is (determined and) valid, because Demon has to stop repeating $\drepeat{}$ eventually so that Angel wins if she patiently plays \m{\pumod{x}{0}} each time.
Similarly, the game in the following hybrid filibuster formula would not be determined without the least fixpoint semantics
\[\ddiamond{\prepeat{(\pupdate{\pumod{x}{0}}; \devolve{\D{x}=1})}}{x=0}\]
because Demon could always evolve continuously to some state where $x>0$ and Angel would never want to stop.
Since Angel will have to stop eventually, she loses and the formula is $\lfalse$ unless $x=0$ holds initially.

It is also important that Angel can only choose real durations $r\in\reals_{\geq0}$ for a continuous evolution game \m{\pevolvein{\D{x}=\genDE{x}}{\ivr}}, not infinity $\infty$, so she  ultimately stops.
Otherwise
\begin{equation}
  \ddiamond{\prepeat{(\devolve{\D{x}=1};\pupdate{\pumod{x}{0}})}}{x=0}
  \label{eq:continuous-filibuster}
\end{equation}
would not be determined, because Angel wants to repeat (unless $x=0$ initially) and $\pupdate{\pumod{x}{0}}$ will make her win once she stops after a nonzero number of repetitions.
Yet, if Demon could choose $\infty$ as the duration for the continuous evolution game $\devolve{\D{x}=1}$, Angel would never get to play the subsequent $\pupdate{\pumod{x}{0}}$ to win.
Since durations need to be real numbers, though, each continuous evolution ultimately has to stop, so the formula in \rref{eq:continuous-filibuster} is valid even if Demon can postpone Angel's victory arbitrarily long.

In order to make sure that \dGL is a classical two-valued modal logic, hybrid games have no draws.
But, because modalities refer to the existence of winning strategies, they only receive classical truth values if, from each state, one of the players has a winning strategy for complementary winning conditions of a hybrid game $\alpha$.
The logical setup of \dGL makes this determinacy proof very simple, without the need to use, e.g., the deep \cite{DBLP:journals/annmathlog/Friedman71} Borel determinacy theorem for winning conditions that are Borel in the product topology induced on game trees by the discrete topology of actions \cite{DBLP:journals/mathann/Martin75}, which does not even fit to the structure of arbitrarily nested inductive and coinductive fixpoints of the winning region semantics of \dGL.
\begin{theorem}[Consistency \& determinacy] \label{thm:dGL-determined}%
  Hybrid games are consistent and determined, i.e.\ \m{\entails \lnot\ddiamond{\alpha}{\lnot\phi} \lbisubjunct \dbox{\alpha}{\phi}}.
\end{theorem}
\proofmove
The proof shows by a straightforward induction on the structure of $\alpha$ that
\m{\scomplement{\strategyfor[\alpha]{\scomplement{X}}} = \dstrategyfor[\alpha]{X}} for all \m{X\subseteq\linterpretations{\Sigma}{V}} and all $\iget[const]{\I}$ with some set of states \m{\linterpretations{\Sigma}{V}}, which implies the validity of \m{\lnot\ddiamond{\alpha}{\lnot\phi} \lbisubjunct \dbox{\alpha}{\phi}} using \m{X\mdefeq\imodel{\I}{\phi}}.
\begin{enumerate}
\item \(\scomplement{\strategyfor[\pupdate{\pumod{x}{\theta}}]{\scomplement{X}}}
= \scomplement{\{\iportray{\I} \in \linterpretations{\Sigma}{V} \with \modif{\iget[state]{\I}}{x}{\ivaluation{\I}{\theta}} \not\in X\}}
= \strategyfor[\pupdate{\pumod{x}{\theta}}]{X}
= \dstrategyfor[\pupdate{\pumod{x}{\theta}}]{X}\)

\item \(\scomplement{\strategyfor[\pevolvein{\D{x}=\genDE{x}}{\ivr}]{\scomplement{X}}}
= \scomplement{\{\varphi(0) \in \linterpretations{\Sigma}{V} \with 
      \varphi(r)\not\in X\)
      for some $0\leq r\in\reals$ and some (differentiable)
      \m{\varphi:[0,r]\to\linterpretations{\Sigma}{V}}
      such that
      \m{\D[t]{\,\varphi(t)(x)} (\zeta) =       %
      \ivaluation{\iconcat[state=\varphi(\zeta)]{\I}}{\theta}}
      and
      \m{\varphi(\zeta)\in\imodel{\I}{\ivr}}
      for all $0\leq\zeta\leq r\}}$
\(= \dstrategyfor[\pevolvein{\D{x}=\genDE{x}}{\ivr}]{X}\),
because the set of states from which there is no winning strategy for Angel to reach a state in $\scomplement{X}$ prior to leaving $\imodel{\I}{\ivr}$ along \m{\pevolvein{\D{x}=\genDE{x}}{\ivr}} is exactly the set of states from which \m{\pevolvein{\D{x}=\genDE{x}}{\ivr}}  always stays in $X$ (until leaving $\imodel{\I}{\ivr}$ in case that ever happens).
\item \(\scomplement{\strategyfor[\ptest{\ivr}]{\scomplement{X}}}
= \scomplement{(\imodel{\I}{\ivr}\cap\scomplement{X})}
= \scomplement{(\imodel{\I}{\ivr})}\cup \scomplement{(\scomplement{X})}
= \dstrategyfor[\ptest{\ivr}]{X}\)

\item \(\scomplement{\strategyfor[\pchoice{\alpha}{\beta}]{\scomplement{X}}}
= \scomplement{(\strategyfor[\alpha]{\scomplement{X}} \cup \strategyfor[\beta]{\scomplement{X}})}
= \scomplement{\strategyfor[\alpha]{\scomplement{X}}} \cap \scomplement{\strategyfor[\beta]{\scomplement{X}}}
= \dstrategyfor[\alpha]{X}\cap\dstrategyfor[\beta]{X}
= \dstrategyfor[\pchoice{\alpha}{\beta}]{X}\)

\item \(\scomplement{\strategyfor[\alpha;\beta]{\scomplement{X}}}
= \scomplement{\strategyfor[\alpha]{\strategyfor[\beta]{\scomplement{X}}}}
= \scomplement{\strategyfor[\alpha]{\scomplement{\dstrategyfor[\beta]{X}}}}
= \dstrategyfor[\alpha]{\dstrategyfor[\beta]{X}}
= \dstrategyfor[\alpha;\beta]{X}\)

\item \(\scomplement{\strategyfor[\prepeat{\alpha}]{\scomplement{X}}}
= \scomplement{\left(\capfold\{Z\subseteq\linterpretations{\Sigma}{V} \with \scomplement{X}\cup\strategyfor[\alpha]{Z}\subseteq Z\}
\right)}\)
= \(\scomplement{\left(\capfold\{Z\subseteq\linterpretations{\Sigma}{V} \with \scomplement{(X\cap\scomplement{\strategyfor[\alpha]{Z}})}\subseteq Z\}\right)}\)\\
= \(\scomplement{\left(\capfold\{Z\subseteq\linterpretations{\Sigma}{V} \with \scomplement{(X\cap\dstrategyfor[\alpha]{\scomplement{Z}})}\subseteq Z\}\right)}\)
= \(\cupfold\{Z\subseteq\linterpretations{\Sigma}{V} \with Z\subseteq X\cap\dstrategyfor[\alpha]{Z}\}
= \dstrategyfor[\prepeat{\alpha}]{X}\).
\footnote{The penultimate equation follows from the $\mu$-calculus equivalence
\(\gfp{Z}{\mapply{\Upsilon}{Z}} \mequiv \lnot\lfp{Z}{\lnot\mapply{\Upsilon}{\lnot Z}}\) and the fact that least pre-fixpoints are fixpoints and that greatest post-fixpoints are fixpoints for monotone functions.}

\item \(\scomplement{\strategyfor[\pdual{\alpha}]{\scomplement{X}}}
= \scomplement{(\scomplement{\strategyfor[\alpha]{\scomplement{(\scomplement{X})}}})}
= \scomplement{\dstrategyfor[\alpha]{\scomplement{X}}}
= \dstrategyfor[\pdual{\alpha}]{X}\)
\qedhere
\end{enumerate}
\endproofmove
One direction of \rref{thm:dGL-determined} implies \m{\entails \lnot\ddiamond{\alpha}{\lnot\phi} \limply \dbox{\alpha}{\phi}}, i.e.\ \m{\entails \ddiamond{\alpha}{\lnot\phi} \lor \dbox{\alpha}{\phi}}, whose validity means that, from all initial states, either Angel has a winning strategy to achieve $\lnot\phi$ or Demon has a winning strategy to achieve $\phi$.
That is, hybrid games are determined, because there are no states from which none of the players has a winning strategy (for the same hybrid game $\alpha$ and complementary winning conditions $\lnot\phi$ and $\phi$, respectively).
At least one player, thus, has a winning strategy for complementary winning conditions.
The other direction of \rref{thm:dGL-determined} implies \m{\entails \dbox{\alpha}{\phi} \limply \lnot\ddiamond{\alpha}{\lnot\phi}}, i.e.\ \m{\entails \lnot(\dbox{\alpha}{\phi} \land \ddiamond{\alpha}{\lnot\phi})}, whose validity means that there is no state from which Demon has a winning strategy to achieve $\phi$ and, simultaneously, Angel has a winning strategy to achieve $\lnot\phi$.
It cannot be that both players have a winning strategy for complementary conditions from the same state.
That is, hybrid games are \emph{consistent}, because at most one player has a winning strategy for complementary winning conditions.
Along with modal congruence rules, which hold for \dGL, \rref{thm:dGL-determined} makes \dGL a classical (multi)modal logic \cite{Chellas}, yet with modalities indexed by hybrid games.

Instead of giving a semantics to \m{\dbox{\cdot}{}} in terms of the existence of a winning strategy for Demon, \rref{thm:dGL-determined} could have been used as a definition of \m{\dbox{\cdot}{}}.
That would have been easier, but would have obscured determinacy and the role of \m{\dbox{\cdot}{}} as the winning strategy operator for Demon.

\subsection{Hybrid Game Equivalences} \label{sec:HG-equivalences}

As usual in programming languages, the same hybrid game can have multiple different syntactical representations.
Some equivalence transformations on hybrid games can be useful to transform hybrid games into a conceptually simpler form.

\begin{definition}[Hybrid game equivalence] \label{def:HG-equivalence}
Hybrid games $\alpha$ and $\beta$ are \emph{equivalent}, denoted \m{\alpha\mequiv\beta}, if \(\strategyfor[\alpha]{X}=\strategyfor[\beta]{X}\) for all $X$ and all $\iget[const]{\I}$.
\end{definition}
By \rref{thm:dGL-determined}, $\alpha$ and $\beta$ are equivalent iff \(\dstrategyfor[\alpha]{X}=\dstrategyfor[\beta]{X}\) for all $X$ and all $\iget[const]{\I}$.
\begin{remark} \label{rem:HG-equivalence}
The equivalences
\[
  \pdual{(\pchoice{\alpha}{\beta})} \mequiv \dchoice{\pdual{\alpha}}{\pdual{\beta}}, \quad
  \pdual{(\alpha;\beta)} \mequiv \pdual{\alpha};\pdual{\beta}, \quad
  \pdual{(\prepeat{\alpha})} \mequiv \drepeat{(\pdual{\alpha})}, \quad
  \alpha^{dd} \mequiv \alpha %
\]
on hybrid games can transform every hybrid game $\alpha$ into an equivalent hybrid game in which $\pdual{}$ only occurs right after atomic games or as part of the definition of the derived operators $\dchoice{}{}$ and $\drepeat{}$.
Other equivalences include \(\prepeat{(\pevolve{\D{x}=\genDE{x}})} \mequiv \pevolve{\D{x}=\genDE{x}}\)
and \(\prepeat{(\pevolvein{\D{x}=\genDE{x}}{\ivr})} \mequiv \pchoice{\ptest{\ltrue}}{\pevolvein{\D{x}=\genDE{x}}{\ivr}}\).
\end{remark}

Quite unlike in hybrid systems and (poor test) differential dynamic logic \cite{DBLP:journals/jar/Platzer08,DBLP:conf/lics/Platzer12b}, every hybrid game containing a differential equation \m{\pevolvein{\D{x}=\genDE{x}}{\ivr}} with evolution domain constraints $\ivr$ can be replaced equivalently by a hybrid game without evolution domain constraints (even using poor tests, i.e.\ each test $\ptest{\ivr}$ uses only first-order formulas $\ivr$).
Evolution domains are definable in hybrid games and can, thus, be removed equivalently from all hybrid games.
\begin{lemma}[Domain reduction]%
\label{lem:pevolvein}%
  Evolution domains of differential equations are definable as hybrid games:
  For every hybrid game there is an equivalent hybrid game that has no evolution domain constraints, i.e.\ all continuous evolutions are of the form \m{\pevolve{\D{x}=\genDE{x}}}.
\end{lemma}
\proofmove
{\renewcommand*{\genDE}[1]{\theta(#1)}%
For notational convenience, assume the (vectorial) differential equation \m{\pevolve{\D{x}=\genDE{x}}} to contain a clock \m{\D{\stime}=1} and that $t_0$ and $z$ are fresh variables. Then \(\pevolvein{\D{x}=\genDE{x}}{\mapply{\ivr}{x}}\) is equivalent to the hybrid game:
\begin{equation}
\pupdate{\pumod{t_0}{\stime}}; \pevolve{\D{x}=\genDE{x}}; \pdual{(\pupdate{\pumod{z}{x}}; \pevolve{\D{z}=-\genDE{z}})}; \ptest{(z_0\geq t_0\limply\mapply{\ivr}{z})}
\label{eq:evolvein}
\end{equation}%
\begin{figure}[tbh]
  \newcommand{\mtime}{t}%
  \begin{minipage}[b]{6.25cm}
    \begin{tikzpicture}[scale=1.5]
  \newcommand{\ws}{\nu}\newcommand{\wt}{}%
  \renewcommand{\I}{\iconcat[state=\ws]{\stdI}}%
  \renewcommand{\It}{\iconcat[state=\wt]{\stdI}}%
  \def\vec#1{#1}%
  \tikzstyle{axes}=[]
  \tikzstyle{mode switch}=[black!70,thin,dotted]
      \begin{scope}[style=axes]
        \draw[->] (-0.1,0) -- (2.4,0) node[right] {$\mtime$} coordinate(t axis);
        \draw[->] (0,-0.1) -- (0,1.2) node[above] {$\vec{x},\vec{z}$} coordinate(x axis);
      \end{scope}
      {
        \draw[draw=vgreen,fill=vgreen!5] (1.1,0.8) ellipse (0.9cm and 0.4cm);
        \node[color=vgreen!140] at (1.6,0.6) {$\ivr$};
      }
      \newcommand{\breakp}{1.8}
      \begin{scope}[xshift=0.7cm,yshift=-0.1cm]
        {
          \draw[thick,domain=-0.7:0.6,smooth,xshift=.5cm]
            plot
            (\x,{exp(-1.5*\x)+1.2*(1-exp(-1.5*\x))})
          node (flowend) {}
          node[above] {$\iget[state]{\It}$};
        }
        \node (flowstart) at (-0.7,0.62847) {};
        \draw[thick,vred,domain=-0.63:0.6,smooth,xshift=0.5cm,yshift=-0.11cm]
          plot[mark=triangle*,mark options={vred},mark phase=5,mark repeat=8]
          (\x,{exp(-1.5*\x)+1.2*(1-exp(-1.5*\x))})
          node (backstart) {};
        \draw[thick,dotted,vred,domain=-1:-0.63,smooth,xshift=0.5cm,yshift=-0.11cm] plot (\x,{exp(-1.5*\x)+1.2*(1-exp(-1.5*\x))});
        \tikzstyle{my loop}=[->,to path={
          .. controls +(10:0.5) and +(-10:0.5) .. (\tikztotarget) \tikztonodes}]
        \draw[thick,my loop] (flowend) to (backstart);
        \node[right,text width=4.5cm,sloped] at (1.5,1.1) {$\pupdate{\pumod{z}{x}}$};
        \node[anchor=west,text width=7.3cm] at (2.4,1.1) {\footnotesize%
        \noindent%
Angel plays forward game, reverts flow and time $\stime$;\\~Demon checks~$\ivr$ in backwards game until initial $t_0$};
      \end{scope}
      \node[above=-1pt,rotate=10] at (1,0.85) {$\pevolve{\D{x}={\genDE{x}}}$};
      \draw[mode switch] (0.5,0) node[below,black] {$\pupdate{\pumod{t_0}{\stime}}$} -- ++(0,1.1);
      \draw[mode switch] (\breakp,0) node[below,black] {$r$} -- ++(0,1.1);
      \path (0.7,0.4) -- node[below,vred] {$\pevolve{\D{z}={-\genDE{z}}}$} (\breakp,0.4);
    \end{tikzpicture}
  \end{minipage}
  \caption{``There and back again game'': Angel evolves $x$ forwards in time along  \m{\pevolve{\D{x}=\genDE{x}}}, Demon checks evolution domain backwards in time along \m{\pevolve{\D{z}=-\genDE{z}}} on a copy $z$ of the state vector $x$}
  \label{fig:backflow}
\end{figure}
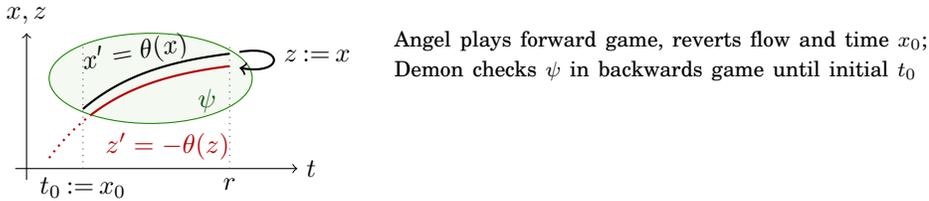

See \rref{fig:backflow} for an illustration.
Suppose the current player is Angel.
The idea behind game equivalence \rref{eq:evolvein} is that the fresh variable $t_0$ remembers the initial time $\stime$, and Angel then evolves forward along \m{\pevolve{\D{x}=\genDE{x}}} for any arbitrary amount of time (Angel's choice).
Afterwards, the opponent Demon copies the state $x$ into a fresh variable (vector) $z$ that he can evolve backwards along \m{\pdual{(\pevolve{\D{z}=-\genDE{z}})}} for any arbitrary amount of time (Demon's choice).
The original player Angel must then pass the challenge \m{\ptest{(z_0\geq t_0\limply\mapply{\ivr}{z})}}, i.e.\ Angel loses immediately if Demon was able to evolve backwards and leave region $\mapply{\ivr}{z}$ while satisfying \m{z_0\geq t_0}, which checks that Demon did not evolve backward for longer than Angel evolved forward.
Otherwise, when Angel passes the test, the extra variables $t_0,z$ become irrelevant (they are fresh) and the game continues from the current state $x$ that Angel chose in the first place (by selecting a duration for the evolution that Demon could not invalidate).
}
\endproofmove
\rref{lem:pevolvein} eliminates all evolution domain constraints equivalently in hybrid games from now on.
While evolution domain constraints are fundamental parts of hybrid systems \cite{DBLP:conf/hybrid/AlurCHH92,DBLP:conf/stoc/HenzingerKPV95,DBLP:journals/tac/BranickyBM98,DBLP:journals/jar/Platzer08}, they turn out to be mere syntactic sugar for hybrid games.
In that sense, hybrid games are more fundamental than hybrid systems, because they feature elementary operators and do not need built-in support for evolution domain constraints.

\subsection{Strategic Closure Ordinals} \label{sec:ClosureOrdinals}

In order to examine whether the \dGL semantics could be implemented directly to compute winning regions for \dGL formulas by a reachability computation or backwards induction for games, this section investigates how many iterations the fixpoint for the semantics $\strategyfor[\prepeat{\alpha}]{X}$ of repetition needs.
The number of required iterations marks a significant difference in analytic complexity of hybrid games compared to hybrid systems.

The semantics, $\strategyfor[\prepeat{\alpha}]{X}$, of $\prepeat{\alpha}$ is a least fixpoint and Knaster-Tarski's seminal fixpoint theorem entails that every least fixpoint of a monotone function on a complete lattice corresponds to some sufficiently large iteration.
That is, there is some ordinal $\bar{\lambda}$ at which the $\bar{\lambda}$th iteration, $\inflopstrat[\bar{\lambda}][\alpha]{X}$, of $\strategyfor[\alpha]{\cdot}$ coincides with $\strategyfor[\prepeat{\alpha}]{X}$, i.e.\ \(\strategyfor[\prepeat{\alpha}]{X}=\inflopstrat[\bar{\lambda}][\alpha]{X}\); see \rref{fig:strategyfor-KnasterTarski-iteration}.
How big is $\bar{\lambda}$, i.e.\ how often does $\strategyfor[\alpha]{\cdot}$ need to iterate to obtain $\strategyfor[\prepeat{\alpha}]{X}$?
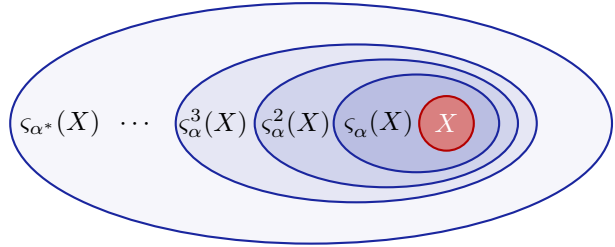
\begin{figure}[htb]
  \centering
\begin{tikzpicture}[thick]
  \tikzstyle{strategy preimage}=[draw=vblue,fill=vblue!20,shape=ellipse]
  \node[strategy preimage,fill=vblue!4,minimum width=8cm,minimum height=3.2cm] at (-1.8,0) {};
  \node[anchor=west] at (-5.8,0) {$\strategyfor[\prepeat{\alpha}]{X} ~~\cdots$};
  \node[strategy preimage,fill=vblue!11,minimum width=4.8cm,minimum height=2.1cm] at (-1.2,0) {};
  \node at (-3.1,0) {$\inflopstrat[3][\alpha]{X}$};
  \node[strategy preimage,minimum width=3.5cm,minimum height=1.7cm] at (-0.8,0) {};
  \node at (-2,0) {$\inflopstrat[2][\alpha]{X}$};
  \node[strategy preimage,fill=vblue!30,minimum width=2.2cm,minimum height=1.3cm] at (-0.4,0) {};
  \node at (-0.9,0) {$\inflopstrat[][\alpha]{X}$};
  \node[draw=vred,fill=vred!50,text=white,shape=circle] at (0,0) {$X$};
\end{tikzpicture}
  \caption{Least fixpoint $\strategyfor[\prepeat{\alpha}]{X}$ corresponds to some higher iterate $\inflopstrat[\bar{\lambda}][\alpha]{X}$ of $\strategyfor[\alpha]{\cdot}$ from winning condition $X$.}
  \label{fig:strategyfor-KnasterTarski-iteration}
\end{figure}

Recall that ordinals extend natural numbers and support (non-commutative) addition, multiplication, and exponentiation, ordered as:
\begin{multline*}
0 < 1 < 2 < \dots \omega < \omega +1 < \omega + 2 < \dots \omega\cdot2 < \omega\cdot2 + 1 < \dots \omega\cdot3 < \omega\cdot3+1 < \dots \\
\omega^2 < \omega^2 + 1 < \dots \omega^2+\omega < \omega^2+\omega + 1 < \dots \omega^\omega < \dots \omega^{\omega^\omega} < \dots \omega_1^{\text{CK}} < \dots \omega_1 < \dots
\end{multline*}
The first infinite ordinal is $\omega$, the Church-Kleene ordinal $\omega_1^{\text{CK}}$, i.e. the first nonrecursive ordinal, and $\omega_1$ the first uncountable ordinal.
Recall that every ordinal $\kappa$ is either a successor ordinal, i.e.\ the smallest ordinal \(\kappa=\iota+1\) greater than some ordinal $\iota$, or a limit ordinal, i.e.\ the supremum of all smaller ordinals.
Depending on the context, 0 is considered as a limit ordinal or as a separate case.

\subsubsection{Iterations and Fixpoints}

For each hybrid game $\alpha$, the semantics $\strategyfor[\alpha]{\cdot}$ is a monotone operator on the complete powerset lattice (\rref{lem:monotone}).
The $\kappa$th iterate, $\inflopstrat[\kappa][\alpha]{\cdot}$, of $\strategyfor[\alpha]{\cdot}$
is defined by a minor variation of Kozen's formulation of Knaster-Tarski \cite[Theorem 1.12]{Harel_et_al_2000}, obtained by considering the sublattice with $x$ at the bottom.

Let $\tau:L\to L$ be a monotone operator on a partial order $L$, then defining 
\(\displaystyle\inflop[\lambda]{x} \mdefeq x \cup \cupfold_{\kappa<\lambda} \tau(\inflop[\kappa]{x})\) 
for all ordinals $\lambda$ is equivalent to defining:%
  \vspace{-\baselineskip}
  \begin{align*}
    \inflop[0]{x} &\mdefeq x\\
    \inflop[\kappa+1]{x} &\mdefeq x \cup \inflop{\inflop[\kappa]{x}}\\
    \inflop[\lambda]{x} &\mdefeq \cupfold_{\kappa<\lambda} \inflop[\kappa]{x} \qquad \lambda\neq0~\text{a limit ordinal}
  \end{align*}
Yet, $\cupfold$ and, thus, $\inflop[\lambda]{x}$ are only guaranteed to exist if $L$ is a complete partial order.
\begin{theorem}[{Knaster-Tarski \cite[Theorem 1.12]{Harel_et_al_2000}}] \label{thm:KnasterTarski}%
  For every complete lattice $L$, there is an ordinal $\bar{\lambda}$ of at most the cardinality of $L$ such that, for each monotone $\tau:L\to L$, i.e.\ \(\tau(x)\subseteq\tau(y)\) for all \(x\subseteq y\), the fixpoints of $\tau$ in $L$ are a complete lattice and for all $x\in L$ and all ordinals $\kappa$:
  \begin{align*}
  \inflop[\dagger]{x} &\mdefeq \capfold \{z\in L \with x\subseteq z, \tau(z)\subseteq z\}
  = \inflop[\bar{\lambda}]{x}
  = \inflop[\bar{\lambda}+\kappa]{x}
  \end{align*}
\end{theorem}
The least ordinal $\bar{\lambda}$ with the property in \rref{thm:KnasterTarski} is called \emph{closure ordinal} of $\tau$.

The operator $\inflop[\kappa]{\cdot}$ enjoys useful properties.
By its extensive / inflationary definition, $\inflop[\kappa]{x}$ is not just monotone in $x$ but also monotone and homomorphic in $\kappa$.
Since $\inflop[0]{x}=x$, this works for all ordinals.
\begin{lemma}[Inductive homomorphism] \label{lem:inflop-inductive-homomorphic}%
  $\tau$ is inductive, i.e.\ \m{\inflop[\kappa]{x}\subseteq\inflop[\lambda]{x}} for all \m{\kappa\leq\lambda}
  and homomorphic in $\kappa$, i.e.\ \m{\inflop[\kappa+\lambda]{x}=\inflop[\lambda]{\inflop[\kappa]{x}}} for all \m{\kappa,\lambda}.
\end{lemma}
\proofmove
Inductiveness, i.e.\ \m{\inflop[\kappa]{x}\subseteq\inflop[\lambda]{x}} for $\kappa\leq\lambda$, which is monotonicity in $\kappa$, holds by definition \cite[Lemma 1.11]{Harel_et_al_2000}.
Homomorphy in $\kappa$, i.e.\ \m{\inflop[\kappa+\lambda]{x}=\inflop[\lambda]{\inflop[\kappa]{x}}} can be proved by induction on $\lambda$, which is either 0, a successor ordinal (second line) or a limit ordinal $\neq0$ (third line):
\begin{align*}
  \inflop[\kappa+0]{x} &=\inflop[\kappa]{x} = \inflop[0]{\inflop[\kappa]{x}}\\
  \inflop[\kappa+(\lambda+1)]{x}&
  = x\cup\tau(\inflop[\kappa+\lambda]{x}) 
  = x\cup\tau(\inflop[\lambda]{\inflop[\kappa]{x}})
  = \inflop[\kappa]{x}\cup\tau(\inflop[\lambda]{\inflop[\kappa]{x}})
  = \inflop[\lambda+1]{\inflop[\kappa]{x}}\\
  \inflop[\kappa+\lambda]{x}&
  = \cupfold_{\iota<\kappa+\lambda} \inflop[\iota]{x} 
  = \cupfold_{\iota<\kappa} \inflop[\iota]{x} \cup \cupfold_{\iota<\lambda} \inflop[\kappa+\iota]{x} 
  \\&= \cupfold_{\iota<\lambda} \inflop[\kappa+\iota]{x} 
  = \cupfold_{\iota<\lambda} \inflop[\iota]{\inflop[\kappa]{x}} 
  = \inflop[\lambda]{\inflop[\kappa]{x}} 
  \qedhere
\end{align*}
\endproofmove

By \rref{thm:KnasterTarski}, there is an ordinal $\bar{\lambda}$ of cardinality at most that of $\reals$ such that
\(\strategyfor[\prepeat{\alpha}]{X} = \inflopstrat[\bar{\lambda}][\alpha]{X}\) for all $\alpha$ and all $X$, because the powerset lattice is complete and \(\strategyfor[\alpha]{\cdot}\) monotone by \rref{lem:monotone}.
This iterative construction \(\inflop[\bar{\lambda}]{X}\) corresponds to backward induction in classical game theory \cite{vonNeumannMorgenstern,Aumann95}, yet it terminates at ordinal $\bar{\lambda}$ which is not necessarily finite.

\subsubsection{Scott-Continuity} \label{sec:HP-Scott-continuous}
The semantics of repetitions in hybrid systems repeats some finite number of times \cite{DBLP:conf/lics/Platzer12b}.
If the semantics of \dGL were Scott-continuous, this would be the case for \dGL as well, because the closure ordinal of Scott-continuous operators on a complete partial order is ${\leq}\omega$ by Kleene's fixpoint theorem. %
Dual-free $\alpha$ are indeed Scott-continuous, in particular, the closure ordinal for hybrid systems is $\omega$.
\begin{lemma}[Scott-continuity of $\pdual{}$-free \dGL] \label{lem:HP-Scott-continuous}%
  The \dGL semantics of $\pdual{}$-free $\alpha$ is Scott-continuous, i.e.\ \(\strategyfor[\alpha]{\cupfold_{n\in J} X_n} = \cupfold_{n\in J} \strategyfor[\alpha]{X_n}\) for all families $\{X_n\}_{n\in J}$ with index set~$J$.
\end{lemma}
\proofmove
By \rref{lem:monotone},
\(\cupfold_{n\in J} \strategyfor[\alpha]{X_n} \subseteq \strategyfor[\alpha]{\cupfold_{n\in J} X_n}\).
The converse inclusion can be shown by a simple induction on the structure of $\alpha$:
\(\strategyfor[\alpha]{\cupfold_{n\in J} X_n} \subseteq \cupfold_{n\in J} \strategyfor[\alpha]{X_n}\).
IH is short for induction hypothesis.
\begin{enumerate}
\item \(\strategyfor[\pupdate{\pumod{x}{\theta}}]{\cupfold_{n\in J} X_n}
= \{\iportray{\I} \in \linterpretations{\Sigma}{V} \with \modif{\iget[state]{\I}}{x}{\ivaluation{\I}{\theta}} \in \cupfold_{n\in J} X_n\}
\subseteq \cupfold_{n\in J}\{\iportray{\I} \in \linterpretations{\Sigma}{V} \with \modif{\iget[state]{\I}}{x}{\ivaluation{\I}{\theta}} \in X_n\}
= \cupfold_{n\in J} \strategyfor[\pupdate{\pumod{x}{\theta}}]{X_n}\),
since \(\modif{\iget[state]{\I}}{x}{\ivaluation{\I}{\theta}}\in \cupfold_{n\in J} X_n\)
implies \(\modif{\iget[state]{\I}}{x}{\ivaluation{\I}{\theta}}\in X_n\) for some $n$.

\item \(\strategyfor[\pevolvein{\D{x}=\genDE{x}}{\ivr}]{\cupfold_{n\in J} X_n} = \{\varphi(0) \in \linterpretations{\Sigma}{V} \with 
      \D[t]{\,\varphi(t)(x)} (\zeta) =       %
      \ivaluation{\iconcat[state=\varphi(\zeta)]{\I}}{\theta}
      ~\text{and}~
      \varphi(\zeta)\in\imodel{\I}{\ivr}\)
      for all $\zeta\leq r$
      for some (differentiable)
      \m{\varphi:[0,r]\to\linterpretations{\Sigma}{V}}
      such that $\varphi(r)\in \cupfold_{n\in J} X_n\}$
\(\subseteq \cupfold_{n\in J} \strategyfor[\pevolvein{\D{x}=\genDE{x}}{\ivr}]{X_n} = \{\varphi(0) \in \linterpretations{\Sigma}{V} \with 
      \dots
      \varphi(r)\in X_n\}\),
because \(\varphi(r) \in \cupfold_{n\in J} X_n\)
implies \(\varphi(r)\in X_n\) for some $n$.

\item \(\strategyfor[\ptest{\ivr}]{\cupfold_{n\in J} X_n} = \imodel{\I}{\ivr}\cap \cupfold_{n\in J} X_n = \cupfold_{n\in J} (\imodel{\I}{\ivr}\cap X_n) = \cupfold_{n\in J} \strategyfor[\ptest{\ivr}]{X_n}\)

\item \(\strategyfor[\pchoice{\alpha}{\beta}]{\cupfold_{n\in J} X_n} = \strategyfor[\alpha]{\cupfold_{n\in J} X_n}\cup\strategyfor[\beta]{\cupfold_{n\in J} X_n}
\stackrel{\text{IH}}{=} (\cupfold_{n\in J} \strategyfor[\alpha]{X_n})\cup(\cupfold_{n\in J} \strategyfor[\beta]{X_n})
= \cupfold_{n\in J} (\strategyfor[\alpha]{X_n}\cup\strategyfor[\beta]{X_n})
= \cupfold_{n\in J} \strategyfor[\pchoice{\alpha}{\beta}]{X_n}\)

\item \(\strategyfor[\alpha;\beta]{\cupfold_{n\in J} X_n} = \strategyfor[\alpha]{\strategyfor[\beta]{\cupfold_{n\in J} X_n}}
\stackrel{\text{IH}}{=} \strategyfor[\alpha]{\cupfold_{n\in J}\strategyfor[\beta]{X_n}}
\stackrel{\text{IH}}{=} \cupfold_{n\in J}\strategyfor[\alpha]{\strategyfor[\beta]{X_n}}
= \cupfold_{n\in J} \strategyfor[\alpha;\beta]{X_n}\)

\item \(\strategyfor[\prepeat{\alpha}]{\cupfold_{n\in J} X_n} 
= (\cupfold_{n\in J} X_n)\cup\strategyfor[\alpha]{\strategyfor[\prepeat{\alpha}]{\cupfold_{n\in J} X_n}}\) is the least fixpoint.
Prove that \(\cupfold_{n\in J} \strategyfor[\prepeat{\alpha}]{X_n}\) is a fixpoint, which implies
\(\strategyfor[\prepeat{\alpha}]{\cupfold_{n\in J} X_n} \subseteq \cupfold_{n\in J} \strategyfor[\prepeat{\alpha}]{X_n}\).
Indeed,
\((\cupfold_{n\in J} X_n)\cup\strategyfor[\alpha]{\cupfold_{n\in J} \strategyfor[\prepeat{\alpha}]{X_n}}
\stackrel{\text{IH}}{=} (\cupfold_{n\in J} X_n)\cup\cupfold_{n\in J} \strategyfor[\alpha]{\strategyfor[\prepeat{\alpha}]{X_n}}
= \cupfold_{n\in J} (X_n\cup\strategyfor[\alpha]{\strategyfor[\prepeat{\alpha}]{X_n}}
= \cupfold_{n\in J} \strategyfor[\prepeat{\alpha}]{X_n}\).
The last equation uses that $\strategyfor[\prepeat{\alpha}]{X_n}$ is a fixpoint.
\qedhere
\end{enumerate}
\endproofmove
But $\strategyfor[\alpha]{\cdot}$ is not generally Scott-continuous, so $\bar{\lambda}$ might potentially be greater than $\omega$ for hybrid games.
Games with both $\pdual{}$ and $\prepeat{}$ do not generally have a Scott-continuous semantics nor an $\omega$-chain continuous semantics, i.e.\ they are not even continuous for a monotonically increasing chain \(X_0\subseteq X_1\subseteq X_2 \subseteq \dots\) with $\omega$ as index set:
\begin{align}
& \reals = 
\strategyfor[\drepeat{\pupdate{\pumod{y}{y+1}}}]{%
{\cupfold_{n<\omega}(-\infty,n]}%
}
\nsubseteq
\cupfold_{n<\omega}\strategyfor[\drepeat{\pupdate{\pumod{y}{y+1}}}]{(-\infty,n]}
=\emptyset
\notag%
\\
&\text{hence}~
\entails \ddiamond{\drepeat{\pupdate{\pumod{y}{y+1}}}}{%
{\lexists[\naturals]{n}{y\leq n}}%
}
~\text{but}~
\nonentails \lexists[\naturals]{n}{\ddiamond{\drepeat{\pupdate{\pumod{y}{y+1}}}}{y\leq n}}
\notag
\end{align}
This example shows that, even though Angel wins this game, there is no upper bound $<\omega$ on the number of iterations it takes her to win, because Demon could repeat \m{\drepeat{\pupdate{\pumod{y}{y+1}}}} arbitrarily often.
This phenomenon is directly related to a failure of the Barcan axiom (\rref{sec:separating-axioms}).
The quantifier $\lexists[\naturals]{n}{}$over natural numbers is not essential here \cite{DBLP:journals/jar/Platzer08} but mere convenience to make both lines above match directly.

If $\tau$ is countably-continuous, i.e.\ continuous for families with countable index sets, on a complete partial order, then its closure ordinal is $\bar{\lambda}\leq\omega_1$.
But this is not the case for $\strategyfor[\alpha]{\cdot}$ either, by the above counterexample with countable index set $\omega$.

A function $\tau$ on sets is $\kappa$-based, for an ordinal $\kappa$, if for all $X$, $x\in\tau(X)$ implies $x\in\tau(Y)$ for some $Y\subseteq X$ of cardinality ${<}\kappa$.
If $\tau$ is $\kappa$-based, then its closure ordinal is ${\leq}\kappa$ \cite[Proposition 1.3.4]{Aczel77}.
The semantics $\strategyfor[\alpha]{\cdot}$ is not even $\omega_1$-based, however, because of \rref{lem:monotone} and removing just one state from the winning condition may lose states in the winning region:
\begin{align*}
& [0,\infty) = 
\strategyfor[\devolve{\D{x}=1}]{[0,\infty)}
\\&\text{but}~
0\not\in\strategyfor[\devolve{\D{x}=1}]{[0,\infty)\setminus\{a\}}=(a,\infty)
~\text{for all}~a>0
\end{align*}%
Consequently, classical bounds on closure ordinals all fail to apply due to the combination of $\pdual{}, \prepeat{}$, and differential equations that makes hybrid games challenging.

\subsubsection{Transfinite Closure Ordinals}
\renewcommand{\inflop}[2][]{\inflopstrat[#1][\alpha]{#2}}%
When will the iteration for the fixpoints in the winning region definitions stop?
Hybrid games may have higher closure ordinals, because $\omega$ many repetitions of the operator (and even ${<}\dGLordinal$ many) may not be enough to compute winning regions.
In other words, $\strategyfor[\prepeat{\alpha}]{X}$ will coincide with iterations $\inflopstrat[\kappa][\alpha]{X}$ as illustrated in \rref{fig:strategyfor-KnasterTarski-iteration}, but this may need many more than $\omega$ iterations to terminate.

\begin{theorem}[Closure ordinals] \label{thm:dGL-closure-lower}%
  The semantics of \dGL has a closure ordinal $\geq\dGLordinal$, i.e.\
  for all $\lambda<\dGLordinal$, there are $\alpha$ and $X$ such that
  \(\strategyfor[\prepeat{\alpha}]{X}\neq\inflop[\lambda]{X}\).
\end{theorem}
\proofmove
{\renewcommand{\inflop}[2][]{\inflopstrat[#1][\pchoice{\alpha}{\beta}]{#2}}%
For concreteness, the proof first shows the weaker bound $\geq\omega\cdot2$.
Minor syntactic variations lead to vastly different closure ordinals (\rref{app:dGL-closure-lower}), so the closure ordinal is not a simple function of the syntactic structure.

To see that the closure ordinal is $>\omega$ even with just one variable, a single loop and dual, consider the set of states in which the following \dGL formula is true:
\begin{equation}
\ddiamond{\prepeat{(\pchoice{\underbrace{\pupdate{\pumod{x}{x+1}}; \devolve{\D{x}=1}}_{\alpha}}{\underbrace{\pupdate{\pumod{x}{x-1}}}_{\beta}})}}{\, (0\leq x<1)}
\label{eq:omega2-closure}
\end{equation}
The winning regions for this \dGL formula stabilize after $\omega\cdot2$ iterations, because $\omega$ many iterations are necessary to show that \emph{all} positive reals can be reduced to $[0,1)$ by choosing $\beta$ sufficiently often, whereas another $\omega$ many iterations are needed to show that choice $\alpha$, which makes progress $\geq1$ but possibly more under Demon's control, can turn $x$ into some positive real.
It is easy to see that \m{\inflop[\omega]{[0,1)}=\cupfold_{n<\omega}\inflop[n]{[0,1)}=[0,\infty)}, because \m{\inflop[n]{[0,1)}=[0,n+1)} holds for all $n\in\naturals$ by a simple inductive argument:
\begin{align*}
  &\inflop[0]{[0,1)} = [0,1)\\
  &\inflop[n+1]{[0,1)} = [0,1)\cup\strategyfor[\pchoice{\alpha}{\beta}]{\inflop[n]{[0,1)}}
  = [0,1)\cup\strategyfor[\pchoice{\alpha}{\beta}]{[0,n+1)}
  \\&= [0,1)\cup\strategyfor[\alpha]{[0,n+1)} \cup \strategyfor[\beta]{[0,n+1)}
  = [0,1)\cup\emptyset \cup [1,n+1+1)
\end{align*}
But the iteration for the winning region does not stop at $\omega$, as \m{\inflop[\omega+n]{[0,1)}=[-n,\infty)} holds for all $n\in\naturals$ by another simple inductive argument:
\begin{align*}
  \inflop[\omega+n+1]{[0,1)} &= [0,1)\cup\strategyfor[\pchoice{\alpha}{\beta}]{\inflop[\omega+n]{[0,1)}}
  \\&= [0,1)\cup\strategyfor[\pchoice{\alpha}{\beta}]{[-n,\infty)}
  \\&= [0,1)\cup\strategyfor[\alpha]{[-n,\infty)} \cup \strategyfor[\beta]{[-n,\infty)}
  \\&= [-n-1,\infty) \cup [-n,\infty)
\end{align*}
Thus,
\(\inflop[\omega\cdot2]{[0,1)}=\inflop[\omega+\omega]{[0,1)}=\cupfold_{n<\omega}\inflop[\omega+n]{[0,1)}=\reals = 
 \strategyfor[\pchoice{\alpha}{\beta}]{\reals}\).
In this case, the closure ordinal is $\omega\cdot2>\omega$, since \m{\strategyfor[\prepeat{(\pchoice{\alpha}{\beta})}]{[0,1)}=\reals \neq \inflop[\omega+n]{[0,1)}} for all $n\in\naturals$. 
}%

To show that the closure ordinal is $\geq\omega_1^{\text{CK}}$, fix an ordinal \(\lambda<\omega_1^{\text{CK}}\), i.e.\ a recursive ordinal. Let \({\prec} \subseteq M \times M\) be a corresponding recursive well-order of order type $\lambda$ on a corresponding set \(M\subseteq\reals\).\footnote{A \emph{well-order} is a linear order $\prec$ on $M$ in which every non-empty subset has a least element.
Two sets $M,N$ have equal \emph{order type} iff they have an order-isomorphism \(\varphi:M\to N\), i.e.\ a monotone bijection with monotone inverse.
More background can be found in the literature \cite{Rogers}.}
That is, let $f_\prec$ a recursive function such that the relation \(x\prec y\) given by \(f_\prec(x,y)=0\) defines a well-order on the set \(M \mdefeq \{x \in\reals \with f_\prec(x,y)=0 ~\text{or} f_\prec(y,x)=0 \text{ for some } y\in\reals\}\).
Without loss of generality, assume that $0\in M$ is the least element of $M$ with respect to $\prec$.
Since $\prec$ is recursive, denote by $\ptest{f_\prec(x,y)=0}$ the program that does not change the value of variables $x,y$ and that implements the recursive function that terminates if $x\in M$ and either \(x\prec y\) or $y\not\in M$ and that otherwise fails (like \(\ptest{(0=1)}\) would).
Consider the \dGL formula
\begin{equation}
  \ddiamond{\prepeat{\big(
    \underbrace{
    \pupdate{\pumod{y}{x}};
    \pdual{(\pevolve{\D{x}=1};\pevolve{\D{x}=-1}; \ptest{f_\prec(x,y)=0})}
    }_\alpha
  \big)}}{\,x=0}
  \label{eq:omega1CK-closure}
\end{equation}
By definition of \(\ptest{f_\prec(x,y)=0}\), formula \rref{eq:omega1CK-closure} is valid, because $x$ is in $M$ after each successful run of \(\ptest{f_\prec(x,y)=0}\), and $\prec$ is a well-order on $M$ with least element 0.
By construction,
\(\strategyfor[\alpha]{X} = \{a\in\reals \with b\in X \text{ for all $b$ with } f_\prec(b,a)=0 \}\) for $X\subseteq\reals$.
Since $\prec$ has order type $\lambda$, \(\inflopstrat[\kappa][\alpha]{\{0\}} \neq\inflopstrat[\lambda][\alpha]{\{0\}}=M\) for all \(\kappa<\lambda\), otherwise the \(\inflopstrat[\iota][\alpha]{\{0\}}\) would induce a monotone injection (even order-isomorphism) from $M$ to $\kappa<\lambda$, which is a contradiction.
Indeed, \(\varphi:M\to\kappa; x\mapsto\inf\{\iota \with x\in\inflopstrat[\iota][\alpha]{\{0\}}\}\) would otherwise be a monotone injection as \(x\prec y\) in $M$ implies \(\varphi(x)<\varphi(y)\), because \(\varphi(x)\geq\varphi(y)\) implies \(y\in\strategyfor[\alpha]{X}\) for a set \(X=\inflopstrat[\varphi(y)-1][\alpha]{\{0\}}\) that does not contain $x$, contradicting \(x\prec y\).
Note that $\varphi(y)$ is a successor ordinal and hence $\varphi(y)-1$ defined, since $\varphi$ maps into successor ordinals and 0 by the definition of $\varphi$.
Consequently,
\(\inflopstrat[\lambda][\alpha]{\{0\}}=M\neq\inflopstrat[\lambda+1][\alpha]{\{0\}}=\strategyfor[\alpha]{M}=\reals=\strategyfor[\prepeat{\alpha}]{\{0\}}\),
where $M\neq\reals$ because $\lambda$ is recursive hence countable and $\prec$ a linear order on $M$.
Thus, the closure ordinal for formula \rref{eq:omega1CK-closure} is $\lambda+1>\lambda$.
Hence, for all recursive ordinals $\lambda$, there is a hybrid game with a bigger closure ordinal.
So, the closure ordinal is $\geq\omega_1^{\text{CK}}$.
\endproofmove

By \rref{thm:dGL-closure-lower}, the closure ordinal for \dGL is between $\omega_1^{\text{CK}}$ and ordinals of the cardinality of the reals (\rref{thm:KnasterTarski}).
The same proof works for other well-orderings that are definable in hybrid games, not just those that are definable by classical recursive functions.
The proof does not permit arbitrary well-orderings of the real numbers, however, because those may not be definable by hybrid games.
Hence, the closure ordinal for \dGL is at least $\omega_1^{\text{HG}}$, defined as the first ordinal $\lambda$ that does not have a well-ordering of order type $\lambda$ that is definable by hybrid games.
This ordinal satisfies \(\omega_1^{\text{CK}} \leq \omega_1^{\text{HG}}\) and is at most of the cardinality of the reals.
A more precise grasp on \(\omega_1^{\text{HG}}\) is in \rref{sec:dGL-dL-expressiveness}.

The fact that hybrid games require highly transfinite closure ordinals has a number of consequences.
It makes reachability computations and backwards induction difficult, because they only terminate after significantly more than $\omega$-infinitely many steps.
It requires higher bounds on the number of repetitions played in hybrid games.
It causes classical arguments for relative completeness to fail (\rref{sec:dGL-complete}).
And it causes acute semantical differences that are only visible in hybrid games, not in hybrid systems.
For example, the \dGL semantics is more general than defining \(\strategyfor[\prepeat{\alpha}]{X}\) to be truncated to $\omega$-repetition \(\inflop[\omega]{X} = \cupfold_{n<\omega}\inflop[n]{X}\), which misses out on the existence of perfectly natural winning strategies.
The semantics of \dGL is also different than advance notice semantics.
For reference, both comparisons are elaborated in \rref{app:alternative-semantics}.

\section{Axiomatization} \label{sec:dGL-axiomatization}

Section~\ref{sec:dGL} has defined \dGL so that every game play has exactly one winner.
Section~\ref{sec:MetaProperties} has shown that hybrid games are determined, i.e.\ from every state, exactly one of the players has a winning strategy for complementary winning conditions.
But how can one find out which of the players that is?
In principle, one could follow the iterated winning region construction according to the semantics until reaching a fixpoint (\rref{sec:ClosureOrdinals}), which corresponds to the reachability computation underlying model checking as well as to the backwards induction technique in games.
But iterated winning region constructions would not generally terminate in finite time, because the closure ordinal is highly transfinite by \rref{thm:dGL-closure-lower}.

Every \dGL sentence without free variables or predicate symbols is either true or false, because \dGL is a classical logic.
But the semantics of \dGL formulas and hybrid games is ineffective, because computing the semantics, like classical model checking or game solving would, requires transfinite computations for the winning regions.
This calls for other ways of proving the validity of \dGL formulas.

Simple \dGL formulas can be checked by a tableau procedure that expands all choices and detects loops for termination as in the game tree examples (\rref{fig:nondetermined} and Appendix).
This principle, however, does not extend to more general hybrid games with differential equations, inherently infinite state spaces, and which need higher ordinals of iteration for computing winning regions by \rref{thm:dGL-closure-lower}.

\subsection{Proof Calculus} \label{sec:dGL-calculus}
A Hilbert-type proof calculus for proving validity of \dGL formulas is presented in \rref{fig:dGL}.%
\begin{figure}[bth]
  \centering
  \renewcommand*{\irrulename}[1]{\text{#1}}%
  \renewcommand{\linferenceRuleNameSeparation}{~~}
  \newdimen\linferenceRulehskipamount%
  \linferenceRulehskipamount=1mm%
  \newdimen\lcalculuscollectionvskipamount%
  \lcalculuscollectionvskipamount=0.1em%
  \begin{calculuscollections}{\columnwidth}
    \begin{calculus}
      \cinferenceRule[box|$\dibox{\cdot}$]{box axiom}
      {\linferenceRule[equiv]
        {\lnot\ddiamond{\alpha}{\lnot\phi}}
        {\dbox{\alpha}{\phi}}
      }
      {}
      \cinferenceRule[assignd|$\didia{:=}$]{assignment / substitution axiom}
      {\linferenceRule[equiv]
        {\mapply[x]{\phi}{\theta}}
        {\ddiamond{\pupdate{\umod{x}{\theta}}}{\mapply[x]{\phi}{x}}}
      }
      {}%
      \cinferenceRule[evolved|$\didia{'}$]{evolve}
      {\linferenceRule[equiv]
        {\lexists{t{\geq}0}{\ddiamond{\pupdate{\pumod{x}{\solf(t)}}}{\phi}}\hspace{0.5cm}}
        {\ddiamond{\pevolve{\D{x}=\genDE{x}}}{\phi}}
      }{\m{\D{\solf}(t)=\genDE{\solf}}}%
      \cinferenceRule[testd|$\didia{?}$]{test}
      {\linferenceRule[equiv]
        {(\ivr \land \phi)}
        {\ddiamond{\ptest{\ivr}}{\phi}}
      }{}
      \cinferenceRule[choiced|$\didia{\cup}$]{axiom of nondeterministic choice}
      {\linferenceRule[equiv]
        {(\ddiamond{\alpha}{\phi} \lor \ddiamond{\beta}{\phi})}
        {\ddiamond{\pchoice{\alpha}{\beta}}{\phi}}
      }{}
      \cinferenceRule[composed|$\didia{{;}}$]{composition}
      {\linferenceRule[equiv]
        {\ddiamond{\alpha}{\ddiamond{\beta}{\phi}}}
        {\ddiamond{\alpha;\beta}{\phi}}
      }{}
      \cinferenceRule[iterated|$\didia{{}^*}$]{iteration/repeat unwind pre-fixpoint, even fixpoint}
      {\linferenceRule[impl]
        {(\phi \lor \ddiamond{\alpha}{\ddiamond{\prepeat{\alpha}}{\phi}})}
        {\ddiamond{\prepeat{\alpha}}{\phi}}
      }{}%
      \cinferenceRule[duald|$\didia{{^d}}$]{dual}
      {\linferenceRule[equiv]
        {\lnot\ddiamond{\alpha}{\lnot\phi}}
        {\ddiamond{\pdual{\alpha}}{\phi}}
      }{}
      \cinferenceRule[M|M]{$\ddiamond{}{}$ monotone / $\ddiamond{}{}$-generalization} %
      {\linferenceRule[formula]
        {\phi\limply\psi}
        {\ddiamond{\alpha}{\phi}\limply\ddiamond{\alpha}{\psi}}
      }{}
      \cinferenceRule[FP|FP]{iteration is least fixpoint / reflexive transitive closure RTC, equivalent to invind in the presence of R}
      {\linferenceRule[formula]
        {(\phi \lor \ddiamond{\alpha}{\psi}) \limply \psi}
        {\ddiamond{\prepeat{\alpha}}{\phi} \limply \psi}
      }{}
    \end{calculus}%
  \end{calculuscollections}
  \caption{Differential game logic axiomatization}
  \label{fig:dGL}
\end{figure}%

The logic \dGL simultaneously generalizes logics of hybrid systems and logics of discrete games and so does its proof calculus.
The proof calculus of \dGL shares axioms with differential dynamic logic \cite{DBLP:conf/lics/Platzer12b} and discrete game logic \cite{DBLP:journals/sLogica/PaulyP03a}.
It is based on the first-order Hilbert calculus (modus ponens, uniform substitution, and Bernays' $\forall$-generalization) with all instances of valid formulas of first-order logic as axioms, including (decidable) first-order real arithmetic \cite{tarski_decisionalgebra51}.
Write \m{\infers \phi} iff \dGL formula $\phi$ can be \emph{proved} with the \dGL proof rules from \dGL axioms (\rref{fig:dGL}).
That is, a \dGL formula is inductively defined to be \emph{provable} in the \dGL calculus if it is an instance of a \dGL axiom or if it is the conclusion (below the rule bar) of an instance of one of the \dGL proof rules \irref{M}, \irref{FP}, modus ponens, uniform substitution, or $\forall$-generalization, whose premises (above the rule bar) are all provable.

The \emph{determinacy axiom} \irref{box} describes the duality of winning strategies for complementary winning conditions of Angel and Demon, i.e.\ that Demon has a winning strategy to achieve $\phi$ in hybrid game $\alpha$ if and only if Angel does not have a counter strategy, i.e.\ winning strategy to achieve $\lnot\phi$ in the same game $\alpha$.
Axiom \irref{assignd} is for \emph{assignments by substitution}.
Formula $\mapply[x]{\phi}{\theta}$ is obtained from $\mapply[x]{\phi}{x}$ by \emph{substituting} $\theta$ for $x$ at all occurrences of $x$, provided $x$ does not occur in the scope of a quantifier or modality binding $x$ or a variable of $\theta$. A modality containing \m{\pupdate{\pumod{x}{}}} or $\D{x}$ outside the scope of tests $\ptest{\ivr}$ or evolution domain constraints \emph{binds} $x$\ignore{(written $x\in BV(\alpha)$)}, because it may change the value of $x$.
In the \emph{differential equation axiom} \irref{evolved}, $\solf(\cdot)$ is the unique \cite[Theorem~10.VI]{Walter:ODE} solution of the symbolic initial value problem \m{\D{\solf}(t)=\genDE{\solf},\solf(0)=x}. The duration $t$ how long to follow solution $\solf$ is for Angel to decide, hence existentially quantified.
It goes without saying that variables like $t$ are fresh in \rref{fig:dGL}.

Axioms \irref{testd}, \irref{choiced}, and \irref{composed} are as in dynamic logic \cite{DBLP:conf/focs/Pratt76} and differential dynamic logic \cite{DBLP:conf/lics/Platzer12b} except that their meaning is quite different, because they refer to winning strategies of hybrid games instead of reachability relations of systems.
The \emph{challenge axiom} \irref{testd} expresses that Angel has a winning strategy to achieve $\phi$ in the test game $\ptest{\ivr}$ exactly from those positions that are already in $\phi$ (because $\ptest{\ivr}$ does not change the state) and that satisfy $\ivr$ for otherwise she would fail the test and lose the game immediately.
The \emph{axiom of choice} \irref{choiced} expresses that Angel has a winning strategy in a game of choice \m{\pchoice{\alpha}{\beta}} to achieve $\phi$ iff she has a winning strategy in either hybrid game $\alpha$ or in $\beta$, because she can choose which one to play.
The \emph{sequential game axiom} \irref{composed} expresses that Angel has a winning strategy in a sequential game $\alpha;\beta$ to achieve $\phi$ iff she has a winning strategy in game $\alpha$ to achieve \m{\ddiamond{\beta}{\phi}}, i.e.\ to get to a position from which she has a winning strategy in game $\beta$ to achieve $\phi$.
The \emph{iteration axiom} \irref{iterated} characterizes $\ddiamond{\prepeat{\alpha}}{\phi}$ as a pre-fixpoint.
It expresses that, if the game is already in a state satisfying $\phi$ or if Angel has a winning strategy for game $\alpha$ to achieve $\ddiamond{\prepeat{\alpha}}{\phi}$, i.e.\ to get to a position from which she has a winning strategy for game $\prepeat{\alpha}$ to achieve $\phi$, then, either way, Angel has a winning strategy to achieve $\phi$ in game $\prepeat{\alpha}$.
The converse of \irref{iterated} can be derived\footnote{
\(\phi \lor \ddiamond{\alpha}{\ddiamond{\prepeat{\alpha}}{\phi}} \limply \ddiamond{\prepeat{\alpha}}{\phi}\) derives by \irref{iterated}.
Thus, \(\ddiamond{\alpha}{(\phi \lor \ddiamond{\alpha}{\ddiamond{\prepeat{\alpha}}{\phi}})} \limply \ddiamond{\alpha}{\ddiamond{\prepeat{\alpha}}{\phi}}\) by \irref{M}.
Hence,
\(\phi\lor\ddiamond{\alpha}{(\phi \lor \ddiamond{\alpha}{\ddiamond{\prepeat{\alpha}}{\phi}})} \limply \phi\lor\ddiamond{\alpha}{\ddiamond{\prepeat{\alpha}}{\phi}}\) by propositional congruence.
Consequently,
\(\ddiamond{\prepeat{\alpha}}{\phi} \limply \phi\lor\ddiamond{\alpha}{\ddiamond{\prepeat{\alpha}}{\phi}}\) by \irref{FP}.
} and is also denoted by \irref{iterated}.
The \emph{dual axiom} \irref{duald} characterizes dual games.
It says that Angel has a winning strategy to achieve $\phi$ in dual game $\pdual{\alpha}$ iff Angel does not have a winning strategy to achieve $\lnot\phi$ in game $\alpha$.
Combining dual game axiom \irref{duald} with the determinacy axiom \irref{box} yields\m{\ddiamond{\pdual{\alpha}}{\phi} \lbisubjunct \dbox{\alpha}{\phi}}, i.e.\ that Angel has a winning strategy to achieve $\phi$ in $\pdual{\alpha}$ iff Demon has a winning strategy to achieve $\phi$ in $\alpha$.
Similar reasoning derives \m{\dbox{\pdual{\alpha}}{\phi} \lbisubjunct \ddiamond{\alpha}{\phi}}.

\emph{Monotonicity rule} \irref{M} is the generalization rule of monotone modal logic \textbf{C} \cite{Chellas}.
It expresses that, if the implication $\phi\limply\psi$ is valid, then, from wherever Angel has a winning strategy in a hybrid game $\alpha$ to achieve $\phi$, she also has a winning strategy to achieve $\psi$, because $\psi$ holds wherever $\phi$ does.
So rule \irref{M} expresses that easier objectives are easier to win.
\emph{Fixpoint rule} \irref{FP} characterizes $\ddiamond{\prepeat{\alpha}}{\phi}$ as a \emph{least} pre-fixpoint.
It says that, if $\psi$ is another formula that is a pre-fixpoint, i.e.\ that holds in all states that satisfy $\phi$ or from which Angel has a winning strategy in game $\alpha$ to achieve that condition $\psi$, then $\psi$ also holds wherever $\ddiamond{\prepeat{\alpha}}{\phi}$ does, i.e.\ in all states from which Angel has a winning strategy in game $\prepeat{\alpha}$ to achieve $\phi$.

As usual, all substitutions in \rref{fig:dGL} are required to be \emph{admissible} to avoid capture of variables, i.e.\ they require all variables $x$ that are being replaced or that occur in their replacements to not occur in the scope of a quantifier or modality binding $x$.
The \emph{uniform substitution} rule \irref{US} \cite[\S35,40]{Church_1956} from first-order logic substitutes \emph{all} occurrences of predicate $p(\cdot)$ by a \dGL formula $\mapply{\psi}{\cdot}$, i.e.\ it replaces all occurrences of $p(\theta)$, for a vectorial term $\theta$, by the corresponding $\mapply{\psi}{\theta}$ simultaneously:
\[
      \cinferenceRule[US|US]{uniform substitution}
      {\linferenceRule[formula]
        {\preusubst[\phi]{p}}
        {\usubst[\phi]{p}{\psi}}
      }{}%
\]
In particular, rule \irref{US} requires all relevant substitutions of $\mapply{\psi}{\theta}$ for $p(\theta)$ to be admissible and requires that no $p(\theta)$ occurs in the scope of a quantifier or modality binding a variable of $\mapply{\psi}{\theta}$ other than the occurrences in $\theta$.
If admissible, the formula $\mapply{\psi}{\theta}$ can use variables other than those in $\theta$, hence, the case where $p$ is a predicate symbol without arguments enables rule \irref{US} to generate all formula instances from the \dGL axioms.
Rule \irref{US} turns axioms into axiom schemes \cite[\S35,40]{Church_1956}, which is a powerful principle that extends to modalities with program constants but is beyond the scope of this article and is pursued in followup work \cite{DBLP:conf/cade/Platzer15}.

Despite their fundamentally different semantics (reachability relations on states of hybrid system runs versus existence of winning strategies into sets of states of interactive hybrid game play) and different dynamical effects (mixed discrete, continuous, and adversarial dynamics), the axiomatization of \dGL ends up surprisingly close to that of the logic \dL for hybrid systems \cite{DBLP:conf/lics/Platzer12b}.
The primary difference of the axiomatization of \dGL compared to that of \dL is the addition of axiom \irref{duald} for dual games, the absence of axiom K, absence of G\"odel's necessitation rule (\dGL only has the monotone modal rule \irref{M}), absence of the Barcan formula (the converse Barcan formula is still derivable\footnote{\label{foot:converseBarcan} From \(\phi\limply\lexists{x}{\phi}\), derive \(\ddiamond{\alpha}{\phi}\limply\ddiamond{\alpha}{\lexists{x}{\phi}}\) by \irref{M}, from which first-order logic derives \(\lforall{x}{(\ddiamond{\alpha}{\phi}\limply\ddiamond{\alpha}{\lexists{x}{\phi}})}\) and then derives \(\lexists{x}{\ddiamond{\alpha}{\phi}}\limply\ddiamond{\alpha}{\lexists{x}{\phi}}\), since converse Barcan assumes that $x$ is not free in the succedent.}), absence of vacuity V, and absence of the hybrid version of Harel's convergence rule \cite{DBLP:conf/stoc/HarelMP77}.
Due to the absence of K, the induction axiom and the convergence axiom are absent in \dGL, while corresponding proof rules are still valid; see \rref{sec:separating-axioms} for details. The induction rule (\irref{invind}) is derivable from \irref{FP}.

A proof of a classical result about the interderivability of \irref{FP} with the induction rule \irref{invind} is included for the sake of completeness.
\begin{lemma}[Invariance] \label{lem:FP-invind}%
Rule \irref{FP} and the induction rule (\irref{invind}) of dynamic logic are interderivable in the \dGL calculus:
\upshape
\[
      \dinferenceRule[invind|ind]{inductive invariant}
      {\linferenceRule[formula]
        {\psi\limply\dbox{\alpha}{\psi}}
        {\psi\limply\dbox{\prepeat{\alpha}}{\psi}}
      }{}
\]
\end{lemma}
\proofmove
Rule \irref{invind} derives from \irref{FP}:
First derive the following minor variant
\[
      \dinferenceRule[invind2|ind$_R$]{inductive invariant}
      {\linferenceRule[formula]
        {\psi\limply\dbox{\alpha}{\psi} & \psi\limply\phi}
        {\psi\limply\dbox{\prepeat{\alpha}}{\phi}}
      }{}
\]
From \(\psi\limply\dbox{\alpha}{\psi}\) and \(\psi\limply\phi\) propositionally derive
\(\psi\limply\phi\land\dbox{\alpha}{\psi}\),
from which contraposition and propositional logic yield
\(\lnot\phi\lor\lnot\dbox{\alpha}{\psi}\limply\lnot\psi\).
With \irref{box}, this gives \(\lnot\phi\lor\ddiamond{\alpha}{\lnot\psi}\limply\lnot\psi\).
Now \irref{FP} derives \(\ddiamond{\prepeat{\alpha}}{\lnot\phi}\limply\lnot\psi\), 
which, by \irref{box}, is \(\lnot\dbox{\prepeat{\alpha}}{\phi}\limply\lnot\psi\),
which gives \(\psi\limply\dbox{\prepeat{\alpha}}{\phi}\) by contraposition.
The classical $\dbox{}{}$-induction rule \irref{invind} follows by $\phi\mdefequiv\psi$.
From \irref{invind}, the variant \irref{invind2} is derivable again by \irref{M} on $\psi\limply\phi$.

Rule \irref{FP} derives from \irref{invind}:
From \(\phi\lor\ddiamond{\alpha}{\psi}\limply\psi\), propositionally derive
\(\phi\limply\psi\) and \(\ddiamond{\alpha}{\psi}\limply\psi\).
By \irref{M}, the former gives \(\ddiamond{\prepeat{\alpha}}{\phi}\limply\ddiamond{\prepeat{\alpha}}{\psi}\).
By contraposition, the latter derives
\(\lnot\psi\limply\lnot\ddiamond{\alpha}{\psi}\),
which gives
\(\lnot\psi\limply\dbox{\alpha}{\lnot\psi}\) by \irref{box}.
Now \irref{invind} derives
\(\lnot\psi\limply\dbox{\prepeat{\alpha}}{\lnot\psi}\).
By contraposition
\(\lnot\dbox{\prepeat{\alpha}}{\lnot\psi}\limply\psi\), which, by \irref{box}, is
\(\ddiamond{\prepeat{\alpha}}{\psi}\limply\psi\).
Thus,
\(\ddiamond{\prepeat{\alpha}}{\phi}\limply\psi\) by the formula derived above.
\qedhere
\endproofmove
Hence, the \dGL calculus could have been equipped with rule \irref{invind} instead of \irref{FP}.

\begin{example}
The dual filibuster game formula \rref{eq:dual-filibuster} from \rref{sec:Determinacy} proves easily by going back and forth between players:
\begin{sequentdeduction}[array]
\linfer%
{\linfer[duald]
{\linfer[box]
  {\linfer[invind]
    {\linfer[box]
      {\linfer%
      {\linfer[duald]
        {\linfer[choiced]
          {\linfer[assignd]
            {\linfer[qear]
              {\lclose}
              {\lsequent{x=0}{0=0\lor1=0}}
            }
            {\lsequent{x=0}{\ddiamond{\pumod{x}{0}}{x=0}\lor\ddiamond{\pumod{x}{1}}{x=0}}}
          }%
          {\lsequent{x=0}{\ddiamond{\pchoice{\pumod{x}{0}}{\pumod{x}{1}}}{x=0}}}
        }%
        {\lsequent{x=0}{\lnot\ddiamond{\pdual{(\pchoice{\pumod{x}{0}}{\pumod{x}{1}})}}{\lnot x=0}}}
        }%
        {\lsequent{x=0}{\lnot\ddiamond{\dchoice{\pumod{x}{0}}{\pumod{x}{1}}}{\lnot x=0}}}
      }%
      {\lsequent{x=0}{\dbox{\dchoice{\pumod{x}{0}}{\pumod{x}{1}}}{x=0}}}
    }%
    {\lsequent{x=0}{\dbox{\prepeat{(\dchoice{\pumod{x}{0}}{\pumod{x}{1}})}}{x=0}}}
  }%
  {\lsequent{x=0}{\lnot\ddiamond{\prepeat{(\dchoice{\pumod{x}{0}}{\pumod{x}{1}})}}{\lnot x=0}}}
}%
{\lsequent{x=0}{\ddiamond{\pdual{\prepeat{(\dchoice{\pumod{x}{0}}{\pumod{x}{1}})}\strut}}{x=0}}}
}%
{\lsequent{x=0}{\ddiamond{\drepeat{(\pchoice{\pumod{x}{0}}{\pumod{x}{1}})}}{x=0}}}
\end{sequentdeduction}
The unmarked proof steps expand the definitions for $\dchoice{}{}$ and $\drepeat{}$.
By pushing dualities through with \rref{rem:HG-equivalence}, for example, the goal formula \rref{eq:dual-filibuster} at the bottom is equivalent to
\({\lsequent{x=0}{\ddiamond{\pdual{\prepeat{(\dchoice{\pumod{x}{0}}{\pumod{x}{1}})}\strut}}{x=0}}}\),
since assignments are unaffected by $\pdual{}$.
\end{example}

A proof of a $\ddiamond{\prepeat{\alpha}}{}$ property will be considered later, because the proof technique for those properties comes from the completeness proof.

\subsection{Soundness} \label{sec:dGL-sound}

Soundness studies whether all provable formulas are valid, which is crucial for ensuring that \dGL proofs always produce correct verification results about hybrid games.
The soundness proof uses that the following modal congruence rule derives from two uses of the monotonicity rule \irref{M}:
  \[
      \dinferenceRule[RE|RE]{$\ddiamond{}{}$ congruence} %
      {\linferenceRule[formula]
        {\phi\lbisubjunct\psi}
        {\ddiamond{\alpha}{\phi}\lbisubjunct\ddiamond{\alpha}{\psi}}
      }{}
  \]

\begin{theorem}[Soundness] \label{thm:dGL-sound}%
  The \dGL proof calculus in \rref{fig:dGL} is sound, i.e.\ all provable formulas are valid.
\end{theorem}
\proofmove
The \dGL proof calculus is sound if all instances of axioms and proof rules are sound.
Proving soundness of an implication axiom \m{\phi\limply\psi} considers an interpretation $\iget[const]{\I}$ with a set of states \m{\linterpretations{\Sigma}{V}} and requires showing \m{\imodel{\I}{\phi}\subseteq\imodel{\I}{\psi}}.
Proving soundness of an equivalence axiom \m{\phi\lbisubjunct\psi} requires showing \m{\imodel{\I}{\phi}=\imodel{\I}{\psi}}.
Proving soundness of a rule
\[
\linferenceRule{\phi}{\psi}
\]
assumes that premise $\phi$ is valid, i.e.\ \m{\imodel{\I}{\phi}=\linterpretations{\Sigma}{V}} in all interpretations $\iget[const]{\I}$ with a set of states \m{\linterpretations{\Sigma}{V}}, and requires showing that conclusion $\psi$ is valid, i.e.\ \m{\imodel{\I}{\psi}=\linterpretations{\Sigma}{V}} in all $\iget[const]{\I}$ with \m{\linterpretations{\Sigma}{V}}.
All proof rules of \dGL except \irref{US} satisfy the stronger condition of \emph{local soundness}, i.e.\ for all interpretations $\iget[const]{\I}$ with a set of states \m{\linterpretations{\Sigma}{V}}: \m{\imodel{\I}{\phi}=\linterpretations{\Sigma}{V}} implies \m{\imodel{\I}{\psi}=\linterpretations{\Sigma}{V}}.
For the proof, recall the $\mu$-calculus notation where \m{\lfp{Z}{\mapply{\Upsilon}{Z}}} denotes the least fixpoint of \m{\mapply{\Upsilon}{Z}} and \m{\gfp{Z}{\mapply{\Upsilon}{Z}}} denotes the greatest fixpoint.

Soundness of modus ponens (\irref{MP}) and $\forall$-generalization (from $\phi$ derive \(\lforall{x}{\phi}\)) is standard and not shown.
The other axioms and rules are proved to be sound subsequently.
\begin{enumerate}
\item[\irref{box}] \(\imodel{\I}{\dbox{\alpha}{\phi}} = \imodel{\I}{\lnot\ddiamond{\alpha}{\lnot\phi}}\) is a corollary to determinacy (\rref{thm:dGL-determined}).
\item[\irref{assignd}] \(\imodel{\I}{\ddiamond{\pupdate{\pumod{x}{\theta}}}{\mapply[x]{\phi}{x}}}
= \strategyfor[\pupdate{\pumod{x}{\theta}}]{\imodel{\I}{\mapply[x]{\phi}{x}}} = \{\iportray{\I} \in \linterpretations{\Sigma}{V} \with \modif{\iget[state]{\I}}{x}{\ivaluation{\I}{\theta}} \in \imodel{\I}{\mapply[x]{\phi}{x}}\}
= \{\iportray{\I} \in \linterpretations{\Sigma}{V} \with \iportray{\I} \in \imodel{\I}{\mapply[x]{\phi}{\theta}}\}
= \imodel{\I}{\mapply[x]{\phi}{\theta}}\),
where the penultimate equation holds by the substitution lemma.
The classical substitution lemma is sufficient for first-order logic $\mapply[x]{\phi}{\theta}$.
Otherwise the proof of the substitution lemma for \dL \cite[Lemma 2.2]{Platzer10} generalizes to \dGL or follows from uniform substitution lemmas \cite{DBLP:conf/cade/Platzer15}.

\item[\irref{evolved}]
\(\imodel{\I}{\ddiamond{\pevolve{\D{x}=\genDE{x}}}{\phi}} 
= \strategyfor[\pevolve{\D{x}=\genDE{x}}]{\imodel{\I}{\phi}} = \{\varphi(0) \in \linterpretations{\Sigma}{V} \with\) 
     for some
      \(\varphi{:}[0,r]\to\linterpretations{\Sigma}{V}\)
      so that $\varphi(r)\in\imodel{\I}{\phi}$ and
      \(\D[t]{\,\varphi(t)(x)} (\zeta) =       %
      \ivaluation{\iconcat[state=\varphi(\zeta)]{\I}}{\theta}
      \)
      for all $\zeta\leq r\}$.
Also,
\newcommand{\Id}{\imodif[state]{\I}{t}{r}}%
\(\imodel{\I}{\lexists{t{\geq}0}{\ddiamond{\pupdate{\pumod{x}{\solf(t)}}}{\phi}}}\)
\(= \{\iportray{\I} \in \linterpretations{\Sigma}{V} \with \modif{\iget[state]{\I}}{t}{r} \in \imodel{\I}{\ddiamond{\pupdate{\pumod{x}{\solf(t)}}}{\phi}} \text{ for some } r\geq0\}
= \{\iportray{\I} \in \linterpretations{\Sigma}{V} \with \modif{\iget[state]{\I}}{t}{r} \in
 \{\iportray{\Iu}\in\linterpretations{\Sigma}{V} \with \modif{\iget[state]{\Iu}}{x}{\ivaluation{\Iu}{\solf(t)}} \in \imodel{\I}{\phi}\} 
 \text{ for } r\geq0\}
=\)\\\(
 \{\iportray{\I} \in \linterpretations{\Sigma}{V} \with \modif{(\iget[state]{\Id})}{x}{\ivaluation{\Id}{\solf(t)}} \in \imodel{\I}{\phi}\
 \text{ for some } r\geq0\}\).
\renewcommand{\Id}{\imodif[state]{\I}{t}{\zeta}}%
 The inclusion ``$\supseteq$'' between both parts holds, because the function
 \(\varphi(\zeta) := \modif{(\iget[state]{\Id})}{x}{\ivaluation{\Id}{\solf(t)}}\) solves the differential equation \m{\pevolve{\D{x}=\genDE{\theta}}} by assumption.
 The inclusion ``$\subseteq$'' follows, because the solution of the (smooth) differential equation \m{\pevolve{\D{x}=\genDE{\theta}}} is unique \cite[Lemma 2.1]{Platzer10}.

\item[\irref{testd}] \(\imodel{\I}{\ddiamond{\ptest{\ivr}}{\phi}} 
= \strategyfor[\ptest{\ivr}]{\imodel{\I}{\phi}}
= \imodel{\I}{\ivr}\cap\imodel{\I}{\phi}
= \imodel{\I}{\ivr\land\phi}\)

\item[\irref{choiced}] \(\imodel{\I}{\ddiamond{\pchoice{\alpha}{\beta}}{\phi}} 
= \strategyfor[\pchoice{\alpha}{\beta}]{\imodel{\I}{\phi}}
= \strategyfor[\alpha]{\imodel{\I}{\phi}} \cup \strategyfor[\beta]{\imodel{\I}{\phi}}
=  \imodel{\I}{\ddiamond{\alpha}{\phi}} \cup \imodel{\I}{\ddiamond{\beta}{\phi}}
=  \imodel{\I}{\ddiamond{\alpha}{\phi} \lor \ddiamond{\beta}{\phi}}\)

\item[\irref{composed}] \(\imodel{\I}{\ddiamond{\alpha;\beta}{\phi}} 
= \strategyfor[\alpha;\beta]{\imodel{\I}{\phi}} 
= \strategyfor[\alpha]{\strategyfor[\beta]{\imodel{\I}{\phi}}}
= \strategyfor[\alpha]{\imodel{\I}{\ddiamond{\beta}{\phi}}} 
= \imodel{\I}{\ddiamond{\alpha}{\ddiamond{\beta}{\phi}}}\).

\item[\irref{iterated}]
Since \(\imodel{\I}{\ddiamond{\prepeat{\alpha}}{\phi}}
= \strategyfor[\prepeat{\alpha}]{\imodel{\I}{\phi}}
= \lfp{Z}{(\imodel{\I}{\phi}\cup\strategyfor[\alpha]{Z})}\)
is a fixpoint, have
\(\imodel{\I}{\ddiamond{\prepeat{\alpha}}{\phi}} =  \imodel{\I}{\phi}\cup\strategyfor[\alpha]{\imodel{\I}{\ddiamond{\prepeat{\alpha}}{\phi}}}\).
Thus,
\(\imodel{\I}{\phi \lor \ddiamond{\alpha}{\ddiamond{\prepeat{\alpha}}{\phi}}}
= \imodel{\I}{\phi} \cup \imodel{\I}{\ddiamond{\alpha}{\ddiamond{\prepeat{\alpha}}{\phi}}} 
= \imodel{\I}{\phi} \cup \strategyfor[\alpha]{\imodel{\I}{\ddiamond{\prepeat{\alpha}}{\phi}}}
= \imodel{\I}{\ddiamond{\prepeat{\alpha}}{\phi}}\).
Consequently,
\(\imodel{\I}{\phi \lor \ddiamond{\alpha}{\ddiamond{\prepeat{\alpha}}{\phi}}}
\subseteq \imodel{\I}{\ddiamond{\prepeat{\alpha}}{\phi}}\).

\item[\irref{duald}] \(\imodel{\I}{\ddiamond{\pdual{\alpha}}{\phi}}
= \strategyfor[\pdual{\alpha}]{\imodel{\I}{\phi}}
= \scomplement{\strategyfor[\alpha]{\scomplement{(\imodel{\I}{\phi})}}}
= \scomplement{\strategyfor[\alpha]{\imodel{\I}{\lnot\phi}}}
= \scomplement{(\imodel{\I}{\ddiamond{\alpha}{\lnot\phi}})}
= \imodel{\I}{\lnot\ddiamond{\alpha}{\lnot\phi}}\)
by \rref{def:HG-semantics}.

\item[\irref{M}]
Assume the premise \m{\phi\limply\psi} is valid in interpretation $\iget[const]{\I}$, i.e.\ \m{\imodel{\I}{\phi}\subseteq\imodel{\I}{\psi}}.
Then the conclusion \m{\ddiamond{\alpha}{\phi}\limply\ddiamond{\alpha}{\psi}} is valid in $\iget[const]{\I}$, i.e.\
\(\imodel{\I}{\ddiamond{\alpha}{\phi}} = \strategyfor[\alpha]{\imodel{\I}{\phi}} \subseteq \strategyfor[\alpha]{\imodel{\I}{\psi}} = \imodel{\I}{\ddiamond{\alpha}{\psi}}\) by monotonicity (\rref{lem:monotone}).

\item[\irref{FP}]
Assume the premise \m{\phi\lor\ddiamond{\alpha}{\psi}\limply\psi} is valid in $\iget[const]{\I}$, i.e.\ \(\imodel{\I}{\phi\lor\ddiamond{\alpha}{\psi}}\subseteq\imodel{\I}{\psi}\).
That is,
\(\imodel{\I}{\phi} \cup \strategyfor[\alpha]{\imodel{\I}{\psi}}
= \imodel{\I}{\phi}\cup\imodel{\I}{\ddiamond{\alpha}{\psi}}
= \imodel{\I}{\phi\lor\ddiamond{\alpha}{\psi}}
\subseteq \imodel{\I}{\psi}\).
Thus, $\psi$ is a pre-fixpoint of \(Z=\imodel{\I}{\phi}\cup\strategyfor[\alpha]{Z}\).
Now using \rref{lem:monotone}, 
\(\imodel{\I}{\ddiamond{\prepeat{\alpha}}{\phi}}
= \strategyfor[\prepeat{\alpha}]{\imodel{\I}{\phi}}
= \lfp{Z}{(\imodel{\I}{\phi}\cup\strategyfor[\alpha]{Z})}\)
is the least fixpoint and the least pre-fixpoint. %
Thus,
\(\imodel{\I}{\ddiamond{\prepeat{\alpha}}{\phi}} \subseteq \imodel{\I}{\psi}\),
which implies that \(\ddiamond{\prepeat{\alpha}}{\phi} \limply \psi\) is valid in $\iget[const]{\I}$.

\item[\irref{US}]
Standard soundness proofs for \irref{US} \cite{Church_1956} generalize to \dGL.
A new proof based on an elegant use of the soundness of \irref{RE} is shown here.
Assume the premise \m{\preusubst[\phi]{p}} is valid, i.e.\  \m{\imodel{\I}{\preusubst[\phi]{p}}=\linterpretations{\Sigma}{V}} in all interpretations $\iget[const]{\I}$ with a set of states \m{\linterpretations{\Sigma}{V}}.
Assume that the uniform substitution is admissible, otherwise rule \irref{US} is not applicable and there is nothing to show.
It needs to be shown that \m{\usubst[\phi]{p}{\psi}} is valid, i.e.\ \m{\imodel{\I}{\usubst[\phi]{p}{\psi}}=\linterpretations{\Sigma}{V}} for all $\iget[const]{\I}$ with \m{\linterpretations{\Sigma}{V}}.
Consider any particular interpretation $\iget[const]{\J}$ with set of states \m{\linterpretations{\Sigma}{V}}.
Without loss of generality, assume $p$ not to occur in $\mapply{\psi}{\cdot}$ (otherwise first replace all occurrences of $p$ in $\mapply{\psi}{\cdot}$ by $q$ and then use rule \irref{US} again to replace those $q$ by $p$).
Thus, by uniform substitution, $p$ does not occur in \m{\usubst[\phi]{p}{\psi}} and the value of $\iget[const]{\J}(p)$ is immaterial for the semantics of \m{\usubst[\phi]{p}{\psi}}.
Therefore, pass to an interpretation $\iget[const]{\I}$ that modifies $\iget[const]{\J}$ by changing the semantics of $p$ such that \m{\imodel{\I}{p(x)}=\imodel{\J}{\mapply{\psi}{x}}} for all values of $x$.
In particular, \m{\imodel{\I}{p(x)}=\imodel{\I}{\mapply{\psi}{x}}} for all values of $x$, since $p$ does not occur in \m{\mapply{\psi}{x}}.
Thus, \m{\iget[const]{\I}\models{\lforall{x}{(p(x)\lbisubjunct\mapply{\psi}{x})}}}.
Since \irref{M} is locally sound, so is the congruence rule \irref{RE}, which derives from \irref{M}.
The principle of substitution of equivalents \cite[Chapter 13]{HughesCresswell96} (from \m{A\lbisubjunct B} derive \m{\mapply{\Upsilon}{A} \lbisubjunct \mapply{\Upsilon}{B}}, where $\mapply{\Upsilon}{B}$ is the formula $\mapply{\Upsilon}{A}$ with some occurrences of $A$ replaced by $B$), thus, generalizes to \dGL and is locally sound. 
Hence, for all particular occurrences of $p(u)$ in $\preusubst[\phi]{p}$, have \m{\iget[const]{\I}\models{p(u)\lbisubjunct\mapply{\psi}{u}}}, which implies \m{\iget[const]{\I}\models{\preusubst[\phi]{p}\lbisubjunct\subst[\phi]{p(u)}{\mapply{\psi}{u}}}} for the ordinary replacement of $p(u)$ by $\mapply{\psi}{u}$.
This process can be repeated for all occurrences of $p(u)$,
leading to \m{\iget[const]{\I}\models{\preusubst[\phi]{p}\lbisubjunct\usubst[\phi]{p}{\psi}}}.
Thus, \(\linterpretations{\Sigma}{V}=\imodel{\I}{\preusubst[\phi]{p}}=\imodel{\I}{\usubst[\phi]{p}{\psi}}\).
Hence, \(\imodel{\J}{\usubst[\phi]{p}{\psi}}=\linterpretations{\Sigma}{V}\), because $p$ no longer occurs after uniform substitution \m{\usubst[\phi]{p}{\psi}}, since all occurrences of $p$ with any arguments will have been replaced at some point (since admissible).
This implies that \m{\usubst[\phi]{p}{\psi}} is valid since interpretation $\iget[const]{\J}$ with set of states \m{\linterpretations{\Sigma}{V}} was arbitrary.
\end{enumerate}
This concludes the soundness proofs for all axioms and proof rules of the \dGL proof calculus, which is, thus, sound.
\qedhere
\endproofmove

The proof calculus in \rref{fig:dGL} does not handle differential equations \m{\pevolvein{\D{x}=\genDE{x}}{\ivr}} with evolution domain constraints $\ivr$ (other than $\ltrue$).
Yet, \rref{lem:pevolvein} from \rref{sec:HG-equivalences} eliminates all evolution domain constraints equivalently from hybrid games, so that evolution domains no longer occur after this equivalence transformation.

\subsection{Completeness} \label{sec:dGL-complete}
\newcommand{\LBase}{\textit{L}\xspace}%

The converse of soundness is completeness, which is the question whether all valid formulas are provable.
Completeness of \dGL is a challenging question related to a famous open problem about completeness of propositional game logic \cite{DBLP:conf/focs/Parikh83}.
Based on G\"odel's second incompleteness theorem \cite{Goedel_1931}, \dL is incomplete \cite[Theorem 2]{DBLP:journals/jar/Platzer08} and so is \dGL.
Hence, the right question to ask is that of \emph{relative completeness} \cite{DBLP:journals/siamcomp/Cook78,DBLP:conf/stoc/HarelMP77}, i.e.\ completeness relative to an oracle logic \LBase.
Relative completeness studies the question whether a proof calculus has all proof rules that are required for proving all valid formulas in the logic from tautologies in \LBase.
Using a notion similar to Leivant's \citeyear{DBLP:conf/fossacs/Leivant09}, the question of relative completeness can be separated from that of expressivity.
Relative completeness can be shown \emph{schematically} for \dGL, i.e.\ the \dGL calculus is complete relative to \emph{any} expressive logic.
This is to be contrasted with \dL, whose relative completeness proof was dependent on the particular base logic and the specifics of its encoding \cite{DBLP:journals/jar/Platzer08}.
In particular, the \dGL completeness result is coding-free \cite{Moschovakis74}, i.e.\ independent of the particular encoding.
It only depends on the ability to express formulas.

\begin{definition}[Expressive]%
  \label{def:expressive}%
  A logic \LBase is \emph{expressive} (for \dGL) if, for each \dGL formula~\m{\phi\ignore{ \in \lformulas{\Sigma}{V}}}
  there is a formula~\m{\reduct{\phi}\ignore{ \in \lformulas[\LBase]{\Sigma}{V}}} of \LBase
  that is equivalent, i.e.\ \m{\entails{\phi \lbisubjunct \reduct{\phi}}}.
  Logic \LBase is \emph{constructively expressive} if, in addition, the mapping \(\phi\mapsto\reduct{\phi}\) is effective.
  The logic \LBase is \emph{differentially expressive} for a given proof calculus if \LBase is expressive and all equivalences of the form
  \m{\ddiamond{\pevolve{\D{x}=\genDE{x}}}{G} \lbisubjunct \reduct{(\ddiamond{\pevolve{\D{x}=\genDE{x}}}{G})}}
  are provable in that calculus.
  The logic \LBase is assumed to be closed under the connectives of first-order logic.\footnote{
  \label{foot:dynamic-quantifier-abbreviation}
  Alternatively, the equivalence \m{\lexists{x}{\phi} \mequiv \ddiamond{\pevolve{\D{x}=1}}{\phi} \lor \ddiamond{\pevolve{\D{x}=-1}}{\phi}} can be used to consider quantifiers as abbreviations in differentially expressive logics \LBase.
}
\end{definition}
Differential expressiveness ensures that the expressive logic \LBase is equipped with proof rules for concluding properties of differential equations from their equivalent expressions in \LBase.
The differential equation axiom \irref{evolved} is available for that purpose but limited to expressible solutions.
More general ways of concluding properties of differential equations for differential expressiveness include differential invariants and differential cuts \cite{DBLP:journals/logcom/Platzer10}, differential ghosts \cite{DBLP:journals/lmcs/Platzer12}, and the Euler axiom \cite{DBLP:conf/lics/Platzer12b}.
Concrete examples of differentially expressive logics are developed in \rref{sec:differentially-expressive} after proving completeness schematically relative to any arbitrary differentially expressive logic.

The classical approach for completeness proofs \cite{DBLP:journals/siamcomp/Cook78,DBLP:conf/stoc/HarelMP77} proceeds in stages of first-order safety assertions, first-order termination assertions, and then the repeated use of those to prove the general case.
That approach does not work for \dGL, because hybrid games are so highly symmetric that they may contain operators whose proof depends on proofs about all other operators.
A proof of \(F\limply\ddiamond{\alpha}{G}\), for example, may require proofs of formulas of the form \(A\limply\dbox{\beta}{B}\), e.g., when $\alpha$ is $\pdual{\beta}$.
Such an attempt of proving completeness for $\ddiamond{\alpha}{}$ formulas would need to assume completeness for $\dbox{\beta}{}$ formulas and vice versa, which is a cyclic assumption.
Even more involved cyclic arguments result from trying to prove completeness of $\ddiamond{\prepeat{\alpha}}{}$ and $\dbox{\prepeat{\alpha}}{}$ formulas that way.
Furthermore, the previous arguments for completeness of $\ddiamond{\prepeat{\alpha}}{}$ formulas  \cite{DBLP:journals/siamcomp/Cook78,DBLP:conf/stoc/HarelMP77,DBLP:journals/jar/Platzer08} depend on proofs about repetition counts.
Those do not work in a hybrid game setting, either, because guaranteed repetition bounds for winning repetition games can be recursively transfinite (\rref{thm:dGL-closure-lower}).
Also compare how the semantical discrepancies discussed in \rref{app:alternative-semantics} relate to repetition bounds.

Instead, completeness for all \dGL formulas of all kinds can be proved simultaneously, yet with a more involved well-founded partial order on formulas that ensures that the inductive argument in the completeness proof stays well-founded.
This generality has beneficial side-effects, though, because the resulting proof architecture enables a result with minimal coding that makes it possible to exactly identify all complex cases. 

\begin{theorem}[Relative completeness] \label{thm:dGL-complete}%
  \index{complete!relatively!dGL@\dGL}%
  The \dGL calculus is a \emph{sound and complete axiomatization} of hybrid games relative to \emph{any} differentially expressive logic \LBase, i.e.\
  every valid \dGL formula is provable in the \dGL calculus from \LBase tautologies.
\end{theorem}
\proofmove
\let\Oracle\LBase%
Write \m{\infers[\Oracle] \phi} to indicate that \dGL formula $\phi$ can be derived in the \dGL proof calculus from valid \LBase formulas.
It takes a moment's thought to conclude that soundness transfers to this case from \rref{thm:dGL-sound}, so it remains to prove completeness.
  For every valid \dGL formula $\phi$ it has to be proved that $\phi$
  can be derived from valid \LBase tautologies within the \dGL calculus:
  from \m{\entails\phi} prove \m{\infers[\Oracle] \phi}.
  The proof proceeds as follows: By propositional
  recombination, inductively identify fragments of $\phi$ that correspond to
  \m{\phi_1 \limply \ddiamond{\alpha}{\phi_2}}
  or
  \m{\phi_1 \limply \dbox{\alpha}{\phi_2}} logically.
  Find structurally simpler formulas from which these Angel or Demon properties can be derived taking care that the resulting formulas are simpler than the original one in a well-founded order.
  Finally, prove that the original \dGL formula can be re-derived from the subproofs in the \dGL calculus.

  By appropriate propositional derivations, assume $\phi$ to be given in conjunctive normal form.
  Assume that negations are pushed inside over modalities using the dualities
  \m{\lnot\dbox{\alpha}{\phi} \mequiv \ddiamond{\alpha}{\lnot\phi}}
  and
  \m{\lnot\ddiamond{\alpha}{\phi} \mequiv \dbox{\alpha}{\lnot\phi}} that are provable by axiom \irref{box},
  and that negations are pushed inside over quantifiers using provable first-order equivalences
  \m{\lnot\lforall{x}{\phi} \mequiv \lexists{x}{\lnot\phi}}
  and
  \m{\lnot\lexists{x}{\phi} \mequiv \lforall{x}{\lnot\phi}}.
  The remainder of the proof follows an induction on a well-founded partial order~$\prec$ induced on \dGL formulas by the lexicographic ordering of the overall structural complexity of the hybrid games in the formula and the structural complexity of the formula itself, with the logic \LBase placed at the bottom of the partial order~$\prec$.
  The base logic \LBase is considered of lowest complexity by relativity, because \m{\entails F} immediately implies \m{\infers[\Oracle] F} for all formulas $F$ of \LBase.
  Well-foundedness of $\prec$ follows (formally from projections into concatenations of finite trees), because the overall structural complexity of hybrid games in any particular formula can only decrease finitely often at the expense of increasing the formula complexity, which can, in turn, only decrease finitely often to result in a formula in $\LBase$.
  The only important property is that, if the structure of the hybrid games in $\psi$ is simpler than those in $\phi$ (somewhere simpler and nowhere worse), then $\psi\prec\phi$ even if the logical formula structure of $\psi$ is larger than that of $\phi$, e.g., when $\psi$ has more propositional connectives, quantifiers or modalities (but of smaller overall complexity hybrid games).
  In the following, \emph{IH} is short for induction hypothesis.
  The proof follows the syntactic structure of \dGL formulas.
  \begin{enumerate}
  \addtocounter{enumi}{-1}%
  \item \label{case:dGL-complete-0}
    If $\phi$ has no hybrid games, then~$\phi$ is a first-order formula; hence provable by assumption (even decidable \cite{tarski_decisionalgebra51} if in first-order real arithmetic, i.e.\ no uninterpreted predicate symbols occur).
  
  \item $\phi$ is of the form \m{\lnot\phi_1}; then~$\phi_1$ is first-order and quantifier-free, as negations are assumed to be pushed inside, so \rref{case:dGL-complete-0} applies.
  
  \item $\phi$ is of the form \m{\phi_1 \land \phi_2}, then \m{\entails\phi_1} and \m{\entails\phi_2}, so individually deduce simpler proofs for
   \m{\infers[\Oracle] \phi_1} and \m{\infers[\Oracle] \phi_2}
    by IH, which combine propositionally
    to a proof for \m{\infers[\Oracle] \phi_1\land\phi_2}.
  
  \item The case where $\phi$ is of the form \m{\lexists{x}{\phi_2}}, \m{\lforall{x}{\phi_2}}, \m{\ddiamond{\alpha}{\phi_2}} or \m{\dbox{\alpha}{\phi_2}} is included in \rref{case:dGL-complete-or} with \(\phi_1\mequiv\lfalse\).

  \item \label{case:dGL-complete-or}
    \newcommand{\precondf}{F}%
    \newcommand{\postcondf}{G}%
    $\phi$ is a disjunction and---without loss of generality---has one of the following forms
    (otherwise use provable associativity and commutativity to reorder disjunction):
    \[
    \begin{array}{r@{~}c@{~}l}
      \phi_1 &\lor& \ddiamond{\alpha}{\phi_2}\\
      \phi_1 &\lor& \dbox{\alpha}{\phi_2}\\
      \phi_1 &\lor& \lexists{x}{\phi_2}\\
      \phi_1 &\lor& \lforall{x}{\phi_2}
      .
    \end{array}
    \]
    Let \m{\phi_1 \lor \dmodality{\alpha}{\phi_2}} be a unified notation for those cases.
    Then, \m{\phi_2\prec\phi},
    since~\(\phi_2\) has less modalities or quantifiers.
    Likewise, \m{\phi_1\prec\phi} because
    \m{\dmodality{\alpha}{\phi_2}} contributes one modality or quantifier to
    $\phi$ that is not part of $\phi_1$.
    When abbreviating the simpler formulas $\lnot\phi_1$ by $\precondf$ and $\phi_2$ by $\postcondf$, the validity \m{\entails \phi} yields \m{\entails \lnot \precondf \lor \dmodality{\alpha}{\postcondf}}, so \m{\entails \precondf \limply \dmodality{\alpha}{\postcondf}}, from which the remainder of the proof inductively derives
    \begin{equation} \label{eq:derivable-triple}
      \infers[\Oracle]{\precondf \limply \dmodality{\alpha}{\postcondf}}
      .
    \end{equation}
    
    The proof of \rref{eq:derivable-triple} is by induction on the syntactic structure of $\dmodality{\alpha}{}$.

  \begin{enumerate}
  \item If $\dmodality{\alpha}{}$ is the operator $\lforall{x}{}$ then \m{\entails \precondf \limply \lforall{x}{\postcondf}}, where $x$ can be assumed not to occur in $\precondf$ by renaming.
  Hence, \m{\entails \precondf \limply \postcondf}.
  Since \(\postcondf \prec \lforall{x}{\postcondf}\), 
  because it has less quantifiers, also \((\precondf\limply \postcondf) \prec (\precondf\limply\lforall{x}{\postcondf})\),
  hence \(\infers[\Oracle] \precondf \limply \postcondf\) is derivable by IH.
  Then, \(\infers[\Oracle] \precondf \limply \lforall{x}{\postcondf}\) derives by $\forall$-generalization of first-order logic, since $x$ does not occur in $\precondf$.
  It is even decidable if in first-order real arithmetic \cite{tarski_decisionalgebra51}.
  The remainder of the proof concludes
  \((\precondf\limply \psi) \prec (\precondf\limply \phi)\) from \(\psi \prec \phi\) without further notice.

  \item If $\dmodality{\alpha}{}$ is the operator $\lexists{x}{}$ then \m{\entails \precondf \limply \lexists{x}{\postcondf}}.
  If $\precondf$ and $\postcondf$ are \LBase formulas, then, since \LBase is closed under first-order connectives, so is the valid formula \m{\precondf \limply \lexists{x}{\postcondf}}, which is, then, provable by IH and even decidable if in first-order real arithmetic \cite{tarski_decisionalgebra51}.
  
  Otherwise, $\precondf,\postcondf$ correspond to \LBase formulas by expressiveness of \LBase (\rref{def:expressive}), which implies the existence of an \LBase formula $\reduct{\postcondf}$ such that
  \(\entails \reduct{\postcondf} \lbisubjunct \postcondf\).
  Since \LBase is closed under first-order connectives, the valid formula \m{\precondf \limply \lexists{x}{(\reduct{\postcondf})}} is provable by IH, because
  \((\precondf \limply \lexists{x}{(\reduct{\postcondf})}) \prec (\precondf \limply \lexists{x}{\postcondf})\) since $\reduct{\postcondf}\in\LBase$ while $\postcondf\not\in\LBase$.
  Now, \(\entails \reduct{\postcondf} \lbisubjunct \postcondf\)
  implies \(\entails \reduct{\postcondf} \limply \postcondf\), which is derivable by IH, because
  \((\reduct{\postcondf} \limply \postcondf) \prec \phi\) since $\reduct{\postcondf}$ is in \LBase.
  From \(\infers[\Oracle] \reduct{\postcondf} \limply \postcondf\), the derivable dual of $\forall$-generalization derives
  \(\infers[\Oracle] \lexists{x}{(\reduct{\postcondf})} \limply \lexists{x}{\postcondf}\),
  which combines with
  \(\infers[\Oracle] \precondf \limply \lexists{x}{(\reduct{\postcondf})}\)
  by modus ponens to
  \(\infers[\Oracle] \precondf \limply \lexists{x}{\postcondf}\).\footnote{
  Expressiveness could also render $\precondf$ and $\postcondf$ as $\reduct{\precondf},\reduct{\postcondf}$ into \LBase in this and other cases of this proof and finally come back to $\postcondf$ using rule \irref{M} instead of $\exists$-generalization.
  But the other cases have direct proofs.
  }

  \item \m{\entails \precondf \limply \ddiamond{\pevolve{\D{x}=\genDE{x}}}{\postcondf}}
  implies
  \m{\entails \precondf \limply \reduct{(\ddiamond{\pevolve{\D{x}=\genDE{x}}}{\postcondf})}},
  which is derivable by IH,
  because \((\precondf \limply\reduct{(\ddiamond{\pevolve{\D{x}=\genDE{x}}}{\postcondf})}) \prec \phi\)
  since \(\reduct{(\ddiamond{\pevolve{\D{x}=\genDE{x}}}{\postcondf})}\) is in \LBase.
  Since \LBase is differentially expressive,
  \(\infers[\Oracle] \ddiamond{\pevolve{\D{x}=\genDE{x}}}{\postcondf} \lbisubjunct \reduct{(\ddiamond{\pevolve{\D{x}=\genDE{x}}}{\postcondf})}\)
  is provable, so
  \m{\infers[\Oracle] \precondf \limply \ddiamond{\pevolve{\D{x}=\genDE{x}}}{\postcondf}}
  derives from
  \m{\infers[\Oracle] \precondf \limply \reduct{(\ddiamond{\pevolve{\D{x}=\genDE{x}}}{\postcondf})}} by modus ponens.

  \item \m{\entails \precondf \limply \dbox{\pevolve{\D{x}=\genDE{x}}}{\postcondf}}
  implies
  \m{\entails \precondf \limply \lnot\ddiamond{\pevolve{\D{x}=\genDE{x}}}{\lnot \postcondf}}.
  Thus,
  \m{\entails \precondf \limply \lnot\reduct{(\ddiamond{\pevolve{\D{x}=\genDE{x}}}{\lnot \postcondf})}},
  which is derivable by IH,
  because \((\precondf \limply\lnot\reduct{(\ddiamond{\pevolve{\D{x}=\genDE{x}}}{\lnot \postcondf})}) \prec \phi\)
  since \(\reduct{(\ddiamond{\pevolve{\D{x}=\genDE{x}}}{\lnot \postcondf})}\) is in \LBase.
  Since \LBase is differentially expressive,
  \(\infers[\Oracle] \ddiamond{\pevolve{\D{x}=\genDE{x}}}{\lnot \postcondf} \lbisubjunct \reduct{(\ddiamond{\pevolve{\D{x}=\genDE{x}}}{\lnot \postcondf})}\)
  is provable, so
  \(\infers[\Oracle] \precondf \limply \lnot\ddiamond{\pevolve{\D{x}=\genDE{x}}}{\lnot \postcondf}\)
  derives from \(\infers[\Oracle] \precondf \limply \lnot\reduct{(\ddiamond{\pevolve{\D{x}=\genDE{x}}}{\lnot \postcondf})}\) by propositional congruence.
  Axiom \irref{box}, thus, derives
  \(\infers[\Oracle] \precondf \limply \dbox{\pevolve{\D{x}=\genDE{x}}}{\postcondf}\).

\item \m{\entails \precondf \limply \ddiamond{\pevolvein{\D{x}=\genDE{x}}{\ivr}}{\postcondf}},
then this formula is, by \rref{lem:pevolvein}, equivalent to a formula without evolution domain restrictions.
  Using equation \rref{eq:evolvein} from the proof of  \rref{lem:pevolvein} as a definitorial abbreviation concludes this case by IH.
  Similarly for \m{\entails \precondf \limply \dbox{\pevolvein{\D{x}=\genDE{x}}{\ivr}}{\postcondf}}.

  \item The cases where~$\alpha$ is of the form
    \m{\pupdate{\umod{x}{\theta}}}, \m{\ptest{\ivr}}, \m{\pchoice{\beta}{\gamma}}, or \m{\beta;\gamma}
    are consequences of the soundness of the equivalence axioms \irref{assignd+testd+choiced+composed} plus the duals obtained via the duality axiom \irref{box}.
    Whenever their respective left-hand side is valid, their right-hand side is valid and of smaller complexity (the games get simpler), and hence derivable by IH.
    Thus, \m{\precondf\limply\ddiamond{\alpha}{\postcondf}} derives by applying the respective axiom.
    This proof focuses on the $\ddiamond{}{}$ cases, because $\dbox{}{}$ cases derive by axiom \irref{box} from the $\ddiamond{}{}$ equivalences.

  \item \m{\entails \precondf \limply \ddiamond{\pupdate{\pumod{x}{\theta}}}{\postcondf}} implies \m{\entails \precondf \land y=\theta \limply \subst[\postcondf]{x}{y}} for a fresh variable $y$, where $\subst[\postcondf]{x}{y}$ is the result of substituting $y$ for $x$.
  Since \((\precondf \land y=\theta \limply \subst[\postcondf]{x}{y}) \prec \ddiamond{\pupdate{\pumod{x}{\theta}}}{\postcondf}\), because there are less hybrid games,
  \(\infers[\Oracle] \precondf \land y=\theta \limply \subst[\postcondf]{x}{y}\) is derivable by IH.
  Hence, \irref{assignd} derives
  \(\infers[\Oracle] \precondf \land y=\theta \limply \ddiamond{\pupdate{\pumod{x}{y}}}{\postcondf}\).
  Propositional logic derives
  \(\infers[\Oracle] \precondf \limply (y=\theta \limply \ddiamond{\pupdate{\pumod{x}{y}}}{\postcondf})\),
  from which 
  \(\infers[\Oracle] \precondf \limply \lforall{y}{(y=\theta \limply \ddiamond{\pupdate{\pumod{x}{y}}}{\postcondf})}\) derives by $\forall$-generalization of first-order logic as $y$ is not in $\precondf$.
  Since $y$ was fresh it does not appear in $\theta$ and $\postcondf$, so substitution validities of first-order logic derive
  \(\infers[\Oracle] \precondf \limply \ddiamond{\pupdate{\pumod{x}{\theta}}}{\postcondf}\).
  Note that direct proofs of \m{\precondf \limply \ddiamond{\pupdate{\pumod{x}{\theta}}}{\postcondf}} by \irref{assignd} are possible when the resulting substitution is admissible, but the substitution forming $\subst[\postcondf]{x}{y}$ is always admissible, because it is a variable renaming replacing $x$ by the fresh $y$.

  \item \m{\entails \precondf \limply \ddiamond{\ptest{\ivr}}{\postcondf}} implies \m{\entails \precondf \limply \ivr \land \postcondf}.
  Since \((\ivr\land \postcondf) \prec \ddiamond{\ptest{\ivr}}{\postcondf}\), because it has less modalities,
  \(\infers[\Oracle] \precondf \limply \ivr \land \postcondf\) is derivable by IH.
  Hence, \irref{testd} derives
  \(\infers[\Oracle] \precondf \limply \ddiamond{\ptest{\ivr}}{\postcondf}\)
  by propositional congruence, which is used without further notice subsequently.
  
  \item \m{\entails \precondf \limply \ddiamond{\pchoice{\beta}{\gamma}}{\postcondf}} implies \m{\entails \precondf \limply \ddiamond{\beta}{\postcondf} \lor \ddiamond{\gamma}{\postcondf}}.
  Since \(\ddiamond{\beta}{\postcondf}\lor\ddiamond{\gamma}{\postcondf} \prec \ddiamond{\pchoice{\beta}{\gamma}}{\postcondf}\), because, even if the propositional and modal structure increased, the structural complexity of both hybrid games $\beta$ and $\gamma$ is smaller than that of $\pchoice{\beta}{\gamma}$ (formula $\postcondf$ did not change),
  \(\infers[\Oracle] \precondf \limply \ddiamond{\beta}{\postcondf} \lor \ddiamond{\gamma}{\postcondf}\) is derivable by IH.
  Hence, \irref{choiced} derives \(\infers[\Oracle] \precondf \limply \ddiamond{\pchoice{\beta}{\gamma}}{\postcondf}\).
  
  \item \m{\entails \precondf \limply \ddiamond{\beta;\gamma}{\postcondf}}, which implies \m{\entails \precondf \limply \ddiamond{\beta}{\ddiamond{\gamma}{\postcondf}}}.
  Since \(\ddiamond{\beta}{\ddiamond{\gamma}{\postcondf}} \prec \ddiamond{\beta;\gamma}{\postcondf}\), because, even if the number of modalities increased, the overall structural complexity of the hybrid games decreased because there are less sequential compositions, 
  \(\infers[\Oracle] \precondf \limply \ddiamond{\beta}{\ddiamond{\gamma}{\postcondf}}\) is derivable by IH.
  Hence, \(\infers[\Oracle] \precondf \limply \ddiamond{\beta;\gamma}{\postcondf}\) derives by \irref{composed}.

\item \m{\entails \precondf \limply \dbox{\prepeat{\beta}}{\postcondf}}
    can be derived by induction as follows.
    Formula \m{\dbox{\prepeat{\beta}}{\postcondf}}, which expresses that Demon has a winning strategy in game~$\prepeat{\beta}$ to satisfy~$\postcondf$,
    is an inductive invariant of $\prepeat{\beta}$, because \m{\dbox{\prepeat{\beta}}{\postcondf}\limply\dbox{\beta}{\dbox{\prepeat{\beta}}{\postcondf}}} is valid, even provable by the variation \(\dbox{\prepeat{\beta}}{\postcondf}\limply \postcondf\land\dbox{\beta}{\dbox{\prepeat{\beta}}{\postcondf}}\) of \irref{iterated} that can be obtained from axioms \irref{iterated} and \irref{box}.
    Thus, its equivalent \LBase encoding from \rref{def:expressive} is also an inductive invariant:
    \[
    \inv \mequiv
    \reduct{(\dbox{\prepeat{\beta}}{\postcondf})}
    .
    \]
    Then \m{\precondf \limply \inv} and \m{\inv \limply \postcondf} are
    valid (Angel controls~$\prepeat{}$), so derivable by IH, since \((\precondf\limply \inv) \prec \phi\) and \((\inv\limply \postcondf) \prec \phi\) hold, because $\inv$ is in \LBase.
    By \irref{M}, \irref{duald} and \irref{box}, the latter derivation \m{\infers[\Oracle] {\inv \limply \postcondf}} extends to
    \m{\infers[\Oracle] {\dbox{\prepeat{\beta}}{\inv} \limply \dbox{\prepeat{\beta}}{\postcondf}}}.
    As above, \m{\inv \limply \dbox{\beta}{\inv}} is valid, and thus derivable by IH, since~$\beta$ has less loops than $\prepeat{\beta}$.
    Thus, \irref{invind}, which derives from \irref{FP} by \rref{lem:FP-invind}, derives
    \m{\infers[\Oracle]{\inv\limply\dbox{\prepeat{\beta}}{\inv}}}.
    The above derivations \m{\infers[\Oracle] \precondf \limply \inv}, \m{\infers[\Oracle]{\inv\limply\dbox{\prepeat{\beta}}{\inv}}}, and \m{\infers[\Oracle] {\dbox{\prepeat{\beta}}{\inv} \limply \dbox{\prepeat{\beta}}{\postcondf}}} combine by modus ponens
    to \m{\infers[\Oracle]{\precondf\limply\dbox{\prepeat{\beta}}{\postcondf}}}.

\item \m{\entails \precondf \limply \ddiamond{\prepeat{\beta}}{\postcondf}}.
   \def\vec#1{#1}%
    Let $\vec{x}$ the vector of free variables of \m{\ddiamond{\prepeat{\beta}}{\postcondf}}.
    Since $\ddiamond{\prepeat{\beta}}{\postcondf}$ is the least pre-fixpoint, for all \dGL formulas $\psi$ with free variables in $\vec{x}$:
    \[
    \entails \lforall{\vec{x}}{(\postcondf\lor\ddiamond{\beta}{\psi}\limply\psi)} \limply (\ddiamond{\prepeat{\beta}}{\postcondf} \limply\psi)
    \]
    by a variation of the soundness argument for \irref{FP}, which is also derivable by the (semantic) deduction theorem from \irref{FP}.
    In particular, this holds for a fresh predicate symbol $p$ with arguments $\vec{x}$:
    \[
    \entails \lforall{\vec{x}}{(\postcondf\lor\ddiamond{\beta}{p(\vec{x})}\limply p(\vec{x}))} \limply (\ddiamond{\prepeat{\beta}}{\postcondf} \limply p(\vec{x}))
    \]
    Using \m{\entails \precondf \limply \ddiamond{\prepeat{\beta}}{\postcondf}}, this implies
    \[
    \entails \lforall{\vec{x}}{(\postcondf\lor\ddiamond{\beta}{p(\vec{x})}\limply p(\vec{x}))} \limply (\precondf \limply p(\vec{x}))
    \]
    As \((\lforall{\vec{x}}{(\postcondf\lor\ddiamond{\beta}{p(\vec{x})}\limply p(\vec{x}))} \limply (\precondf \limply p(\vec{x}))) \prec \phi\), because, even if the formula complexity increased, the structural complexity of the hybrid games decreased, because $\phi$ has one more loop, this fact is derivable by IH:
    \[
    \infers[\Oracle] \lforall{\vec{x}}{(\postcondf\lor\ddiamond{\beta}{p(\vec{x})}\limply p(\vec{x}))} \limply (\precondf \limply p(\vec{x}))
    \]
    By uniformly substituting $\ddiamond{\prepeat{\beta}}{\postcondf}$, which has free variables $\vec{x}$, for $p(\vec{x})$, \irref{US} derives using $p\not\in \precondf,\postcondf, \beta$:
    \begin{equation}
    \infers[\Oracle] \lforall{\vec{x}}{(\postcondf\lor\ddiamond{\beta}{\ddiamond{\prepeat{\beta}}{\postcondf}}\limply \ddiamond{\prepeat{\beta}}{\postcondf})} \limply (\precondf \limply \ddiamond{\prepeat{\beta}}{\postcondf})
    \label{eq:F-implies-lfp}
    \end{equation}
    Yet, \irref{iterated} derives
    \(\infers \postcondf\lor\ddiamond{\beta}{\ddiamond{\prepeat{\beta}}{\postcondf}} \limply \ddiamond{\prepeat{\beta}}{\postcondf}\),
    from which
    \(\infers \lforall{\vec{x}}{(\postcondf\lor\ddiamond{\beta}{\ddiamond{\prepeat{\beta}}{\postcondf}} \limply \ddiamond{\prepeat{\beta}}{\postcondf})}\) derives by $\forall$-generalization.
    Now modus ponens derives
    \(\infers[\Oracle] \precondf \limply \ddiamond{\prepeat{\beta}}{\postcondf}\)
    by \rref{eq:F-implies-lfp}.

\item \m{\entails \precondf \limply \ddiamond{\pdual{\beta}}{\postcondf}} implies
   \m{\entails \precondf \limply \lnot\ddiamond{\beta}{\lnot \postcondf}}, which implies
   \m{\entails \precondf \limply \dbox{\beta}{\postcondf}}.
   Since \(\dbox{\beta}{\postcondf} \prec \ddiamond{\pdual{\beta}}{\postcondf}\), because $\pdual{\beta}$ is more complex than $\beta$ even if the modality changed,
   \m{\infers[\Oracle] \precondf \limply \dbox{\beta}{\postcondf}} can be derived by IH.
   Axiom \irref{box}, thus, derives
   \m{\infers[\Oracle] \precondf \limply \lnot\ddiamond{\beta}{\lnot \postcondf}},
   from which axiom \irref{duald} derives
   \m{\infers[\Oracle] \precondf \limply \ddiamond{\pdual{\beta}}{\postcondf}}.

\item \m{\entails \precondf \limply \dbox{\pdual{\beta}}{\postcondf}} implies
   \m{\entails \precondf \limply \lnot\ddiamond{\pdual{\beta}}{\lnot \postcondf}}, hence
   \m{\entails \precondf \limply \ddiamond{\beta}{\postcondf}}.
   Since \(\ddiamond{\beta}{\postcondf} \prec \dbox{\pdual{\beta}}{\postcondf}\), because $\pdual{\beta}$ is more complex than $\beta$ even if the modality changed,
   \m{\infers[\Oracle] \precondf \limply \ddiamond{\beta}{\postcondf}} can be derived by IH.
   Consequently, \m{\infers[\Oracle] \precondf \limply \lnot\lnot\ddiamond{\beta}{\lnot\lnot \postcondf}} can be derived using \irref{M} on \(\infers \postcondf \limply \lnot\lnot \postcondf\).
   Hence, \irref{duald} derives \m{\infers[\Oracle] \precondf \limply \lnot\ddiamond{\pdual{\beta}}{\lnot \postcondf}}, from which axiom \irref{box} derives
   \m{\infers[\Oracle] \precondf \limply \dbox{\pdual{\beta}}{\postcondf}}.

    \end{enumerate}
    \noindent
    This concludes the derivation of \rref{eq:derivable-triple}, because all operators \m{\dmodality{\alpha}{}} for the form \rref{eq:derivable-triple} have been considered.
    From \rref{eq:derivable-triple}, which is 
    \m{\infers[\Oracle] \lnot\phi_1 \limply \dmodality{\alpha}{\phi_2}} after resolving abbreviations,
    \m{\infers[\Oracle] \phi_1 \lor \dmodality{\alpha}{\phi_2}}
    derives propositionally.

  \end{enumerate}
This completes the proof of completeness (\rref{thm:dGL-complete}), because all syntactical forms of \dGL formulas have been covered.
\qedhere
\endproofmove
The proof of \rref{thm:dGL-complete} is constructive, so \rref{thm:dGL-complete} is constructive if \LBase is constructively expressive.
The proof is Moschovakis coding-free \cite{Moschovakis74}.
It even works entirely without coding, except for \m{\pevolve{\D{x}=\genDE{x}}}, $\exists$ and $\dbox{\prepeat{\beta}}{}$.
Using \irref{US}, the case for \(\ddiamond{\prepeat{\beta}}{G}\) in the proof of \rref{thm:dGL-complete} reveals an explicit $\reduct{}$-free reduction to a \dGL formula with less loops, which can be considered a modal analogue of characterizations in the Calculus of Constructions \cite{DBLP:journals/iandc/CoquandH88}.
Using \rref{thm:dGL-complete}, these observations easily reprove a classical result of Meyer and Halpern \cite{DBLP:journals/jacm/MeyerH82} about the semidecidability of termination assertions (logical formulas \(F\limply\ddiamond{\alpha}{G}\) of uninterpreted dynamic logic with first-order $F,G$ and regular programs $\alpha$ without differential equations).
In fact, this proves a stronger result about semidecidability of dynamic logic without any $\dbox{\alpha}{\cdot}$ with loops \cite{DBLP:journals/iandc/Schmitt84}.
\rref{thm:dGL-complete} shows that this result continues to hold for uninterpreted game logic in the fragment where $\prepeat{}$ only occurs with even $\pdual{}$-polarity in $\ddiamond{\alpha}{}$ and only of odd $\pdual{}$-polarity in $\dbox{\alpha}{}$ (the conditions on tests in $\alpha$ are accordingly).

The constructive nature of \rref{thm:dGL-complete} characterizes exactly which part of hybrid games proving is difficult:
finding computationally succinct weaker invariants %
for $\dbox{\prepeat{\beta}}{G}$ and finding succinct differential (in)variants \cite{DBLP:journals/logcom/Platzer10} for $\dbox{\pevolve{\D{x}=\genDE{x}}}{}$ and $\ddiamond{\pevolve{\D{x}=\genDE{x}}}{}$, of which a solution is a special case \cite{DBLP:journals/lmcs/Platzer12}.
The case $\lexists{x}{G}$ is interesting in that a closer inspection of \rref{thm:dGL-complete} reveals that its complexity depends on whether that quantifier supports Herbrand disjunctions. That is the case for uninterpreted first-order logic and first-order real arithmetic \cite{tarski_decisionalgebra51}, but not for $G\mequiv\dbox{\prepeat{\beta}}{\psi}$, which already gives $\lexists{x}{G}$ the full $\Pi^1_1$-complete complexity even for classical dynamic logic \cite[Theorems 13.1,13.2]{Harel_et_al_2000}.
Herbrand disjunctions for $\lexists{x}{G}$ justify how \rref{thm:dGL-complete} implies the result of Schmitt \cite{DBLP:journals/iandc/Schmitt84}.

The proof of \rref{thm:dGL-complete} uses \emph{minimal} coding.
The case $\dbox{\prepeat{\beta}}{}$ needs encoding, because \(F\limply\dbox{\prepeat{\beta}}{G}\) validity is already $\Pi^0_2$-complete for classical dynamic logic \cite[Theorem 13.5]{Harel_et_al_2000}.
The case $\exists{}{}$ needs encoding in the presence of \(\dbox{\prepeat{\beta}}{}\), because \(\lexists{x}{\dbox{\prepeat{\beta}}{G}}\) validity is $\Pi^1_1$-complete for classical dynamic logic \cite[Theorems 13.1]{Harel_et_al_2000}.
The case $\pevolve{\D{x}=\genDE{x}}$ leads to classical $\Delta^1_1$-hardness over $\naturals$ \cite[Lemma 4]{DBLP:journals/jar/Platzer08}.

The completeness proof indicates a coding-free way of proving Angel properties \(\ddiamond{\prepeat{\beta}}{G}\) that is similar to characterizations in the Calculus of Constructions and works in practice (\rref{app:ExampleProofs}). %
In particular, \dGL does not need Harel's convergence rule \cite{DBLP:conf/stoc/HarelMP77} for completeness and, thus, neither does logic for hybrid systems, even though it was previously based on it \cite{DBLP:conf/lics/Platzer12b}.
These results correspond to a hybrid game reading of influential views of understanding program invariants as fixpoints \cite{DBLP:conf/popl/CousotC77,Clarke79}.

\subsection{Expressibility} \label{sec:differentially-expressive}

The \dGL calculus is complete relative to any differentially expressive logic \LBase (\rref{thm:dGL-complete}).
One natural choice for an oracle logic is \muD, the modal $\mu$-calculus of differential equations (\emph{fixpoint logic of differential equations}):
\[
\phi \bebecomes X(\theta) \alternative p(\theta) \alternative \theta_1\geq\theta_2 \alternative \lnot \phi \alternative \phi \land \psi \alternative \ddiamond{\pevolve{\D{x}=\genDE{x}}}{\phi} \alternative \lfp{X}{\phi}
\]
where the least fixpoint $\lfp{X}{\phi}$ requires all occurrences of $X$ in $\phi$ to be positive.
The semantics is the usual, e.g., $\lfp{X}{\phi}$ binds set variable $X$ and real variable (vector) $x$ and is interpreted as the least fixpoint $X$ of $\phi$, i.e.\ the smallest denotation of $X$ such that \(X(x)\lbisubjunct\phi\) holds for all $x$ \cite{DBLP:journals/tcs/Kozen83,DBLP:conf/lics/Lubarsky89}.
A more careful inspection of the proofs in this article reveals that the two-variable fragment of \muD is enough, which gives a stronger statement as long as the variable hierarchy for \muD does not collapse \cite{DBLP:journals/mst/BerwangerGL07}.
The logic \muD is considered in this context, because it exposes the most natural interactivity on top of differential equations and makes the constructions most apparent and minimally coding themselves.

\begin{lemma}[Continuous expressibility] \label{lem:expressive/continuous}%
  \muD is constructively differentially expressive for \dGL.
\end{lemma}
\proofmove
Of course, \(\reduct{(p(\theta))}=p(\theta)\) etc.
The inductive cases are shown in \rref{fig:expressive}.
\begin{figure}[htb]
  \vspace{-\baselineskip}
  \begin{align*}
    \reduct{(\lnot\phi)} &\mequiv \lnot(\reduct{\phi})\\
    \reduct{(\phi\land\psi)} &\mequiv \reduct{\phi} \land \reduct{\psi}\\
    \reduct{(\lexists{x}{\phi})} &\mequiv \lexists{x}{(\reduct{\phi})}\\
    \reduct{(\ddiamond{\pupdate{\pumod{x}{\theta}}}{\phi})} &\mequiv \lforall{y}{(y=\theta\limply\reduct{(\subst[\phi]{x}{y})})}\\
    \reduct{(\ddiamond{\pevolve{\D{x}=\genDE{x}}}{\phi})} &\mequiv \ddiamond{\pevolve{\D{x}=\genDE{x}}}{(\reduct{\phi})}\\
    \reduct{(\ddiamond{\ptest{\ivr}}{\phi})} &\mequiv \reduct{(\ivr \land \phi)}\\
    \reduct{(\ddiamond{\pchoice{\alpha}{\beta}}{\phi})} &\mequiv \reduct{(\ddiamond{\alpha}{\phi} \lor \ddiamond{\beta}{\phi})}\\
    \reduct{(\ddiamond{\alpha;\beta}{\phi})} &\mequiv \reduct{(\ddiamond{\alpha}{\ddiamond{\beta}{\phi}})}\\
    \reduct{(\ddiamond{\prepeat{\alpha}}{\phi})} &\mequiv \lfp{X}{\reduct{(\phi\lor\ddiamond{\alpha}{X(x)})}}\\
    \reduct{(\ddiamond{\pdual{\alpha}}{\phi})} &\mequiv \reduct{(\lnot\ddiamond{\alpha}{\lnot\phi})}\\
    \reduct{(\dbox{\alpha}{\phi})} &\mequiv \reduct{(\ddiamond{\pdual{\alpha}}{\phi})}
  \end{align*}
  \caption{Inductive cases for constructive \emph{expressivity} of \muD.}
  \label{fig:expressive}%
\end{figure}
  It is easy to check that $\reduct{\phi}$ is equivalent to $\phi$, e.g. based on the soundness of the \dGL axioms.
  Note that
  \(\reduct{(\phi\lor\psi)} \mequiv \reduct{\phi} \lor \reduct{\psi}\) is a consequence of the above definitions and the abbreviation \(\phi\lor\psi \mequiv \lnot(\lnot\phi \land \lnot\psi)\).
  The quantifier in the definition of \(\reduct{(\ddiamond{\pupdate{\pumod{x}{\theta}}}{\phi})}\) is not necessary if the substitution of $\theta$ for $x$ is admissible. The variable renaming of fresh variable $y$ for $x$ in $\phi$ with the result $\subst[\phi]{x}{y}$ is always admissible.
  Quantifiers are expressible in \muD using \rref{foot:dynamic-quantifier-abbreviation}.
  Also \(\pevolvein{\D{x}=\genDE{x}}{\ivr}\) is expressible by \rref{lem:pevolvein}.
  The case \(\reduct{(\ddiamond{\prepeat{\alpha}}{\phi})}\) is defined as the least fixpoint of the reduction of \(\phi\lor\ddiamond{\alpha}{X(x)}\), where $x$ are the variables of $\alpha$ using classical short notation \cite{DBLP:conf/lics/Lubarsky89}.
  In particular, \(\reduct{(\ddiamond{\prepeat{\alpha}}{\phi})}\) satisfies
  \(\phi\lor\ddiamond{\alpha}{\reduct{(\ddiamond{\prepeat{\alpha}}{\phi})}} \lbisubjunct \reduct{(\ddiamond{\prepeat{\alpha}}{\phi})}\)
  and \(\reduct{(\ddiamond{\prepeat{\alpha}}{\phi})}\) is the formula with the smallest such interpretation, which is all that these proofs depend on.
  Finally, \muD is differentially expressive, because it includes all formulas of the form \m{\ddiamond{\pevolve{\D{x}=\genDE{x}}}{\phi}}.
\endproofmove

A discrete analog of \rref{lem:expressive/continuous} follows from a (constructive) equi-expressibility result \cite[Theorem 9]{DBLP:conf/lics/Platzer12b} using the \emph{Euler axiom}, which relates properties of differential equations to properties of their Euler discretizations \cite{DBLP:conf/lics/Platzer12b}.

\begin{corollary}[Discrete expressibility] \label{cor:expressive/discrete}%
  The (first-order) discrete $\mu$-calculus over $\reals$ is constructively differentially expressive for \dGL (with the Euler axiom).
\end{corollary}

This aligns the discrete and the continuous side of hybrid games in a constructive provably equivalent way similar to corresponding results about hybrid systems \cite{DBLP:conf/lics/Platzer12b}.
Yet, the interactivity of two-variable fixpoints stays, which turns out to be necessary (\rref{sec:dGL-dL-expressiveness}).
\begin{corollary}[Relative completeness]
  The \dGL calculus is a sound and complete axiomatization of \dGL relative to \muD.
  With the Euler axiom, the \dGL calculus is a sound and complete axiomatization of \dGL relative to the discrete $\mu$-calculus over $\reals$.
\end{corollary}
\begin{proof}
Follows from \rref{thm:dGL-complete}, \rref{lem:expressive/continuous}, and \rref{cor:expressive/discrete}.
\end{proof}

An interesting question is whether fragments of \dGL are complete relative to smaller logics, which \rref{thm:dGL-complete} and \rref{lem:expressive/continuous} reduce solely to a study of expressing (two-variable) \muD.
This yields the following hybrid versions of Parikh's completeness results for fragments of game logic \cite{DBLP:conf/focs/Parikh83}.

\begin{corollary}[Relative completeness of $\prepeat{}$-free \dGL]%
  \label{cor:dGL-bounded-complete}%
  The \dGL calculus is a sound and complete axiomatization of $\prepeat{}$-free hybrid games relative to \dL.
\end{corollary}
\proofmove
\rref{lem:expressive/continuous} reduces to \dL, even the first-order logic of differential equations \cite{DBLP:conf/lics/Platzer12b}, for $\prepeat{}$-free hybrid games.
\endproofmove
\begin{corollary}[Relative completeness of $\pdual{}$-free \dGL] \label{cor:dGL-noninteractive-complete}
  The \dGL calculus is a sound and complete axiomatization of $\pdual{}$-free hybrid games relative to \dL.
\end{corollary}
\proofmove
All $\pdual{}$-free loops are Scott-continuous by \rref{lem:HP-Scott-continuous}, so have closure ordinal $\omega$ and are, thus, equivalent to their \dL form, and even expressible in the first-order logic of differential equations by \cite[Theorem 9]{DBLP:conf/lics/Platzer12b}.
\endproofmove
By \rref{cor:dGL-noninteractive-complete}, \dL is relatively complete without the convergence rule that had been used before \cite{DBLP:journals/jar/Platzer08}.
In combination with the first and second relative completeness theorems of \dL \cite{DBLP:conf/lics/Platzer12b}, it follows that the \dGL calculus is a sound and complete axiomatization of $\prepeat{}$-free hybrid games and of $\pdual{}$-free hybrid games relative to the first-order logic of differential equations. When adding the Euler axiom \cite{DBLP:conf/lics/Platzer12b}, both are sound and complete axiomatizations of those classes of hybrid games relative to discrete dynamic logic \cite{DBLP:conf/lics/Platzer12b}.
Similar completeness results for \dGL relative to \dL, and, thus, relative to the first-order logic of differential equations, follow from \rref{thm:dGL-complete} with some more thought, e.g., for the case of hybrid games with winning regions that are finite rank Borel sets.

\subsection{Separating Axioms} \label{sec:separating-axioms}

In order to illustrate how and why reasoning about hybrid games differs from reasoning about hybrid systems, this section identifies separating axioms, i.e.\ axioms of \dL \cite{DBLP:journals/jar/Platzer08,DBLP:conf/lics/Platzer12b} that do not hold in \dGL.
The following result identifies the \emph{axiomatic separation}, i.e.\ all axioms differing in the respective complete axiomatizations of hybrid systems and hybrid games.
It investigates the difference in terms of important classes of modal logics; recall \cite{HughesCresswell96} or \rref{app:separating-axioms}.
\begin{theorem}[Axiomatic separation] \label{thm:separating-axioms}%
  The axiomatic separation of hybrid games compared to hybrid systems is exactly the Kripke axiom K, the loop induction axiom I, Harel's loop convergence axiom C, the Barcan axiom B, the vacuous axiom V, and the normal G\"odel generalization rule G.
  Hence, \dGL is a subregular, sub-Barcan, monotone modal logic without the loop induction loop convergence axioms and vacuity.
\end{theorem}
The proof of \rref{thm:separating-axioms} is in \rref{app:separating-axioms}, where a simple counterexample for each separating axiom illustrates what makes hybrid games different than hybrid systems.
The difference in axioms is summarized in \rref{fig:separating-axioms}, where $\closureall{}$ is the universal closure with respect to all variables bound in hybrid game $\alpha$.
Besides the axiomatic separation, \rref{fig:separating-axioms} shows additional related axioms or proof rules for illustration purposes.
\newcommand*{\cancel}[2][thick]{\tikz[baseline] \node [strike out,draw=red,anchor=text,inner sep=0pt,text=black,#1]{#2};}%
\begin{figure*}[tbh]
\vspace{-\baselineskip}
\[
\def\arraystretch{2}%
\begin{array}{@{}l@{~}l@{\hspace{0.5cm}}l@{~}l@{}}
\cancel{\text{K}} &
\linferenceRule[impl]
        {\dbox{\alpha}{(\phi\limply\psi)}}
        {(\dbox{\alpha}{\phi}\limply\dbox{\alpha}{\psi})}
&
\text{M}_{\dbox{\cdot}{}} &
\linferenceRule[formula]
        {\phi\limply\psi}
        {\dbox{\alpha}{\phi}\limply\dbox{\alpha}{\psi}}
\\
\cancel{$\overleftarrow{\text{M}}$}&
  \ddiamond{\alpha}{(\phi\lor\psi)}
  \limply
  \ddiamond{\alpha}{\phi} \lor \ddiamond{\alpha}{\psi}
  &
\text{M} & %
  \ddiamond{\alpha}{\phi} \lor \ddiamond{\alpha}{\psi} \limply \ddiamond{\alpha}{(\phi\lor\psi)}
\\
\cancel{\,\text{I}\,} &%
  \dbox{\prepeat{\alpha}}{(\phi\limply\dbox{\alpha}{\phi})} \limply (\phi\limply\dbox{\prepeat{\alpha}}{\phi})
&
\forall\text{I} &%
\closureall
{(\phi\limply\dbox{\alpha}{\phi})} \limply (\phi\limply\dbox{\prepeat{\alpha}}{\phi})
\\
\cancel{\text{C}} &%
  \linferenceRule[impll]
        {\dbox{\prepeat{\alpha}}{\lforall{v{>}0}{(\mapply{\var}{v}\limply\ddiamond{\alpha}{\mapply{\var}{v-1}})}}}
        {\lforall{v}{(\mapply{\var}{v} \limply
            \ddiamond{\prepeat{\alpha}}{\lexists{v{\leq}0}{\mapply{\var}{v}}})}}
      ~~({\m{v\not\in\alpha}})%
\\
\cancel{\text{B}} & %
\linferenceRule[impl]
        {\ddiamond{\alpha}{\lexists{x}{\phi}}}
        {\lexists{x}{\ddiamond{\alpha}{\phi}}}
      \qquad\hspace{1.5cm}({\m{x\not\in\alpha}})
&
\overleftarrow{\text{B}} & %
      {\linferenceRule[impl]
        {\lexists{x}{\ddiamond{\alpha}{\phi}}}
        {\ddiamond{\alpha}{\lexists{x}{\phi}}}
      }{\qquad(\m{x\not\in\alpha})}
\\
\cancel{\text{V}} & %
      {\linferenceRule[impl]
        {\phi}
        {\dbox{\alpha}{\phi}}
      }{\qquad\hspace{0.8cm}(\freevars{\phi}\cap \boundvars{\alpha}=\emptyset)}%
&
\text{VK} & %
      {\linferenceRule[impl]
        {\phi}
        {(\dbox{\alpha}{\ltrue}{\limply}\dbox{\alpha}{\phi})}
      }{~(\freevars{\phi}{\cap} \boundvars{\alpha}{=}\emptyset)}%
\\
\cancel{\text{G}} & %
      \linferenceRule[formula]
        {\phi}
        {\dbox{\alpha}{\phi}}
&
\text{M}_{\dbox{\cdot}{}} &
\linferenceRule[formula]
        {\phi\limply\psi}
        {\dbox{\alpha}{\phi}\limply\dbox{\alpha}{\psi}}
\\
\cancel{\text{R}} & %
      \linferenceRule[formula]
        {\phi_1\land\phi_2\limply\psi}
        {\dbox{\alpha}{\phi_1} \land \dbox{\alpha}{\phi_2} \limply \dbox{\alpha}{\psi}}
&
\text{M}_{\dbox{\cdot}{}} &
\linferenceRule[formula]
        {\phi_1\land\phi_2\limply\psi}
        {\dbox{\alpha}{(\phi_1\land\phi_2)}\limply\dbox{\alpha}{\psi}}
\\
\cancel{\text{FA}} & %
\ddiamond{\prepeat{\alpha}}{\phi} \limply \phi \lor \ddiamond{\prepeat{\alpha}}{(\lnot\phi\land\ddiamond{\alpha}{\phi})}
\end{array}
\]
  \caption{Separating axioms: The axioms and rules on the left are sound for hybrid systems but not for hybrid games. The related axioms and rules on the right are sound for hybrid games.}
  \label{fig:separating-axioms}
\end{figure*}

While explicit counterexamples proving the separation in \rref{thm:separating-axioms} are in \rref{app:separating-axioms}, the sequel explains the intuition for the difference causing unsoundness in hybrid games of the axioms identified in \rref{thm:separating-axioms}.
Kripke's K is unsound (for hybrid games):
even if Demon can play RoboCup so that his robots score a goal every time they pass the ball (just never try) and Demon can also play RoboCup so that his robots always pass the ball (somewhere), that does not mean Demon would have a strategy to always score goals.
The converse monotonicity axiom $\overleftarrow{\text{M}}$ is unsound:
just because Angel \W has a strategy to be close to \E or far away does not mean \W would either have a strategy to always end up close to E or a strategy that is always far away.
The induction axiom I is unsound:
just because Demon has a strategy for his RoboCup robots (e.g. power down) that, no matter how often $\prepeat{\alpha}$ repeats, Demon still has a strategy such that his robots do not run out of battery for just one more control cycle, that does not mean he has a strategy to keep his robots' batteries nonempty all the time.
Harel's convergence rule C is unsound:
even if Demon may have a strategy (e.g. waiting) such that after any number of rounds of the game, he has a strategy to move his robots closer to the goal for one control cycle, he still may not have a strategy to ultimately reach the goal, because that requires many rounds of guaranteed progress not just one.
The Barcan axiom B is unsound:
just because the winner of a RoboCup tournament can be chosen for $x$ after the robot game $\alpha$ does not mean it would be possible to predict this winner $x$ before the game $\alpha$.
By contrast, the converse Barcan axiom $\overleftarrow{\text{B}}$ is sound for hybrid games since, if known before the game $\alpha$, selecting the winner for $x$ can still be postponed until after the game $x$.
The vacuous axiom V, in which no free variable of $\phi$ is bound by $\alpha$, is unsound:
even if $\phi$ does not change its truth-value during $\alpha$ does not mean it would be possible for Demon to reach any final state at all without being tricked into violating the rules of the game along the way.
With an additional assumption (\m{\dbox{\alpha}{\ltrue}}) that Demon has a winning strategy to reach any final state at all (in which
$\ltrue$, i.e.\ true, holds which imposes no condition), the possible vacuous axiom VK is sound.
G\"odel's rule G is unsound:
even if $\phi$ holds in all states, Demon may still fail to win $\dbox{\alpha}{\phi}$ if he loses prematurely since Angel tricks Demon into violating the rules during the hybrid game $\alpha$.
Regularity rule R is unsound:
just because Demon's RoboCup robots have a strategy to focus the robots on strong defense and another strategy to, instead, focus them on strong offense that does not mean he would have a strategy to win RoboCup even if simultaneously strong defense and strong offense together might imply victory, because offensive and defensive strategies are in conflict.
First arrival FA is unsound:
just because Angel's robot has a strategy to ultimately capture Demon's faster robot with less battery does not mean she would either start with capture or would have a strategy to repeat her control cycle so that she exactly captures Demon's robot during the next control cycle, as Demon might save up his energy and speed up just when Angel predicted to catch him.
Having a better battery, Angel will still ultimately win even if Demon sped ahead, but not in the round she thought to be able to predict.

Unlike Hare's convergence axiom, Harel's convergence rule \cite{DBLP:conf/stoc/HarelMP77} is not a separating axiom, because it is sound for \dGL, just unnecessary.
In light of \rref{thm:dGL-closure-lower}, it is questionable whether the convergence rule would be relatively complete for hybrid games, because it is based on the existence of bounds on the repetition count.
The hybrid version of Harel's convergence rule \cite{DBLP:journals/jar/Platzer08} reads as follows (it assumes that $v$ does not occur in $\alpha$):
\[
      {\linferenceRule[formula]
        {\mapply{\var}{v+1}\land v+1>0\limply\ddiamond{\alpha}{\mapply{\var}{v}}}
        {\lexists{v}{\mapply{\var}{v}} \limply
            \ddiamond{\prepeat{\alpha}}{\lexists{v{\leq}0}{\mapply{\var}{v}}}}
      }{}%
\]
If the convergence rule could prove, e.g., \dGL formula \rref{eq:omega2-closure} from \rref{thm:dGL-closure-lower}, then $\mapply{\var}{\cdot}$ would yield a bound on the number of repetitions, which, by the proof of \rref{thm:dGL-closure-lower} does not exist below closure ordinal $\omega\cdot2$.
The premise of the convergence rule makes the bound induced by $\mapply{\var}{v}$ progress by 1 in each iteration.
The postcondition in the conclusion makes it terminate for $v\leq0$.
And the conclusion's antecedent requires a real number for the initial bound.
Thus, the convergence rule only permits bounds below $\omega$, not the required transfinite ordinal $\omega\cdot2$.

These thoughts further suggest a transfinite version of the convergence rule with an extra inductive premise for limit ordinals.
That would be interesting, but is technically more involved than the simple \dGL axiomatization, because it would require multi-sorted quantifiers and proof rules for ordinal arithmetic.

\section{Expressiveness} \label{sec:dGL-dL-expressiveness}

Differential game logic \dGL is a logic for hybrid games whose axiomatic separation to differential dynamic logic \dL for hybrid systems \cite{DBLP:journals/jar/Platzer08,DBLP:conf/lics/Platzer12b} has been characterized in \rref{sec:separating-axioms}.
How does \dGL compare in expressiveness to differential dynamic logic \dL, which is the corresponding logic for hybrid systems?
Hybrid systems are expected to be single-player hybrid games where one of the players never gets to decide.
And, \dL is expected to be a sublogic of \dGL. But what about the converse?
How the expressiveness of \dGL relates to that of \dL is related to classical long-standing open questions for the propositional case \cite{DBLP:journals/anndiscrmath/Parikh85,DBLP:journals/mst/BerwangerGL07}.
Note that even known classical results about expressiveness for the propositional case do not transfer to \dGL, because they hinge on finite state \cite{DBLP:journals/anndiscrmath/Parikh85}.

The notation \(L_1 \leq L_2\) signifies that logic $L_2$ is expressive for logic $L_1$ (\rref{def:expressive}).
Likewise, \(L_1 \mequiv L_2\) signifies equivalent expressiveness, i.e.\ \(L_1 \leq L_2\) and \(L_2 \leq L_1\).
Further, \(L_1 < L_2\) means that $L_1$ is strictly less expressive than $L_2$, i.e.\ \(L_1 \leq L_2\) but not \(L_2 \leq L_1\).

\begin{lemma}[Single-player hybrid games]%
  \label{lem:dL<=dGL}%
  \(\dL \leq \dGL\) by syntactic embedding.
\end{lemma}
\begin{proof}
Hybrid systems form single-player hybrid games, i.e.\ $\pdual{}$-free hybrid games.
The identity function is a syntactic embedding of \dL into \dGL, which preserves the semantics as follows.
With \rref{lem:HP-Scott-continuous}, Kleene's fixpoint theorem implies that $\omega$ is the closure ordinal for $\pdual{}$-free hybrid games $\alpha$.
Hence, for $\pdual{}$-free $\alpha$, a simple induction shows
\begin{equation}
  \strategyfor[\prepeat{\alpha}]{X} = \inflopstrat[\omega][\alpha]{X} = \cupfold_{n<\omega} \inflopstrat[n][\alpha]{X} = \cupfold_{n<\omega} \strategyfor[\alpha^n]{X}
  \label{eq:dual-free-repetition}
\end{equation}
where $\alpha^n$ is the $n$-fold sequential composition of $\alpha$ given by \(\alpha^0 \mequiv\, \ptest{\ltrue}\) and \(\alpha^{n+1} \mequiv \alpha;\alpha^n\).
The semantics of $\pdual{}$-free \dGL agrees with that defined for \dL originally \cite{DBLP:journals/jar/Platzer08,DBLP:conf/lics/Platzer12b} by a simple comparison using \rref{eq:dual-free-repetition} for the crucial case $\prepeat{\alpha}$.
\end{proof}

What about the converse?
Is the logic \dGL truly new or could it have been expressed in \dL?
Unlike \dL, \dGL is meant for hybrid games and makes it more convenient to refer directly to questions about hybrid games.\footnote{
Even if a logic is not strictly more expressive but ``only'' more convenient, it is still often strongly preferable. Program logics and their cousins, for example, are used widely, even though first-order integer arithmetic would theoretically suffice \cite{DBLP:conf/stoc/HarelMP77,DBLP:journals/iandc/HarelK84,Harel_et_al_2000}.
}
Does \dGL provide features strictly necessary for hybrid games that \dL is missing?
Finitely bounded hybrid games are expressible in \dL by \rref{cor:dGL-bounded-complete}.
What about other hybrid games?
Both possible outcomes are interesting.
If \(\dL \mequiv \dGL\), then \rref{thm:dGL-complete} implies that \dGL is complete relative to \dL and relative to the smaller logics that \dL is complete for \cite{DBLP:conf/lics/Platzer12b}.
If \(\dL < \dGL\), instead, then \dGL is a provably more expressive logic with features that are strictly necessary for hybrid games.
The answer takes some preparations but it also characterizes the general expressiveness of \dL and \dGL as a byproduct.

Let \(\pair{y_0,\dots}{y_n}\) denote a $\reals$-G\"odel encoding, i.e.\ a bijective function pairing \m{(n{+}1)}-tuples of real numbers $y_0,y_1,\dots,y_n$ into a single real, \(\pair{y_0,\dots}{y_n}\), that, along with its inverse, is definable in \FOD \cite[Lemma 4]{DBLP:journals/jar/Platzer08}.
\FOD is the \emph{first-order logic of differential equations} \cite{DBLP:journals/jar/Platzer08}, i.e.\ the first-order fragment of \dL and \dGL where all hybrid games $\alpha$ are of the form \m{\pevolve{\D{x}=\genDE{x}}}.
By \rref{lem:dL<=dGL}, \FOD is a sublogic of \dGL and, thus, $\reals$-G\"odel encodings are definable in \dGL.
\emph{(Rich-test) regular dynamic logic} (\DL) \cite{Harel_et_al_2000,DBLP:conf/lics/Platzer12b} over $\reals$ is the fragment of \dL (and by \rref{lem:dL<=dGL} of \dGL) without $\pdual{}$ and without differential equations.
Both \FOD \cite[Lemma 4]{DBLP:journals/jar/Platzer08} and \DL \cite[Theorem 9]{DBLP:conf/lics/Platzer12b} can define $\reals$-G\"odel encodings.
\emph{Acceptable structures} are structures in which elementary $\reals$-G\"odel encodings are definable \cite{Moschovakis74}.

The open \emph{recursive game quantifier} $\Game$ of length $\omega$ applied to formula \(\mapply{\varphi}{x,y}\) is
\begin{equation}
\Game y\, \mapply{\varphi}{x,y} \mdefequiv
\lforall{y_0}{\lexists{y_1}{\lforall{y_2}{\lexists{y_3}{\ldots \lorfold_{n<\omega} \mapply{\varphi}{x,\pair{y_0,\dots}{y_n}}}}}}
\label{eq:Game-quantifier}
\end{equation}
The semantics of this infinitary formula with its $\omega$ many quantifiers and its infinitary disjunction of length $\omega$ is defined by a Gale-Stewart \citeyear{GaleS53} game in which two players alternate in choosing values for the $\omega$ many variables $y_{2i}$ (for player $\forall$) and $y_{2i+1}$ (for player $\exists$). 
Player $\exists$ wins if \(\mapply{\varphi}{x,\pair{y_0,\dots}{y_n}}\) holds for some $n<\omega$, i.e.\ \(\mapply{\varphi}{x,y}\) holds when $y$ has the value \(\pair{y_0,\dots}{y_n}\) that is a $\reals$-G\"odel encoding of the tuple \((y_0,\dots,y_n)\).
If player $\exists$ has a winning strategy, the infinitary disjunction \(\lorfold_{n<\omega} \mapply{\varphi}{x,\pair{y_0,\dots}{y_n}}\) is satisfied; see \cite{Moschovakis74,Vaananen} for details.%
\begin{lemma}[Game quantifier]%
  \label{lem:dGL-Game-quantifier}%
  Recursive game quantifier $\Game$ is definable in \dGL.
\end{lemma}
\begin{proof}
Let \(\mapply{\varphi}{x,y}\) a \dGL formula, which, to simplify notation, is assumed to check the sequence that $y$ encodes only at odd indices $n$.
Then \(\Game y\, \mapply{\varphi}{x,y}\) is definable in \dGL:
\begin{equation}
  \big\langle{\pupdate{\pumod{y}{\tupleempty}};\prepeat{\big(
    \devolve{\D{z}=1};\devolve{\D{z}=-1}; \pupdate{\pumod{y}{\pair{y}{z}}};
    \pevolve{\D{z}=1};\pevolve{\D{z}=-1}; \pupdate{\pumod{y}{\pair{y}{z}}}
  \big)}\big\rangle}
  {\,\mapply{\varphi}{x,y}}
  \label{eq:dGL-Game-quantifier}
\end{equation}
This \dGL formula uses
\(\pair{y_0,y_1,y_2,\dots}{y_n}\) reordered as \(\pair{\dots\pair{\pair{\pair{\tupleempty}{y_0}}{y_1}}{y_2}}{\dots y_n}\) by a recursive permutation starting from the empty tuple encoding $\tupleempty$ for simplicity reasons.
Angel and Demon alternate differential equations for $z$ in \rref{eq:dGL-Game-quantifier} that get successively paired into $y$ by the pairing assignment \(\pupdate{\pumod{y}{\pair{y}{z}}}\), which is definable \cite[Lemma 4]{DBLP:journals/jar/Platzer08}.
This alternation of differential equations corresponds to the alternation of quantifiers in $\Game$.
The number of actual alternations played in \rref{eq:dGL-Game-quantifier} can be exactly any arbitrary \(n<\omega\), because the semantics of \(\ddiamond{\prepeat{\alpha}}{}\) is a least fixpoint, so well-founded.
In each round, Demon first changes $z$ to an arbitrary value by evolving along \(\devolve{\D{z}=1}\) and then \(\devolve{\D{z}=-1}\) for a suitable amount of time.
Subsequently, Angel changes $z$ to an arbitrary value, and both values of $z$ are paired into $y$ using the $\reals$-G\"odel encoding.
\end{proof}
Note that \rref{eq:dGL-Game-quantifier} equivalently defines \rref{eq:Game-quantifier}, even though the \rref{eq:dGL-Game-quantifier}  is a finite \dGL formula while \rref{eq:Game-quantifier} is an infinite formula in an infinitary logic augmented with the game quantifier \cite{Vaananen}, so \rref{eq:dGL-Game-quantifier} is infinitely more concise.
The closed recursive game quantifier \(\lnot\Game y\, \lnot \mapply{\phi}{x,y}\) is definable in \dGL by duality as well, noting that open as well as closed Gale-Stewart games are determined \cite{GaleS53}.
Finally observe that \rref{eq:dGL-Game-quantifier} would not define \rref{eq:Game-quantifier} in the (weaker) advance notice semantics (\rref{app:alternative-semantics}), which corresponds to swapping the quantifier alternation with \(\lorfold_{n<\omega}\) to the finitary:
\[
\lorfold_{n<\omega} \lforall{y_0}{\lexists{y_1}{\lforall{y_2}{\lexists{y_3}{\ldots \mapply{\varphi}{x,\pair{y_0,\dots}{y_n}}}}}}
\]

With this preparation, \dGL can be proved to be strictly more expressive than \dL, which means that hybrid games are fundamentally more expressive than hybrid systems.
In passing, the expressiveness of \dGL and \dL are characterized in terms of inductive and first-order definability, respectively, over acceptable reals.

\begin{theorem}[Expressive power]%
  \label{thm:dL<dGL}%
  \(\dL < \dGL\).
\end{theorem}
\begin{proof}
By \rref{lem:dL<=dGL}, it only remains to refute \(\dGL \leq \dL\).
$\reals$-G\"odel encodings etc.\ are elementarily definable in \FOD \cite[Lemma 4]{DBLP:journals/jar/Platzer08}, thus, also in \DL over $\reals$ \cite[Theorem 9]{DBLP:conf/lics/Platzer12b}.
This makes $\reals$ an acceptable structure \cite{Moschovakis74} when augmented with the corresponding definitions from \FOD or \DL over $\reals$.
Further, \(\dL \mequiv \FOD\) \cite{DBLP:journals/jar/Platzer08} and \(\dL \mequiv \DL\) over $\reals$ \cite[Theorem 9]{DBLP:conf/lics/Platzer12b}.
On acceptable structures, \DL defines exactly all first-order definable relations \cite[Theorems 3 and 4]{DBLP:journals/iandc/HarelK84}.
On acceptable structures, the open recursive game quantifier \(\Game y\, \mapply{\varphi}{x,y}\) for first-order formulas \(\mapply{\varphi}{x,y}\) exactly defines all (positive first-order) inductively definable relations \cite[Theorem 5C.2]{DBLP:journals/pams/Moschovakis72,Moschovakis74}.\footnote{
The game quantifier in \cite{DBLP:journals/pams/Moschovakis72}
starts with $\exists y_1$, which is a difference easily overcome.
}
Game quantifier $\Game$ is definable in \dGL by \rref{lem:dGL-Game-quantifier}, and so are all inductive relations.
In acceptable structures, not all inductively definable relations are first-order definable \cite[Theorem 5B.2]{Moschovakis74}. %
Thus, \dGL defines an inductive relation that \DL cannot define over $\reals$, so neither can \dL.
Hence, \(\dL \mequiv \DL < \dGL\) over $\reals$.
\end{proof}

Thus, hybrid games can characterize relations that hybrid systems cannot, because \dGL defines all inductive relations over (augmented) $\reals$, while \dL defines exactly all first-order definable relations.
The proof of \rref{thm:dGL-closure-lower} implies that $\omega_1^{\text{HG}}$ exceeds all order types of all inductive well-orders, because all inductive relations can be characterized in \dGL.
All closure ordinals of inductive relations occur as order types of some inductive well-order, because the staging order of inductive definitions is well-founded \cite[Theorems 3A.3, 3C.1]{Moschovakis74}.
So $\omega_1^{\text{HG}}$ equals the closure ordinal of the underlying structure.

The game quantifier and its characterization in the proof of \rref{lem:dGL-Game-quantifier} along with the differential equation characterization of G\"odel encodings \cite[Lemma 4]{DBLP:journals/jar/Platzer08} implies the existence of a smaller syntactic fragment of \dGL that is differentially expressive, so that \dGL is complete relative to this fragment of \dGL by \rref{thm:dGL-complete}.
By \rref{eq:Game-quantifier}, alternating differential equations in a single loop are the dominant feature of this fragment.
The only modification to the proof of \rref{lem:expressive/continuous} is the case of \(\reduct{(\ddiamond{\prepeat{\alpha}}{\phi})}\) which then uses \rref{eq:Game-quantifier} with a (definable) formula \(\mapply{\varphi}{x,\pair{y_0,\dots}{y_n}}\) that simply checks whether the decision sequence \(y_0,\dots,y_n\) gives a valid play of hybrid game $\prepeat{\alpha}$ in which Angel wins.
The fact that $\Game$ assumes strict alternation of the players is easily overcome by choosing $\varphi$ to be independent of $y_i$ when the player for its quantifier does not get to choose at step $i$ in $\prepeat{\alpha}$.
The actions can be chosen, e.g., as discussed in \rref{app:operational-HG-semantics}.

\section{Related Work} \label{sec:RelatedWork}

Games and logic have been shown to interact fruitfully in many ways \cite{GaleS53,Ehrenfeucht61,DBLP:conf/focs/Parikh83,DBLP:journals/anndiscrmath/Parikh85,Aumann95,HintikkaS97,Stirling01,DBLP:journals/jacm/AlurHK02,DBLP:journals/sLogica/PaulyP03a,AptGradel,Vaananen}.
The present article focuses on using logic to specify and verify properties of hybrid games, inspired by Parikh's propositional game logic for finite-state discrete games \cite{DBLP:conf/focs/Parikh83,DBLP:journals/anndiscrmath/Parikh85,DBLP:journals/sLogica/PaulyP03a}.

Parikh's game logic generalizes (propositional discrete) dynamic logic to discrete games played on a finite state space, and subsumes $\Delta$PDL and CTL$^*$ \cite{DBLP:journals/sLogica/PaulyP03a}.
After more than two decades, its expressiveness has only begun to be understood. It has been shown that the alternation hierarchy in propositional game logic is strict and encodes parity games that span the full alternation hierarchy of the (propositional) modal $\mu$-calculus \cite{DBLP:journals/sLogica/Berwanger03} and that, being in the two variable fragment, it is less expressive than the (propositional) modal $\mu$-calculus \cite{DBLP:journals/mst/BerwangerGL07}.
Another influential propositional modal logic, ATL$^*$ has been used for model checking finite-state systems \cite{DBLP:journals/jacm/AlurHK02} and is related to propositional game logic \cite{DBLP:conf/icla/BerwangerP09}.
Applications and relations of game logic, ATL$^*$ \cite{DBLP:journals/jacm/AlurHK02}, and  strategy logics with explicit strategies \cite{DBLP:journals/iandc/ChatterjeeHP10,DBLP:conf/concur/MogaveroMPV12} are discussed in the literature \cite{DBLP:journals/jacm/AlurHK02,DBLP:journals/sLogica/PaulyP03a,DBLP:conf/icla/BerwangerP09,DBLP:journals/iandc/ChatterjeeHP10,DBLP:conf/concur/MogaveroMPV12,DBLP:journals/aamas/BullingJ14}.
Completeness of the ATL fragment of ATL$^*$ has been considered \cite{DBLP:journals/tcs/GorankoD06} as well as its expressiveness and complexity \cite{DBLP:journals/lmcs/LaroussinieMO08}.
But logical investigations of ATL-type logics are scarce according to a recent survey \cite{DBLP:journals/aamas/BullingJ14} with more detailed comparisons and progress on the satisfiability problem for ATL with imperfect information.
These logics for the propositional case of finite-state discrete games are interesting, but it is not clear how their decision procedures should be generalized to the highly undecidable domain of hybrid games with differential equations, uncountable choices, and higher closure ordinals.
The logic \dGL shows how such hybrid games can be proved, enjoys compositionality, completeness, and comes with a rich theory.

Differential games have been studied with many different notions of solutions  \cite{Isaacs:DiffGames,Friedman,Petrosjan93,Bressan10}. They are of interest when actions are solely in continuous time.
The present article considers the complementary model of hybrid games where the underlying system is that of a hybrid system with interacting discrete and continuous dynamics, but the game actions are chosen at discrete instants of time, even if their outcomes take effect in continuous time.

Hybrid games provide a complementary perspective on differential games, just like hybrid systems provide a complementary perspective on continuous dynamical systems.
Differential games formalize various notions of adversarial control on variables for a single differential equation \cite{Isaacs:DiffGames,Friedman,Petrosjan93}, including solutions based on a non-anticipatory measurable input to an integral interpretation of the differential equations \cite{Friedman}, joint limits for $\delta\to0$ of lower and upper limits of $\delta$-anticipatory or $\delta$-delayed strategies \cite{Petrosjan93}, and Pareto-optimal, Nash, or Stackelberg equilibria, whose computation requires solving PDEs that quickly become ill-posed (already for feedback Nash equilibria except in very special cases); see Bressan \citeyear{Bressan10} for an overview.
Hybrid games, instead, distinguish discrete versus continuous parts of the dynamics, which simplifies the concepts, because easier pieces are involved, and, simultaneously, have been argued to make other aspects like delays in decisions and the integration of computer-decision into continuous physics more realistic \cite{DBLP:journals/tac/TomlinPS:98,DBLP:journals/IEEE/TomlinLS00,DBLP:journals/corr/abs-0911-4833,DBLP:journals/tcs/VladimerouPVD11,DBLP:conf/ecc/PrandiniHP01,DBLP:conf/cade/QueselP12}.
The situation is similar to hybrid systems, which provide a complementary perspective on continuous dynamical systems \cite{DBLP:conf/hybrid/NerodeK92a,DBLP:journals/tcs/AlurCHHHNOSY95,DBLP:journals/tac/BranickyBM98} that can model more complicated systems as a combination of simpler concepts \cite{DBLP:conf/lics/Platzer12b} and can model computational effects more realistically.

Some reachability aspects of games for hybrid systems have been studied before.
A game view on hybrid systems verification has been proposed following a Hamilton-Jacobi-Bellman PDE formulation \cite{DBLP:journals/pieee/TomlinMBO03,DBLP:journals/tac/MitchellBT05}, with subsequent extensions by Gao et al.\ \citeyear{DBLP:journals/tac/GaoLQ07}.
Their primary focus is on adversarial choices in the continuous dynamics not on interactions with the discrete dynamics or on interactive game play.
Similar observations apply to the viability theory approach to differential games, which gives powerful answers when the differential game and its winning conditions satisfy a number of conditions \cite{Cardaliaguet2007}.

WCTL properties of STORMED hybrid games, which are restricted to evolve linearly in one ``direction'' all the time, have been shown to be decidable using bisimulation quotients \cite{DBLP:journals/tcs/VladimerouPVD11}.
STORMED hybrid games generalize o-minimal hybrid games which have been shown to be decidable before \cite{DBLP:journals/corr/abs-0911-4833}.
Timed games \cite{DBLP:journals/corr/abs-1011-0688} as well as initialized rectangular hybrid games are known to be decidable \cite{DBLP:conf/concur/HenzingerHM99}, which is limited to the case where all evolution domains and jump constraints are bounded rectangles independent of the previous state of the system and when the controller can only either disable transitions or decide when to take transitions, not both \cite{DBLP:conf/concur/HenzingerHM99}.
Many applications do not fall into these decidable classes \cite{DBLP:conf/cade/QueselP12}, so that a study of more general hybrid games is called for.
The results in this article have implications for such reachability analyses. They show, for example, that reachability computations and backwards induction for hybrid games require highly transfinite closure ordinals $\geq\dGLordinal$.
The completeness proof further exactly characterizes the challenging cases in hybrid games verification.

This article takes a complementary view and studies logics and proofs for hybrid games instead of searching for decidable fragments using bisimulation quotients \cite{DBLP:conf/concur/HenzingerHM99,DBLP:journals/corr/abs-0911-4833,DBLP:journals/tcs/VladimerouPVD11}, which cannot generally exist.
It provides a proof-based and compositional verification technique for more general hybrid games with nonlinear dynamics.
This article's notion of hybrid games is more flexible, because it allows arbitrary nested hybrid game choices rather than one fixed pattern of interaction such as the game of a discrete controller against a continuous plant considered in related work.
This results in \dGL's ability to express more general logical formulas with the flexibility expected from a logic and programming language, including arbitrarily nested game operators and nested modalities, which leads to a rich logical theory.

There is more than one way how logic can help to understand games of hybrid systems.
Concurrent work has shown that games can also be added as separate constructs on top of unmodified differential dynamic logic \cite{DBLP:conf/cade/QueselP12}, focusing on the special case of advance notice semantics (\rref{app:alternative-semantics}).
The present article follows an entirely different principle.
Instead of leaving differential dynamic logic untouched and adding several separate game constructs on top of full hybrid systems reachability modalities \cite{DBLP:conf/cade/QueselP12}, the logic \dGL becomes a proper game logic by adding a single operator $\pdual{}$ for adversariality into the system dynamics.
The logic \dGL results in a much simplified but nevertheless more general logic with a simpler and more general semantics and simpler and more general proof calculus.
The present article studies a Hilbert calculus and focuses on fundamental logical properties and theory.
See \cite{DBLP:conf/cade/QueselP12} for practical aspects like a very challenging robotic factory automation case study that translates to \dGL.
Since \dGL is a gentle extension with the single operator $\pdual{}$, it is more elegant and significantly easier to implement.
What is more difficult in \dGL in comparison to that fragment \cite{DBLP:conf/cade/QueselP12}, however, is the need to carefully identify which axioms are no longer sound for games, which has been pursued in \rref{sec:separating-axioms}.

The logic \dGL presented here has similarities with stochastic differential dynamic logic (\SdL) \cite{DBLP:conf/cade/Platzer11}, because both may be used to verify properties of the hybrid system dynamics with partially uncertain behavior.
Both approaches do, however, address uncertainty in fundamentally different ways.
\SdL takes a probabilistic perspective on uncertainty in the system dynamics.
The \dGL approach put forth in this paper, instead, takes an adversarial perspective on uncertainty.
Both views on how to handle uncertain behavior are useful but serve quite different purposes, depending on the nature of the system analysis question at hand.
A probabilistic understanding of uncertainty can be superior whenever good information is available about the distribution of choices made by the environment and other agents.
Whenever that is not possible, adversarial views may be more appropriate, since they do not lead to the inadequate biases that arbitrary probabilistic assumptions would impose.
Adversarial dynamics is also called for in cases of true competition, like in RoboCup.

\section{Conclusions and Future Work} \label{sec:Conclusion}

This article introduced differential game logic (\dGL) for hybrid games that combine discrete, continuous, and adversarial dynamics.
Just like hybrid games unify hybrid systems with discrete games, \dGL unifies logic of hybrid systems with Parikh's propositional game logic of finite-state discrete games.
Hybrid games are challenging, since computing their winning regions may require closure ordinals ${\geq}\omega_1^{\text{CK}}$.
The logic \dGL for hybrid games is fundamentally more expressive than the corresponding logic \dL for hybrid systems, because it defines all inductive relations over the augmented structure of $\reals$ rather than exactly the first-order definables.
Nevertheless, \dGL has a simple modal semantics and a simple proof calculus, which is proved to be a sound and complete axiomatization of hybrid games relative to any (differentially) expressive logic.

The completeness proof is constructive with minimal coding, thereby exactly characterizing all difficult parts of hybrid games proving.
The proof identifies an efficient fixpoint-style proof technique, which can be considered a modal analogue of characterizations in the Calculus of Constructions \cite{DBLP:journals/iandc/CoquandH88}, and relates to hybrid game versions of influential views of understanding program invariants as fixpoints \cite{DBLP:conf/popl/CousotC77,Clarke79}.
Relative completeness shows that \dGL has all axioms and proof rules for dealing with hybrid games and only the base games of differential equations themselves are difficult.
The study of (fragments of) \dGL which are complete for smaller logics is interesting future work.
By the schematic completeness result, this reduces solely to questions of expressiveness, which give rise to interesting questions in descriptive set theory.

It is intriguing to observe the overwhelming impact of the innocent addition of a duality operator. Yet, it is also reassuring to find that logical robustness makes logical foundations continue to work despite the formidable extra challenges of hybrid games.
To wit, this article contrasted hybrid games with hybrid systems in terms of their analytic complexity, axiomatizations, and expressiveness.

The \dGL axiomatization is strikingly similar to the calculus for stochastic differential dynamic logic \SdL \cite{DBLP:conf/cade/Platzer11}, despite their fundamentally different semantical presuppositions (adversarial nondeterminism versus stochasticity), which indicates the existence of a deeper logical connection relating stochastic and adversarial uncertainty despite their different mathematical basis (fixpoints, closure ordinals, acceptable structures, inductive definability, and game theory versus stochastic processes, martingales, Markov times, and infinitesimal generators).
Because of the axiomatic similarity, the rich theory of \dGL may shed light on the logical theory of stochastic hybrid systems, which so far remained elusive.

The logic of hybrid games opens up many directions for future work, including the study of computationally bounded winning strategies, e.g., only strategies that are constructible with small closure ordinals, or with finite rank Borel winning regions, as well as a study of constructive \dGL to retain the winning strategies as explicit proof terms.
Yet, challenges abound, given the ability of \dGL to define closed elementary games won by a player for whom no hyperelementary quasiwinning strategies exist, which follows from \rref{thm:dL<dGL} by \cite[Chapter 7]{Moschovakis74}. %

Draws, coalitions, rewards, and payoffs different from $\pm1$ can be expressed easily in \dGL using extra variables, but it may be useful to include direct syntactical support.
\rref{thm:dL<dGL} shows that all inductively definable game concepts are expressible in \dGL. Rather than including direct support for each, \dGL focuses on the most fundamental aspects of hybrid games for reasons of simplicity and elegance.
Concurrent games and their equivalent sequential imperfect information games are interesting but are challenging even in the discrete case, because imperfect information leads to Henkin quantifiers.
By \rref{thm:dL<dGL}, the challenge is not to add concurrent games but rather to sustain \dGL's compositional verification principles.
The logic \dGL presented here can be augmented with differential games as a new kind of atomic games \cite{DBLP:journals/corr/Platzer15:dGI}.
Thanks to its compositional semantics, this results in a modular construction, but is not pursued in this article, because it requires a separate body of mathematics.
Combining \dGL with axioms for differential equations \cite{DBLP:journals/logcom/Platzer10,DBLP:conf/lics/Platzer12b} already provides a way of handling hybrid games with nonlinear differential equations, differential-algebraic inequalities and differential equations with input.

\appendix

\section{Example \dGL Proofs} \label{app:ExampleProofs}

The completeness proof suggests the use of iteration axiom \irref{iterated} and \irref{US} to prove $\ddiamond{\prepeat{\alpha}}{}$ properties.
The following examples illustrate how this works in practice.
Observe how logic programming saturation with widening quickly proves the resulting arithmetic.
\begin{example}[Non-game system] \label{ex:dGL-proof-example1}
The simple non-game \dGL formula
\[
x\geq0 \limply \ddiamond{\prepeat{(\pupdate{\pumod{x}{x-1}})}}{0\leq x<1}
\]
is provable, shown in \rref{fig:dGL-proof-example1},
\newcommand{\orgfo}{\ddiamond{\prepeat{\alpha}}{0{\leq}x{<}1}}%
where $\orgfo$ is short for
\(\ddiamond{\prepeat{(\pupdate{\pumod{x}{x-1}})}}{(0\leq x<1)}\).
The \irref{MP} use in \rref{fig:dGL-proof-example1} is Hilbert-style, i.e.\ combines the two lines above by modus ponus.%
\begin{figure*}[tbhp]
\begin{minipage}{\textwidth}
\begin{sequentdeduction}[array]
\linfer[MP]
{\linfer[iterated+gena]
{
\linfer[US]
{\linfer[assignd]
  {\linfer[RCFp]
    {\lclosea}
    {\lforall{x}{(0\leq x<1 \lor p(x-1) \limply p(x))} \limply (x\geq0 \limply p(x))}
  }
  {\lforall{x}{(0\leq x<1 \lor \ddiamond{\pupdate{\pumod{x}{x-1}}}{p(x)} \limply p(x))} \limply (x\geq0 \limply p(x))}
}
{\lforall{x}{(0\leq x<1 \lor \ddiamond{\pupdate{\pumod{x}{x-1}}}{\orgfo} \limply \orgfo)} \limply (x\geq0 \limply \orgfo)}
}
{\lforall{x}{(0\leq x<1 \lor \ddiamond{\pupdate{\pumod{x}{x-1}}}{\orgfo} \limply \orgfo)}}
}
{x\geq0 \limply \orgfo}
\end{sequentdeduction}
\end{minipage}
\caption{\dGL Angel proof for \rref{ex:dGL-proof-example1} using technique from completeness proof}
\label{fig:dGL-proof-example1}
\end{figure*}
\end{example}

\begin{example}[Choice game] \label{ex:dGL-proof-example2}
The \dGL formula
\[
x=1\land a=1 \limply \ddiamond{\prepeat{(\dchoice{\pupdate{\pumod{x}{a}};\pupdate{\pumod{a}{0}}}{\pupdate{\pumod{x}{0}}})}}{x\neq1}
\]
which comes from \rref{eq:advance-notice-ex} on p.\,\pageref{eq:advance-notice-ex} is provable as shown in \rref{fig:dGL-proof-example2},
\newcommand{\orgfo}{\ddiamond{\prepeat{(\dchoice{\beta}{\gamma})}}{x\neq1}}%
where $\dchoice{\beta}{\gamma}$ is short for
\(\dchoice{\pupdate{\pumod{x}{a}};\pupdate{\pumod{a}{0}}}{\pupdate{\pumod{x}{0}}}\) and $\orgfo$ short for \(\ddiamond{\prepeat{(\dchoice{\pupdate{\pumod{x}{a}};\pupdate{\pumod{a}{0}}}{\pupdate{\pumod{x}{0}}})}}{x\neq1}\):
\begin{figure*}[tbhp]
\begin{sequentdeduction}[array]
\linfer[RCFp]
{\linfer[iterated+gena+MP]
{\linfer[US]
{\linfer[duald+choiced]
  {\linfer[composed+assignd]
    {\linfer[RCFp]
      {\lclosea}
      {\lforall{x}{(x\neq1 \lor p(a,0)\land p(0,a) \limply p(x,a))} \limply (\ltrue \limply p(x,a))}
    }
    {\lforall{x}{(x\neq1 \lor \ddiamond{\beta}{p(x,a)}\land\ddiamond{\gamma}{p(x,a)} \limply p(x,a))} \limply (\ltrue \limply p(x,a))}
  }
  {\lforall{x}{(x\neq1 \lor \ddiamond{\dchoice{\beta}{\gamma}}{p(x,a)} \limply p(x,a))} \limply (\ltrue \limply p(x,a))}
}
{\lforall{x}{(x\neq1 \lor \ddiamond{\dchoice{\beta}{\gamma}}{\orgfo} \limply \orgfo)} \limply (\ltrue \limply \orgfo)}
}
{\ltrue \limply \orgfo}
}
{x=1\land a=1 \limply \orgfo}
\end{sequentdeduction}
\caption{\dGL Angel proof for \rref{ex:dGL-proof-example2} using technique from completeness proof}
\label{fig:dGL-proof-example2}
\end{figure*}
\end{example}

\begin{example}[Hybrid game] \label{ex:dGL-proof-example3}
The \dGL formula
\[
\ddiamond{\prepeat{(\pchoice{\pupdate{\pumod{x}{1}};\devolve{\D{x}=1}}{\pupdate{\pumod{x}{x-1}}})}}{0\leq x<1}
\]
from \rref{eq:omega-strategic-ex} on p.\,\pageref{eq:omega-strategic-ex} is provable as shown in \rref{fig:dGL-proof-example3},
\newcommand{\orgfo}{\ddiamond{\prepeat{(\pchoice{\beta}{\gamma})}}{0{\leq}x{<}1}}%
where the notation $\orgfo$ is short for
\(\ddiamond{\prepeat{(\pchoice{\pupdate{\pumod{x}{1}};\devolve{\D{x}=1}}{\pupdate{\pumod{x}{x-1}}})}}{(0\leq x<1)}\):
\begin{figure*}[tbhp]
\begin{minipage}{\textwidth}
\begin{sequentdeduction}[array]
\linfer[iterated+gena+MP]
{
\linfer[US]
{\linfer[choiced]
  {\linfer[composed+duald]
  {\linfer[evolved]
    {\linfer[assignd]
      {\linfer[RCFp]
        {\lclosea}
        {\lforall{x}{(0\leq x<1 \lor \lforall{t{\geq}0}{p(1+t)} \lor p(x-1) \limply p(x))} \limply (\ltrue \limply p(x))}
        }
        {\lforall{x}{(0\leq x<1 \lor \ddiamond{\pupdate{\pumod{x}{1}}}{\lnot\lexists{t{\geq}0}{\ddiamond{\pupdate{\pumod{x}{x+t}}}{\lnot p(x)}}} \lor p(x-1) \limply p(x))} \limply (\ltrue \limply p(x))}
      }
      {\lforall{x}{(0\leq x<1 \lor \ddiamond{\pupdate{\pumod{x}{1}}}{\lnot\ddiamond{\pevolve{\D{x}=1}}{\lnot p(x)}} \lor p(x-1) \limply p(x))} \limply (\ltrue \limply p(x))}
    }
    {\lforall{x}{(0\leq x<1 \lor \ddiamond{\beta}{p(x)} \lor \ddiamond{\gamma}{p(x)} \limply p(x))} \limply (\ltrue \limply p(x))}
  }
  {\lforall{x}{(0\leq x<1 \lor \ddiamond{\pchoice{\beta}{\gamma}}{p(x)} \limply p(x))} \limply (\ltrue \limply p(x))}
}
{\lforall{x}{(0{\leq} x{<}1 {\lor}\ddiamond{\pchoice{\beta}{\gamma}}{\orgfo} {\limply} \orgfo)} {\limply} (\ltrue {\limply}\orgfo)}
}
{\ltrue \limply \orgfo}
\end{sequentdeduction}
\end{minipage}
\caption{\dGL Angel proof for \rref{ex:dGL-proof-example3} using technique from completeness proof}
\label{fig:dGL-proof-example3}
\end{figure*}
Here and in \rref{fig:dGL-proof-example2}, the \irref{gena+iterated+MP} steps conclude as in \rref{fig:dGL-proof-example1}.
The proof step \irref{evolved} uses that \(t\mapsto x+t\) is the solution of the differential equation, so the subsequent use of \irref{assignd} substitutes 1 in for $x$ to obtain \(t\mapsto 1+t\).
Recall that the winning regions for formula \rref{eq:omega-strategic-ex} need ${>}\omega$ iterations to converge.
It is still provable easily.
A variation of this proof shows \dGL formula \rref{eq:ex-explosive-clock} from p.\,\pageref{eq:ex-explosive-clock}, where the handling of the nonlinear differential equation is a bit more complicated.
\end{example}
A variation of \rref{ex:dGL-proof-example3} also proves \dGL formula \rref{eq:omega2-closure} from the proof of \rref{thm:dGL-closure-lower}, whose closure ordinal is $\omega\cdot2$.

\section{Proof of Separating Axioms} \label{app:separating-axioms}

This section proves \rref{thm:separating-axioms} with an emphasis on simple counterexamples for each separating axiom to identify the logical essence of the informal explanations shown in \rref{sec:separating-axioms}.

\paragraph{Subnormal Modal Logic}
Unlike \dL, \dGL is not a normal modal logic  \cite{HughesCresswell96}.
Axiom K, the modal modus ponens from normal modal logic \cite{HughesCresswell96}, dynamic logic \cite{DBLP:conf/focs/Pratt76}, and differential dynamic logic \cite{DBLP:conf/lics/Platzer12b}, i.e.\
\[
\cancel{\text{K}}~~
\linferenceRule[impl]
        {\dbox{\alpha}{(\phi\limply\psi)}}
        {(\dbox{\alpha}{\phi}\limply\dbox{\alpha}{\psi})}
\displaywidowpenalty=10000
\]
\begin{figure}[htbp]
  \centering
    \tikzstyle{box}+=[minimum width=0.5cm,minimum height=0.5cm]
  \begin{tikzpicture}[grow'=down,every node/.style={box},every edge/.style={boxa},level distance=1.2cm]
    \node {xy}
      child {node[boxwon] {00}}
      child {node[diawon] {10}}
      ;
     \node[action,vblue] at (+0,-2.4) {$\dbox{\alpha}{(x=1\limply y=1)}$};
    \node at (3,0) {xy}
      child {node[diawon] {00}}
      child {node[boxwon] {10}}
      ;
     \node[action,vblue] at (+3,-2.4) {$\dbox{\alpha}{x=1}$};
    \node at (6,0) {xy}
      child {node[diawon] {00}}
      child {node[diawon] {10}}
      ;
     \node[action,vred] at (+6,-2.4) {$\dbox{\alpha}{y=1}$};
  \end{tikzpicture}
  \caption{Game trees for counterexample to axiom K using \(\alpha\mequiv(\dchoice{\pupdate{\pumod{x}{1}}}{\pupdate{\pumod{x}{0}}});\pupdate{\pumod{y}{0}}\).}
  \label{fig:K-cex}
\end{figure}%

\noindent
is not sound for \dGL as witnessed using the choice 
\(\alpha\mequiv(\dchoice{\pupdate{\pumod{x}{1}}}{\pupdate{\pumod{x}{0}}});\pupdate{\pumod{y}{0}}\)
and
\(\phi\mequiv x=1\), \(\psi\mequiv y=1\); see \rref{fig:K-cex}.
The global version of K, i.e.\ the implicative version of G\"odel's generalization rule is still sound and derives with \irref{duald} and \irref{box} from \irref{M} using $\alpha\equiv\pdual{\beta}$
\[
\text{M}_{\dbox{\cdot}{}}~~
\linferenceRule[formula]
        {\phi\limply\psi}
        {\dbox{\beta}{\phi}\limply\dbox{\beta}{\psi}}
\]
The normal G\"odel generalization rule G, i.e.\
\[
\cancel{\text{G}}~~
      \linferenceRule[formula]
        {\phi}
        {\dbox{\alpha}{\phi}}
\]
however, is not sound for \dGL as witnessed by the choice \(\alpha\mequiv\pdual{(\ptest{\lfalse})}, \phi\mequiv\ltrue\).

The vacuous axiom V, which expresses that formulas do not change their truth-value along a program if their free variables are not bound, is sound for dynamic logics and differential dynamic logic \dL when no free variable of $\phi$ is bound by $\alpha$ \cite{DBLP:conf/lics/Platzer12b}:
\[
\cancel{\text{V}}~~ %
      {\linferenceRule[impl]
        {\phi}
        {\dbox{\alpha}{\phi}}
      }{\qquad(\freevars{\phi}\cap \boundvars{\alpha}=\emptyset)}%
\]
but the vacuous axiom is not sound for \dGL as witnessed by the choice \(\phi \mequiv x=0\) and \(\alpha \mequiv \pupdate{\pumod{y}{0}}; \pdual{(\ptest{y=1})}\).
With an additional assumption \m{\dbox{\alpha}{\ltrue}} expressing that the game can be played to a final state at all, the possible vacuous axiom \irref{VK} is sound for \dGL:
\[
\cinferenceRule[VK|VK]{vacuous possible $\dbox{}{}$}
      {\linferenceRule[impl]
        {\phi}
        {(\dbox{\alpha}{\ltrue} \limply \dbox{\alpha}{\phi})}
      }{\freevars{\phi}\cap \boundvars{\alpha}=\emptyset}%
\]
If Demon can always finish the game (\(\dbox{\alpha}{\ltrue}\)) then $\phi$ will continue to hold if it was true initially since $\alpha$ only changes bound variables and $\phi$ only depends on its free variables.

The closest counterpart to G that is sound for hybrid games is $\text{M}_{\dbox{\cdot}{}}$ and the closest counterpart to V that is sound is \irref{VK}.
Both require the extra assumption \(\dbox{\alpha}{\ltrue}\), which is only trivial for hybrid systems not for hybrid games.

\paragraph{Subregular Modal Logic}
Regular modal logics are monotone modal logics \cite{Chellas} that are weaker than normal modal logics.
But the regular modal generalization rule \cite{Chellas}, i.e.
\[
\cancel{\text{R}}~~
      \linferenceRule[formula]
        {\phi_1\land\phi_2\limply\psi}
        {\dbox{\alpha}{\phi_1} \land \dbox{\alpha}{\phi_2} \limply \dbox{\alpha}{\psi}}
\]
is not sound for \dGL either as witnessed by the choice 
\(\alpha\mequiv(\dchoice{\pupdate{\pumod{x}{1}}}{\pupdate{\pumod{x}{0}}});\pupdate{\pumod{y}{0}}\), \(\phi_1\mequiv x=1, \phi_2\mequiv x=y, \psi\mequiv x=1\land x=y\); see \rref{fig:Reg-cex}.
\begin{figure}[htbp]
  \centering
    \tikzstyle{box}+=[minimum width=0.5cm,minimum height=0.5cm]
  \begin{tikzpicture}[grow'=down,every node/.style={box},every edge/.style={boxa},level distance=1.2cm]
    \node {xy}
      child {node[diawon] {00}}
      child {node[boxwon] {10}}
      ;
     \node[action,vblue] at (+0,-2.4) {$\dbox{\alpha}{x=1}$};
    \node at (3,0) {xy}
      child {node[boxwon] {00}}
      child {node[diawon] {10}}
      ;
     \node[action,vblue] at (+3,-2.4) {$\dbox{\alpha}{x=y}$};
    \node at (6,0) {xy}
      child {node[diawon] {00}}
      child {node[diawon] {10}}
      ;
     \node[action,vred] at (+6,-2.4) {$\dbox{\alpha}{x=1\land x=y}$};
  \end{tikzpicture}
  \caption{Game trees for counterexample to regular modal rule using \(\alpha\mequiv(\dchoice{\pupdate{\pumod{x}{1}}}{\pupdate{\pumod{x}{0}}});\pupdate{\pumod{y}{0}}\).}
  \label{fig:Reg-cex}
\end{figure}%

\paragraph{Monotone Modal Logic}
The axiom that is closest to K but still sound for \dGL is a monotonicity axiom.
This axiom is sound for \dGL, yet already included in the monotonicity rule \irref{M}:
\begin{lemma}[{\cite[Theorem 8.13]{Chellas}}] 
  In the presence of rule \irref{RE} from p.\,\pageref{ir:RE}, rule \irref{M} is interderivable with axiom M:
  \upshape%
  \[
  \dinferenceRule[Maxiom|M]{monotonicity axiom}
  {
  \ddiamond{\alpha}{\phi} \lor \ddiamond{\alpha}{\psi} \limply \ddiamond{\alpha}{(\phi\lor\psi)}
  }{}
  \]
\end{lemma}
\begin{proof}
Axiom M derives from rule \irref{M}:
From \(\phi\limply\phi\lor\psi\), \irref{M} derives \(\ddiamond{\alpha}{\phi}\limply\ddiamond{\alpha}{(\phi\lor\psi)}\).
From \(\psi\limply\phi\lor\psi\), \irref{M} derives \(\ddiamond{\alpha}{\psi}\limply\ddiamond{\alpha}{(\phi\lor\psi)}\), from which propositional logic yields \(\ddiamond{\alpha}{\phi}\lor\ddiamond{\alpha}{\psi}\limply\ddiamond{\alpha}{(\phi\lor\psi)}\).

Conversely, rule \irref{M} derives from axiom M and rule \irref{RE}:
From \(\phi\limply\psi\) propositional logic derives \(\phi\lor\psi\lbisubjunct\psi\), from which \irref{RE} derives \(\ddiamond{\alpha}{(\phi\lor\psi)}\lbisubjunct\ddiamond{\alpha}{\psi}\).
From axiom M, propositional logic, thus, derives \(\ddiamond{\alpha}{\phi}\limply\ddiamond{\alpha}{\psi}\).
\end{proof}
The converse of axiom M is sound for \dL but not for \dGL, however, as witnessed by \(\alpha\mequiv\dchoice{\pupdate{\pumod{x}{1}}}{\pupdate{\pumod{x}{0}}}\), \(\phi\mequiv x=1, \psi\mequiv x=0\); see \rref{fig:Minverse-cex}:
  \[
  \cancel{$\overleftarrow{\text{M}}$}~~
  \ddiamond{\alpha}{(\phi\lor\psi)}
  \limply
  \ddiamond{\alpha}{\phi} \lor \ddiamond{\alpha}{\psi}
  \]
\begin{figure}[htbp]
  \centering
    \tikzstyle{box}+=[minimum width=0.5cm,minimum height=0.5cm]
  \begin{tikzpicture}[grow'=down,every node/.style={box},every edge/.style={boxa},level distance=1.2cm]
    \node {x}
      child {node[diawon] {0}}
      child {node[diawon] {1}}
      ;
     \node[action,vred] at (+0,-2.4) {$\ddiamond{\alpha}{(x=1\lor x=0)}$};
    \node at (3,0) {x}
      child {node[boxwon] {0}}
      child {node[diawon] {1}}
      ;
     \node[action,vblue] at (+3,-2.4) {$\ddiamond{\alpha}{x=1}$};
    \node at (6,0) {x}
      child {node[diawon] {0}}
      child {node[boxwon] {1}}
      ;
     \node[action,vblue] at (+6,-2.4) {$\ddiamond{\alpha}{x=0}$};
  \end{tikzpicture}
  \caption{Game trees for counterexample to converse monotone axiom using \(\alpha\mequiv\dchoice{\pupdate{\pumod{x}{1}}}{\pupdate{\pumod{x}{0}}}\).}
  \label{fig:Minverse-cex}
\end{figure}%

\noindent
The presence of the regular congruence rule \irref{RE} and the fact that \m{\dbox{\alpha}{\phi} \lbisubjunct \lnot\ddiamond{\alpha}{\lnot\phi}} by determinacy (\rref{thm:dGL-determined}) still make \dGL a classical modal logic \cite{Chellas}.
Rule \irref{M} even makes \dGL a monotone modal logic  \cite{Chellas}.

\paragraph{Sub-Barcan}
The most important axioms about the interaction of quantifiers and modalities in first-order modal logic are the Barcan and converse Barcan axioms \cite{DBLP:journals/jsyml/Barcan46}, which, together, characterize constant domain in normal first-order modal logics \cite{HughesCresswell96}.
The Barcan axiom B, which characterizes anti-monotone domains in first-order modal logic \cite{HughesCresswell96}, is sound for constant-domain first-order dynamic logic and for differential dynamic logic \dL when $x$ does not occur in $\alpha$ \cite{DBLP:conf/lics/Platzer12b}:
\[
\cancel{\text{B}}~~
\linferenceRule[impl]
        {\ddiamond{\alpha}{\lexists{x}{\phi}}}
        {\lexists{x}{\ddiamond{\alpha}{\phi}}}
      \qquad({\m{x\not\in\alpha}})
\]
but the Barcan axiom is not sound for \dGL as witnessed by the choice \(\alpha\mequiv\drepeat{\pupdate{\pumod{y}{y+1}}}\) or \(\alpha\mequiv\devolve{\D{y}=1}\) and \(\phi\mequiv (x\geq y)\).
The equivalent Barcan formula
\[
\cancel{\text{B}}~~
\linferenceRule[impl]
        {\lforall{x}{\dbox{\alpha}{\phi}}}
        {\dbox{\alpha}{\lforall{x}{\phi}}}
      \qquad({\m{x\not\in\alpha}})
\]
is not sound for \dGL as witnessed by the choice \(\alpha\mequiv\drepeat{\pupdate{\pumod{y}{y+1}}}\) or \(\alpha\mequiv\devolve{\D{y}=1}\) and \(\phi\mequiv y\geq x\).
The converse Barcan formula of first-order modal logic, which characterizes monotone domains \cite{HughesCresswell96}, is sound for \dGL and can be derived when $x$ does not occur in $\alpha$ (see \rref{foot:converseBarcan} on p.\,\pageref{foot:converseBarcan}):
\[
      \dinferenceRule[cB|$\overleftarrow{\text{B}}$]{converse Barcan$\ddiamond{}{}\lexists{}$}
      {\linferenceRule[impl]
        {\lexists{x}{\ddiamond{\alpha}{\phi}}}
        {\ddiamond{\alpha}{\lexists{x}{\phi}}}
      }{\m{x\not\in\alpha}}
\]

\paragraph{No Induction Axiom}

The induction axiom
\begin{equation}
\cancel{\,\text{I}\,}~~
  \dbox{\prepeat{\alpha}}{(\phi\limply\dbox{\alpha}{\phi})} \limply (\phi\limply\dbox{\prepeat{\alpha}}{\phi})
  \label{eq:induction-axiom}
\end{equation}
holds for \dL, but, unlike the induction rule \irref{invind}, does not hold for \dGL as witnessed by
\(\alpha\mequiv(\dchoice{(\pupdate{\pumod{x}{a}};\pupdate{\pumod{a}{0}})}{\pupdate{\pumod{x}{0}}})\) and \(\phi\mequiv (x=1)\); see \rref{fig:cex-I}.
\begin{figure}[ptbh]
    \tikzstyle{box}+=[minimum width=0.5cm,minimum height=0.5cm]
  \begin{tikzpicture}[grow'=down]
    \tikzstyle{level 1}=[sibling distance=+36mm]
    \tikzstyle{level 2}=[sibling distance=+22mm]
    \tikzstyle{level 3}=[sibling distance=+12mm]
    \tikzstyle{level 4}=[sibling distance=+7mm]
    \tikzstyle{level 5}=[sibling distance=+6mm]
    \tikzstyle{level 6}=[level distance=1.2cm]
    \node[diamond] (PO) {51}
      child[pdia] {node[box] {51}
        child[pbox] {node[diamond] (01d) {01}
          child[pdia] {node[box] {01}
            child[pbox] {node[diamond] (back01d) {01}}
            child[pbox] {node[diamond] {10}
              child[pdia] {node[box] {10} 
                child[pbox] {node[diamond,diawon] {00}}
                child[pbox] {node[diamond,diawon] {00}}
                edge from parent node[action] {repeat}
              }
              child[pdia] {node[box,diawon] (10bagain) {10} 
              edge from parent node[action] {stop}
              }
              edge from parent node[action] {stop}
            }
            edge from parent node[action] {repeat}
          }
          child[pdia] {node[box,boxwonstrategy] {01} edge from parent node[action] {stop}}
        }
        child[pbox] {node[diamond] {10}
          child[pdia] {node[box] {01} edge from parent node[action] {repeat}}
          child[pdia] {node[box] (10b) {10} 
            child[pbox] {node[box,diawon] {00}}
            child[pbox] {node[box,diawon] {00}}
            edge from parent node[action] {stop}
          }
        }
        edge from parent node[action] {repeat}
      }
      child[pdia] {node[box] {51}
        child[pbox] {node[box,diawon] {01}}
        child[pbox] {node[box,boxwonstrategy] {10}}
        edge from parent node[action] {stop}
      }
      ;
      \draw[backedge] (01d) to[bend left=60] (back01d);
      \draw[backedge] (10bagain) to (10b);
      \begin{pgfonlayer}{background}
        \draw[endgame] (PO-2.north west) -- (PO-2-2.north west) -- (PO-2-2.south west) -- (PO-2-1.south east) -- (PO-2-1.north east) -- (PO-2.north east) -- cycle;
        \draw[endgame] (10b.north west) -- (10b-2.north west) -- (10b-2.south west) -- (10b-1.south east) -- (10b-1.north east) -- (10b.north east) -- cycle;
      \end{pgfonlayer}
  \end{tikzpicture}%
  \hfill
  \begin{tikzpicture}[grow'=down]
    \tikzstyle{level 2}=[sibling distance=+22mm]
    \tikzstyle{level 3}=[sibling distance=+12mm]
    \tikzstyle{level 4}=[sibling distance=+7mm]
    \tikzstyle{level 5}=[sibling distance=+6mm]
    \tikzstyle{level 6}=[level distance=1.2cm]
    \node[diamond] (PO) {51}
      child[pdia] {node[box] {51}
        child[pbox] {node[diamond] (01) {01}
          child[pdia] {node[box] {01}
            child[pbox] {node[diamond,diawon] (back01) {01}}
            child[pbox] {node[diamond] {10}
              child[pdia] {node[box] {10}
                child[pbox] {node[diamond,diawon] {00}}
                child[pbox] {node[diamond,diawon] {00}}
                edge from parent node[action] {repeat}
              }
            child[pdia] {node[box,boxwon] {10} edge from parent node[action] {stop}}
          }
            edge from parent node[action] {repeat}
          }
          child[pdia] {node[box,diawonstrategy] {01} edge from parent node[action] {stop}}
        }
        child[pbox] {node[diamond] {10}
          child[pdia] {node[box] {10}
            child[pbox] {node[diamond,diawonstrategy] {00}}
            child[pbox] {node[diamond,diawonstrategy] {00}}
            edge from parent node[action] {repeat}
          }
          child[pdia] {node[box,boxwon] {10} edge from parent node[action] {stop}}
        }
        edge from parent node[action] {repeat}
      }
      child[pdia] {node[box,boxwon] {51}
        edge from parent node[action] {stop}
      }
      ;
      \draw[backedge] (01) to[bend left=80] (back01);
  \end{tikzpicture}
  \caption{Game trees for counterexample to induction axiom (notation: $x,a$) with game
  \m{\alpha\mequiv\dchoice{(\pupdate{\pumod{x}{a}};\pupdate{\pumod{a}{0}})}{\pupdate{\pumod{x}{0}}}}.
    \textbf{(left)} \m{\dbox{\prepeat{\alpha}}{(x=1\limply\dbox{\alpha}{x=1})}} is true by the strategy ``if Angel chose stop, choose $\pupdate{\pumod{x}{a}};\pupdate{\pumod{a}{0}}$, otherwise always choose $\pupdate{\pumod{x}{0}}$''
    \textbf{(right)} \(\dbox{\prepeat{\alpha}}{x=1}\) is false by the strategy ``repeat once and repeat once more if $x=1$, then stop.''
  If a winning state can be reached by a winning strategy, the mark is enclosed in a circle \usebox{\tmpdiawonstrategy} or \usebox{\tmpboxwonstrategy}, respectively.
  }
    \label{fig:cex-I}
\end{figure}%
The failure of the induction axiom in the counterexample for \rref{eq:induction-axiom} hinges on the fact that Angel is free to decide whether to repeat $\alpha$ after each round depending on the state.
This would be different for an advance notice semantics for $\prepeat{\alpha}$; see \rref{app:alternative-semantics}.
By a variation of the soundness argument for \irref{FP} or the semantic deduction theorem applied to the \irref{invind} rule, it can be shown, however, that a variation of the induction axiom is still sound if the induction rule \irref{invind} is translated into an axiom using the universal closure, denoted $\closureall{}$, with respect to all variables bound in $\alpha$:
\[
\dinferenceRule[allI|\usebox{\closurallI}]{universally closed induction axiom}
{
\closureall{(\phi\limply\dbox{\alpha}{\phi})} \limply (\phi\limply\dbox{\prepeat{\alpha}}{\phi})
}{}
\]

Universal closures do not rescue the first arrival axiom, a dual of induction axiom I:
\[
\cancel{\text{FA}}~~
\ddiamond{\prepeat{\alpha}}{\phi} \limply \phi \lor \ddiamond{\prepeat{\alpha}}{(\lnot\phi\land\ddiamond{\alpha}{\phi})}
\]
This axiom holds for \dL.
It expresses that, if $\phi$ holds after a repetition of $\alpha$, then it either holds right away or $\alpha$ can be repeated so that $\phi$ does not hold yet but can hold after one more repetition \cite{DBLP:journals/sLogica/PaulyP03a}.
This axiom does not hold, however, for \dGL as witnessed by
\(\alpha\mequiv((\dchoice{\pupdate{\pumod{x}{x-y}}}{\pupdate{\pumod{x}{0}}});\pupdate{\pumod{y}{x}})\) and \(\phi\mequiv (x=0)\), since two iterations surely yield $x=0$, but one iteration may or may not yield $x=0$, depending on Demon's choice; see \rref{fig:no-first-arrival}.
Observe how the failure of the first arrival axiom in \dGL relates to the impossibility of predicting precise enough repetition counts in hybrid games (recall corresponding discussions for \rref{thm:dGL-closure-lower}, \rref{sec:dGL-complete}, and \rref{app:alternative-semantics}).

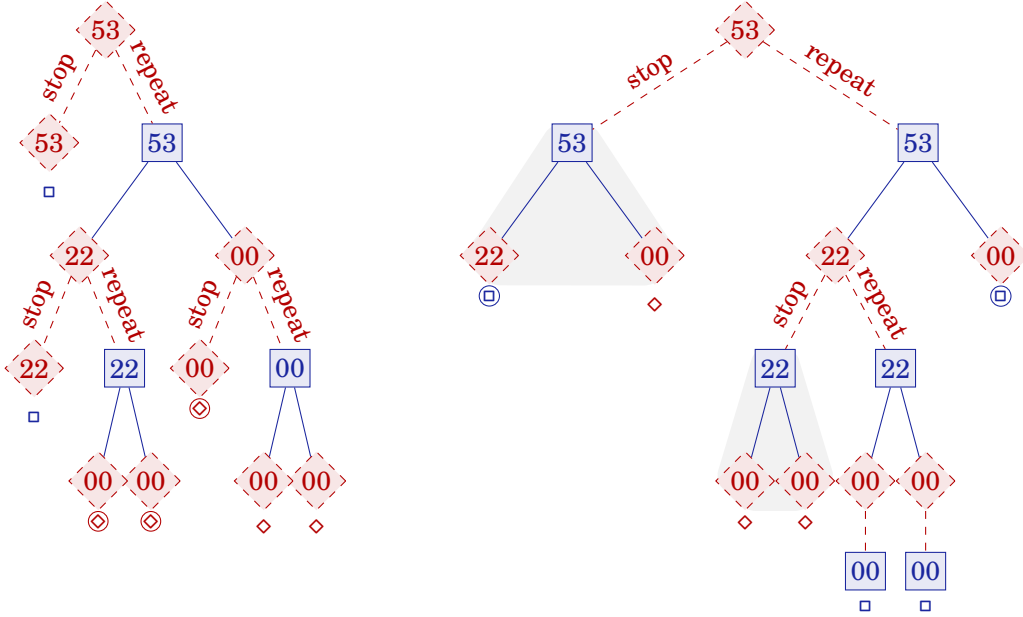
\begin{figure}[pt!bh]
    \tikzstyle{box}+=[minimum width=0.5cm,minimum height=0.5cm]
  \begin{tikzpicture}[grow'=down]
  \begin{scope}
    \tikzstyle{level 2}=[sibling distance=+22mm]
    \tikzstyle{level 3}=[sibling distance=+12mm]
    \tikzstyle{level 4}=[sibling distance=+7mm]
    \tikzstyle{level 5}=[sibling distance=+6mm]
    \node[diamond] (PO) {53}
      child[pdia] {node[box] {53}
        child[pbox] {node[diamond] (00) {00}
          child[pdia] {node[box] {00}
            child[pbox] {node[diamond,diawon] (back00) {00}}
            child[pbox] {node[diamond,diawon] {00}}
            edge from parent node[action] {repeat}
          }
          child[pdia] {node[diamond,diawonstrategy] {00} edge from parent node[action] {stop}}
        }
        child[pbox] {node[diamond] {22}
          child[pdia] {node[box] {22}
            child[pbox] {node[diamond,diawonstrategy] {00}}
            child[pbox] {node[diamond,diawonstrategy] {00}}
            edge from parent node[action] {repeat}
          }
          child[pdia] {node[diamond,boxwon] {22} edge from parent node[action] {stop}}
        }
        edge from parent node[action] {repeat}
      }
      child[pdia] {node[diamond,boxwon] {53}
        edge from parent node[action] {stop}
      }
      ;
  \end{scope}
  \begin{scope}[xshift=8.5cm]
    \tikzstyle{level 1}=[sibling distance=+46mm,set style={{edge from parent}+=[diamonda]},set style={{every node}+=[or]}]
    \tikzstyle{level 2}=[sibling distance=+22mm,set style={{edge from parent}+=[boxa]},set style={{every node}+=[and]}]
    \tikzstyle{level 3}=[sibling distance=+16mm,set style={{edge from parent}+=[diamonda]},set style={{every node}+=[or]}]
    \tikzstyle{level 4}=[sibling distance=+8mm,set style={{edge from parent}+=[boxa],set style={{every node}+=[and]}}]
    \tikzstyle{level 5}=[sibling distance=+5mm,set style={{edge from parent}+=[diamonda],set style={{every node}+=[or]}},level distance=1.2cm]
    \tikzstyle{level 6}=[sibling distance=+3mm,set style={{edge from parent}+=[boxa],set style={{every node}+=[and]}}]
    \node[diamond] (PO) {53}
      child {node[box] {53}
        child {node[diamond,boxwonstrategy] {00}}
        child {node[diamond] {22}
          child {node[box] {22} 
            child {node[diamond] {00} child {node[box,boxwon] {00}}}
            child {node[diamond] {00} child {node[box,boxwon] {00}}}
            edge from parent node[action] {repeat}
          }
          child {node[box] {22}
            child {node[diamond,diawon] {00}}
            child {node[diamond,diawon] {00}}
          edge from parent node[action] {stop}
          }
        }
      edge from parent node[action] {repeat}
      }
      child {node [box] {53} 
        child {node[diamond,diawon] {00}}
        child {node[diamond,boxwonstrategy] {22}}
        edge from parent node[action] {stop}
      }
      ;
      \begin{pgfonlayer}{background}
        \draw[endgame] (PO-2.north west) -- (PO-2-2.north west) -- (PO-2-2.west) -- (PO-2-2.south) -- (PO-2-1.south) -- (PO-2-1.east) -- (PO-2-1.north east) -- (PO-2.north east) -- cycle;
        \draw[endgame] (PO-1-2-2.north west) -- (PO-1-2-2-2.west) -- (PO-1-2-2-2.south) -- (PO-1-2-2-1.south) -- (PO-1-2-2-1.east) -- (PO-1-2-2.north east) -- cycle;
      \end{pgfonlayer}
    \end{scope}
  \end{tikzpicture}%
  \caption{Game trees for counterexample to first arrival axiom with game
  \m{\alpha\mequiv(\dchoice{\pupdate{\pumod{x}{x-y}}}{\pupdate{\pumod{x}{0}}});\pupdate{\pumod{y}{x}}}
  (notation: $x,y$).
    \textbf{(left)} \m{\ddiamond{\prepeat{\alpha}}{x=0}} is true no matter what Demon chooses
    \textbf{(right)} \m{\ddiamond{\prepeat{\alpha}}{(x\neq0\land\ddiamond{\alpha}{x=0})}} is false,
    because stop can be defeated by \m{\pupdate{\pumod{x}{x-y}}} and repeat can be defeated by \m{\pupdate{\pumod{x}{0}}}.
    }%
    \label{fig:no-first-arrival}%
\end{figure}%

The hybrid systems axiom version \cite{DBLP:conf/lics/Platzer12b} of Harel's convergence rule \cite{DBLP:conf/stoc/HarelMP77}, in which $v$ does not occur in $\alpha$ (written $v\not\in\alpha$),
\[
\cancel{\text{C}}~~
  \linferenceRule[impl]
        {\dbox{\prepeat{\alpha}}{\lforall{v{>}0}{(\mapply{\var}{v}\limply\ddiamond{\alpha}{\mapply{\var}{v-1}})}}}
        {\lforall{v}{(\mapply{\var}{v} \limply
            \ddiamond{\prepeat{\alpha}}{\lexists{v{\leq}0}{\mapply{\var}{v}}})}}
      \quad({\m{v\not\in\alpha}})%
\]
holds for \dL, but not for \dGL as witnessed by
\m{\alpha \mequiv (\dchoice{\pupdate{\pumod{x}{x-y}};\pupdate{\pumod{y}{0}}}{\pupdate{\pumod{x}{x-1}}})} and
\m{\mapply{\var}{v} \mequiv (x\leq v)}.
In a state where \(y=1,x\geq2\), \m{\dbox{\prepeat{\alpha}}{\lforall{v{>}0}{(\mapply{\var}{v}\limply\ddiamond{\alpha}{\mapply{\var}{v-1}})}}} is true by the strategy ``always choose \m{\pupdate{\pumod{x}{x-1}}}'' for Demon for \m{\dbox{\prepeat{\alpha}}{}} and arbitrary strategies for Angel for the nested \m{\ddiamond{\alpha}{}}.
Yet, \m{\ddiamond{\prepeat{\alpha}}{\lexists{v{\leq}0}{\mapply{\var}{v}}}} is false by the strategy ``always choose \m{\pupdate{\pumod{x}{x-y}};\pupdate{\pumod{y}{0}}}'', because $x$ will no longer change after the first iteration then.
The hybrid version of Harel's convergence rule is sound but unnecessary (\rref{sec:separating-axioms}).

This completes the proof of \rref{thm:separating-axioms} by inspecting the complete axiomatization of \dGL from \rref{thm:dGL-complete} compared to the complete axiomatization of hybrid systems \cite{DBLP:conf/lics/Platzer12b}.
Each axiom of hybrid systems has been considered and either continues to hold for hybrid games (\rref{thm:dGL-sound}) or has been refuted with a counterexample (K, I, C, B, V, G) or continues to hold but is unnecessary for completeness (Harels' convergence rule).
A few additional axioms and rules that are not part of the hybrid systems axiomatization have been considered for illustration purposes, because they are closely related and highlight interesting aspects of the axiomatic similarity (axiom \irref{Maxiom}, \irref{allI}, \irref{cB}, \irref{VK}, $\text{M}_{\dbox{\cdot}{}}$) or difference ($\overleftarrow{\text{M}}$, R, FA) between hybrid systems and hybrid games.

\section{Operational Game Semantics} \label{app:operational-HG-semantics}

\newcommand{\aleft}{\mathfrak{l}}%
\newcommand{\aright}{\mathfrak{r}}%
\newcommand{\astop}{\mathfrak{s}}%
\newcommand{\arepeat}{\mathfrak{g}}%
\newcommand{\adual}{\mathfrak{d}}%
\newcommand{\aappend}{\text{\textasciicircum}}%
\newcommand*{\aevolvein}[3]{\pevolvein{#1}{#2}@#3}%
\newcommand{\astutter}{\mathsf{f}}%
\newcommand{\gameplay}[2][]{\mathsf{g}(#1)(#2)}%
\newcommand{\igameplay}[2][]{\gameplay[#1]{\iportray{#2}}}%
\newcommand{\play}[2][]{\lfloor#1\rfloor_{#2}}%
\newcommand{\iplay}[2][]{\play[#1]{\iportray{#2}}}%

\newcommand*{\fulltree}[1]{}%

In order to relate the intuition of interactive game play to the denotational semantics of hybrid games, this section shows an operational semantics for hybrid games that is more complicated than the modal semantics from \rref{sec:dGL-semantics} but makes strategies explicit and directly reflects the intuition how hybrid games are played interactively.
The modal semantics is beneficial, because it is simpler.
The results in this section are not needed in the rest of the paper and play an informative role.
The operational semantics formalizes the intuition behind the game tree in \rref{fig:nondetermined} and relates to standard notions in game theory and descriptive set theory.
\rref{thm:equivalent-semantics} below proves that the operational game semantics is equivalent to the modal semantics from \rref{sec:dGL-semantics}.
The (denotational) modal semantics is much simpler but the operational semantics makes winning strategies explicit.
As the set of actions $A$ for a hybrid game choose:
\begin{multline*}
\{\aleft,\aright,\astop,\arepeat,\adual\}
\cup \{(\pupdate{\pumod{x}{\theta}}) \with x~\text{variable},~\theta~\text{term}\}
\cup \{\ptest{\ivr} \with \ivr~\text{formula}\}
\\\cup \{(\aevolvein{\D{x}=\theta}{\ivr}{r}) \with x~\text{variable},~\theta~\text{term}, \ivr~\text{formula}, r\in\reals_{\geq0}\}
\end{multline*}
For game $\pchoice{\alpha}{\beta}$, action $\aleft$ decides to descend left into $\alpha$, $\aright$ is the action of descending right into $\beta$.
In game $\prepeat{\alpha}$, action $\astop$ decides to stop repeating, action $\arepeat$ decides to go back and repeat.
Action $\adual$ starts and ends a dual game for $\pdual{\alpha}$.
The other actions represent the actions for atomic games: assignment actions, continuous evolution actions (in which time $r$ is the critical decision), and test actions.

The operational game semantics uses standard notions from descriptive set theory \cite{Kechris}.
The set of finite sequences of actions is denoted by $A^{(\naturals)}$, the set of countably infinite sequences by $A^\naturals$.
The empty sequence of actions is $()$.
The concatenation, $s\aappend t$, of sequences $s,t\in A^{(\naturals)}$ is defined as \((s_1,\dots,s_n,t_1,\dots,t_m)\) if  \(s=(s_1,\dots,s_n)\) and \(t=(t_1,\dots,t_m)\).
For an $a\in A$, write \(a\aappend t\) for \((a)\aappend t\) and write \(t\aappend a\) for \(t\aappend (a)\).
For a set $S\subseteq A^{(\naturals)}$,  write \(S\aappend t\) for \(\{s\aappend t \with s\in S\}\) and $t\aappend S$ for \(\{t\aappend s \with s\in S\}\).
The state $\iplay[t]{\I}$ reached by \emph{playing} a sequence of actions $t\in A^{(\naturals)}$ from a state $\iget[state]{\I}$ in interpretation $\iget[const]{\I}$ is inductively defined by applying the actions sequentially, i.e.\ as follows:
\begin{enumerate}
\item \(\iplay[\pupdate{\pumod{x}{\theta}}]{\I} = \modif{\iget[state]{\I}}{x}{\ivaluation{\I}{\theta}}\)
\item \(\iplay[\aevolvein{\D{x}=\genDE{x}}{\ivr}{r}]{\I} = \varphi(r)\) for the unique
      \m{\varphi:[0,r]\to\linterpretations{\Sigma}{V}}
      differentiable, 
      \m{\varphi(0)=\iget[state]{\I}},
      \m{\D[t]{\,\varphi(t)(x)} (\zeta) =       %
      \ivaluation{\iconcat[state=\varphi(\zeta)]{\I}}{\theta}}
      and
      \m{\varphi(\zeta)\in\imodel{\I}{\ivr}}
      for all $\zeta\leq r$.
      Note that \(\iplay[\aevolvein{\D{x}=\genDE{x}}{\ivr}{r}]{\I}\) is not defined if no such $\varphi$ of duration $r$ exists.
\item \(\iplay[\ptest{\ivr}]{\I} =
\begin{cases}
\iportray{\I} &\text{if}~\iportray{\I}\in\imodel{\I}{\ivr}\\
\text{not defined} &\text{otherwise}
\end{cases}\)
\item \(\iplay[\aleft]{\I} = \iplay[\aright]{\I} = \iplay[\astop]{\I} = \iplay[\arepeat]{\I} = \iplay[\adual]{\I} = \iplay[()]{\I} = \iportray{\I}\)
\item \(\iplay[a\aappend t]{\I} = \play[t]{(\iplay[a]{\I})}\) for $a\in A$ and $t\in A^{(\naturals)}$
\end{enumerate}
A \emph{tree} is a set $T\subseteq A^{(\naturals)}$ that is closed under prefixes, that is, whenever $t\in T$ and $s$ is a prefix of $t$ (i.e.\ $t=s\aappend r$ for some $r\in A^{(\naturals)}$), then $s\in T$.
A node \m{t\in T} is a successor of node \m{s\in T} iff $t=s\aappend a$ for some $a\in A$.
Denote by \(\leaf(T)\) the set of all leaves of $T$, i.e.\ nodes $t\in T$ that have no successor in $T$.
\begin{definition}[Operational game semantics] \label{def:HG-operational-semantics}
The \emph{operational game semantics} of hybrid game $\alpha$ is, for each state $\iget[state]{\I}$ of each interpretation $\iget[const]{\I}$, a tree \m{\igameplay[\alpha]{\I}\subseteq A^{(\naturals)}} defined as follows (see \rref{fig:gameplay} for a schematic illustration):%
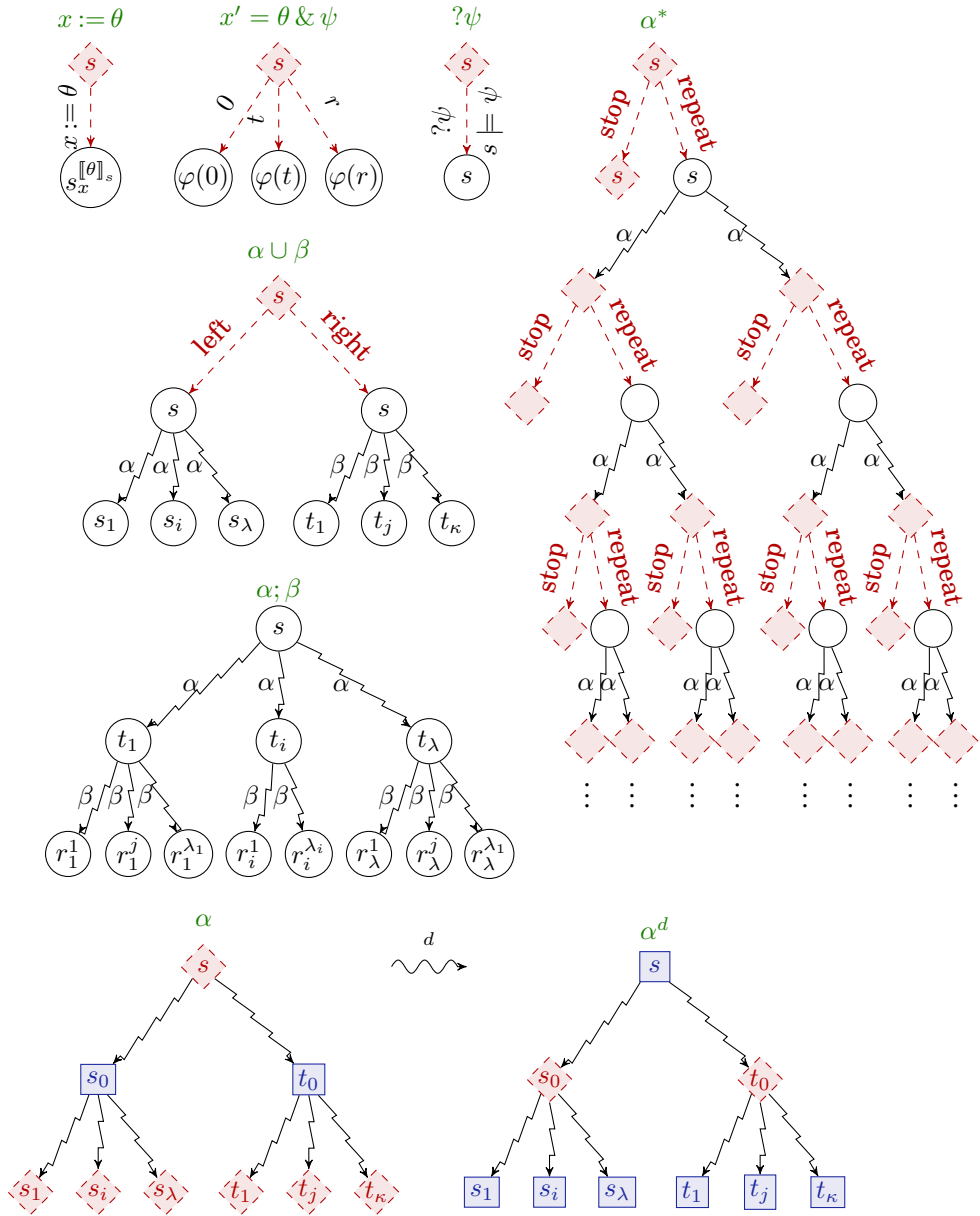
\begin{figure*}[btp]%
  \centering
  \vspace{-2\baselineskip}
  \begin{tikzpicture}[grow'=down,every node/.style={someone,minimum width=+6mm,minimum height=+6mm,inner sep=+0pt},>=stealth']
    \tikzstyle{diamond}+=[minimum width=6mm,minimum height=+6mm,inner sep=+0pt]
    \tikzstyle{diamonda}+=[->]
    \tikzstyle{subgamea}+=[->]
    \tikzstyle{level 1}=[sibling distance=+10mm]
    \node[diamond,label={[gamelabel,above=-4pt]{$\pupdate{\pumod{x}{\theta}}$}}] at (-4.5,+0) {$\iportray{\I}$}
      child[pdia] {node[inner sep=+0] {$\modif{\iget[state]{\I}}{x}{\ivaluation{\I}{\theta}}$} edge from parent node[action] {$\pupdate{\pumod{x}{\theta}}$}};
    \node[diamond,label={[gamelabel,above=-14pt]{$\pevolvein{\D{x}={\theta}}{\ivr}$}}] at (-2,+0) {$\iportray{\I}$}
      child[pdia] {node[inner sep=+0] {$\varphi(r)$} edge from parent node[action] {$r$}}
      child[pdia] {node[inner sep=+0] {$\varphi(t)$} edge from parent node[action] {$t$}}
      child[pdia] {node[inner sep=+0] {$\varphi(0)$} edge from parent node[action] {$0$}};
    \node[diamond,label={[gamelabel]{$\ptest{\ivr}$}}] at (+0.5,+0) {$\iportray{\I}$} %
      child[pdia] {node {$\iportray{\I}$} edge from parent node[action] {$\ptest{\ivr}$} node[action,below,sloped] {$\imodels{\I}{\ivr}$}};
   \begin{scope}[set style={{every node}+=[solid,black,minimum width=6mm,inner sep=0pt]}]
    \tikzstyle{level 1}=[sibling distance=+28mm]   
    \tikzstyle{level 2}=[sibling distance=+9mm]   
    \node[diamond,label={[gamelabel,above=-4pt]{$\pchoice{\alpha}{\beta}$}}]  at (-2,-3.1) {$\iportray{\I}$}
      child {node {$\iportray{\I}$} 
        child {node {$t_\kappa$} edge from parent [subgamea] node[subgame] {$\beta$}}
        child {node {$t_j$} edge from parent [subgamea]  node[subgame] {$\beta$}}
        child {node {$t_1$} edge from parent [subgamea] node[subgame] {$\beta$}}
        edge from parent [diamonda] node[diamonda,action] {right}}
      child {node {$\iportray{\I}$} 
        child {node {$s_\lambda$} edge from parent [subgamea] node[subgame] {$\alpha$}}
        child {node {$s_i$} edge from parent [subgamea]  node[subgame] {$\alpha$}}
        child {node {$s_1$} edge from parent [subgamea] node[subgame] {$\alpha$}}
        edge from parent [diamonda] node[diamonda,action] {left}};
    \end{scope}
   \begin{scope}[set style={{every node}+=[minimum width=+6mm,inner sep=+0pt]}]
    \tikzstyle{level 1}=[sibling distance=20mm]   
    \tikzstyle{level 2}=[sibling distance=8mm]   
    \node[label={[gamelabel,above=-4pt]{$\alpha;\beta$}}]  at (-2,-7.5) {$\iportray{\I}$}   %
      child {node {$t_\lambda$} 
        child {node {$r_\lambda^{\lambda_1}$} edge from parent [subgamea] node[subgame] {$\beta$}}
        child {node {$r_\lambda^j$} edge from parent [subgamea]  node[subgame] {$\beta$}}
        child {node {$r_\lambda^1$} edge from parent [subgamea] node[subgame] {$\beta$}}
        edge from parent [subgamea] node[subgame] {$\alpha$}}
      child {node {$t_i$}
        child {node {$r_i^{\lambda_i}$} edge from parent [subgamea] node[subgame] {$\beta$}}
        child {node {$r_i^1$} edge from parent [subgamea] node[subgame] {$\beta$}}
        edge from parent [subgamea]  node[subgame] {$\alpha$}}
      child {node {$t_1$} 
        child {node {$r_1^{\lambda_1}$} edge from parent [subgamea] node[subgame] {$\beta$}}
        child {node {$r_1^j$} edge from parent [subgamea]  node[subgame] {$\beta$}}
        child {node {$r_1^1$} edge from parent [subgamea] node[subgame] {$\beta$}}
        edge from parent [subgamea] node[subgame] {$\alpha$}};
    \end{scope}
  \begin{scope}[grow'=down]
    \tikzstyle{someone}+=[minimum width=+5mm,minimum height=+5mm,inner sep=+0pt]
    \tikzstyle{level 2}=[sibling distance=+29mm]
    \tikzstyle{level 3}=[sibling distance=+15mm]
    \tikzstyle{level 4}=[sibling distance=+14mm]
    \tikzstyle{level 5}=[sibling distance=+6mm]
    \tikzstyle{action}+=[diamonda]
    \node[diamond,label={[gamelabel]{$\prepeat{\alpha}$}}] at (+3,+0) {$\iportray{\I}$}
      child[pdia] {node[someone] {$\iportray{\I}$}
        child[psubgame] {node[diamond] {}
          child[pdia] {node[someone] {}
            child[psubgame] {node[diamond] {}
              child[pdia] {node[someone] {}
                child[psubgame] {node[diamond,etc] {} edge from parent node[subgame] {$\alpha$}}
                child[psubgame] {node[diamond,etc] {} edge from parent node[subgame] {$\alpha$}}
                edge from parent node[action] {repeat}
              }
            child[pdia] {node[diamond] {} edge from parent node[action] {stop}}
            edge from parent node[subgame] {$\alpha$}}
            child[psubgame] {node[diamond] {}
              child[pdia] {node[someone] {}
                child[psubgame] {node[diamond,etc] {} edge from parent node[subgame] {$\alpha$}}
                child[psubgame] {node[diamond,etc] {} edge from parent node[subgame] {$\alpha$}}
                edge from parent node[action] {repeat}
              }
            child[pdia] {node[diamond] {} edge from parent node[action] {stop}}
            edge from parent node[subgame] {$\alpha$}
          }
            edge from parent node[action] {repeat}
          }
          child[pdia] {node[diamond] {} edge from parent node[action] {stop}}
          edge from parent node[subgame] {$\alpha$}
        }
        child[psubgame] {node[diamond] {}
          child[pdia] {node[someone] {}
            child[psubgame] {node[diamond] {}
              child[pdia] {node[someone] {}
                child[psubgame] {node[diamond,etc] {} edge from parent node[subgame] {$\alpha$}}
                child[psubgame] {node[diamond,etc] {} edge from parent node[subgame] {$\alpha$}}
                edge from parent node[action] {repeat}
              }
            child[pdia] {node[diamond] {} edge from parent node[action] {stop}}
            edge from parent node[subgame] {$\alpha$}}
            child[psubgame] {node[diamond] {}
              child[pdia] {node[someone] {}
                child[psubgame] {node[diamond,etc] {} edge from parent node[subgame] {$\alpha$}}
                child[psubgame] {node[diamond,etc] {} edge from parent node[subgame] {$\alpha$}}
                edge from parent node[action] {repeat}
              }
            child[pdia] {node[diamond] {} edge from parent node[action] {stop}}
            edge from parent node[subgame] {$\alpha$}}
            edge from parent node[action] {repeat}
          }
          child[pdia] {node[diamond] {} edge from parent node[action] {stop}}
          edge from parent node[subgame] {$\alpha$}
        }
        edge from parent node[action] {repeat}
      }
      child[pdia] {node[diamond] {$\iportray{\I}$}
        edge from parent node[action] {stop}
      }
      ;
    \end{scope}
   \begin{scope}[set style={{every node}+=[solid,black,minimum width=6mm,inner sep=0pt]},yshift=-12cm]
    \tikzstyle{level 1}=[sibling distance=+28mm]   
    \tikzstyle{level 2}=[sibling distance=+9mm]   
    \node[diamond,label={[gamelabel]{$\alpha$}}] (alpha) at (-3,+0)
     {$\iportray{\I}$}
      child {node[box] {$t_0$} 
        child {node[diamond] {$t_\kappa$} edge from parent [subgamea]}
        child {node[diamond] {$t_j$} edge from parent [subgamea]}
        child {node[diamond] {$t_1$} edge from parent [subgamea]}
        edge from parent [subgamea]}
      child {node[box] {$s_0$} 
        child {node[diamond] {$s_\lambda$} edge from parent [subgamea]}
        child {node[diamond] {$s_i$} edge from parent [subgamea]}
        child {node[diamond] {$s_1$} edge from parent [subgamea]}
        edge from parent [subgamea]};
    \node[box,label={[gamelabel]{$\pdual{\alpha}$}}] (alphad) at (3,+0)
     {$\iportray{\I}$}
      child {node[diamond] {$t_0$} 
        child {node[box] {$t_\kappa$} edge from parent [subgamea]}
        child {node[box] {$t_j$} edge from parent [subgamea]}
        child {node[box] {$t_1$} edge from parent [subgamea]}
        edge from parent [subgamea]}
      child {node[diamond] {$s_0$} 
        child {node[box] {$s_\lambda$} edge from parent [subgamea]}
        child {node[box] {$s_i$} edge from parent [subgamea]}
        child {node[box] {$s_1$} edge from parent [subgamea]}
        edge from parent [subgamea]};
      \draw[decorate,decoration={coil,aspect=0},->] (-0.5,0) -- node[above,draw=none]{$\pdual{}$} (0.5,0);
    \end{scope}
  \end{tikzpicture}
  \caption{Operational game semantics for hybrid games of \dGL}
  \label{fig:gameplay}
\end{figure*}%

\begin{enumerate}
\item \(\igameplay[\pupdate{\pumod{x}{\theta}}]{\I} = \{\fulltree{(),}(\pupdate{\pumod{x}{\theta}})\}\)
\item \(\igameplay[\pevolvein{\D{x}=\genDE{x}}{\ivr}]{\I} = \{\fulltree{(),}(\aevolvein{\D{x}=\genDE{x}}{\ivr}{r}) \with 
      r\in\reals, r\geq0, \varphi(0)=\iget[state]{\I}\)
      for some (differentiable)
      \m{\varphi:[0,r]\to\linterpretations{\Sigma}{V}}
      such that
      \m{\D[t]{\,\varphi(t)(x)} (\zeta) =       %
      \ivaluation{\iconcat[state=\varphi(\zeta)]{\I}}{\theta}}
      and
      \m{\varphi(\zeta)\in\imodel{\I}{\ivr}}
      for all $\zeta\leq r\}$
\item \(\igameplay[\ptest{\ivr}]{\I} = \{\fulltree{(),}(\ptest{\ivr})\}\)
\item \(\igameplay[\pchoice{\alpha}{\beta}]{\I} = 
\fulltree{\{(),(\aleft),(\aright)\}\cup{}}\aleft\aappend\igameplay[\alpha]{\I} \cup \aright\aappend\igameplay[\beta]{\I}\)
\item \(\displaystyle\igameplay[\alpha;\beta]{\I} = \igameplay[\alpha]{\I} \cup \cupfold_{t\in\leaf(\igameplay[\alpha]{\I})}{\gameplay[\beta]{\iplay[t]{\I}}}\)
\item
\(\displaystyle\igameplay[\prepeat{\alpha}]{\I} = \cupfold_{n<\omega} f^n(\{\fulltree{(),}(\astop),(\arepeat)\})\)
\\where $f^n$ is the $n$-fold composition of the function\\
\(f(Z) \mdefeq Z \cup \cupfold_{t\aappend\arepeat\in\leaf(Z)} 
t\aappend\arepeat\aappend\gameplay[\alpha]{\iplay[t\aappend\arepeat]{\I}}\aappend\{\fulltree{(),}(\astop),(\arepeat)\}\)

\item \(\igameplay[\pdual{\alpha}]{\I} = \fulltree{\{(),(\adual)\}\cup{}} 
\fulltree{\adual\aappend\igameplay[\alpha]{\I} \cup{}}
\adual\aappend\igameplay[\alpha]{\I}\aappend\adual\)
\end{enumerate}
\end{definition}
Note the implicit closure under prefixes in the definition of \(\igameplay[\alpha]{\I}\) for readability.
For example, \(\igameplay[\pdual{\alpha}]{\I} = \adual\aappend\igameplay[\alpha]{\I}\aappend\adual\) means \(\igameplay[\pdual{\alpha}]{\I} = \{(),(\adual)\}\cup 
\adual\aappend\igameplay[\alpha]{\I} \cup
\adual\aappend\igameplay[\alpha]{\I}\aappend\adual\).

Angel gets to choose which action to take at node \m{t\in\igameplay[\alpha]{\I}} if $t$ has an even number of occurrences of $\adual$, otherwise Demon gets to choose.
In the former case \emph{Angel acts at $t$}, in the latter \emph{Demon acts at $t$}.
Thus, at every $t$, exactly one of the players acts at $t$.
If the player who acts at $t$ is deadlocked, then that player loses immediately.
A player who acts at \m{t\in\igameplay[\alpha]{\I}} is \emph{deadlocked} at $t$ if \m{t\not\in\leaf(\igameplay[\alpha]{\I})} and no successor $s$ is enabled, i.e.\ $\iplay[s]{\I}$ is not defined.
This can happen if the last action in $s$ has a condition that is not satisfied like $\ptest{x\geq0}$ or \m{\pevolvein{\D{x}=\theta}{x\geq0}} at a state where $x<0$.
Note that the player who acts at \m{t\in\igameplay[\prepeat{\alpha}]{\I}} cannot choose $\arepeat$ infinitely often for that loop because $n<\omega$.

The players use Markov strategies, i.e.\ their choices only depend on the current state of the system and they have no additional information about the strategy of the other player.
A \emph{strategy for Angel} from initial state $\iportray{\I}$ is a nonempty subtree \m{\sigma\subseteq\igameplay[\alpha]{\I}} that accepts all of Demon's actions at nodes $t$ where Demon acts and selects a unique Angel action when Angel acts at $t$:
\begin{enumerate}
\item for all $t\in\sigma$ at which Demon acts, $t\aappend a\in\sigma$ for all $a\in A$ such that $t\aappend a\in\igameplay[\alpha]{\I}$.
\item for all $t\in\sigma$ at which Angel acts, if \(t\not\in\leaf(\igameplay[\alpha]{\I})\), then there is a unique $a\in A$ with $t\aappend a\in\sigma$.
\end{enumerate}
Strategies for Demon are defined accordingly, with ``Angel'' and ``Demon'' swapped.
The action sequence \(\sigma\oplus\tau\) played from state $\iget[state]{\I}$ in interpretation $\iget[const]{\I}$ when Angel plays strategy $\sigma$ and Demon plays strategy $\tau$ from $\iportray{\I}$ is defined as the sequence \((a_1,\dots,a_n)\in A^{(\naturals)}$ of maximal length such that
\[
  a_{n+1} :=
  \begin{cases}
  a &\text{if Angel acts at $(a_1,\dots,a_n)$ and $(a_1,\dots,a_n)\aappend a\in\sigma$}\\
  a &\text{if Demon acts at $(a_1,\dots,a_n)$ and $(a_1,\dots,a_n)\aappend a\in\tau$}\\
  \text{not defined} &\text{otherwise}
  \end{cases}
\]
By definition of strategies, \(\sigma\oplus\tau\) is unique.
A \emph{winning strategy for Angel} for winning condition $X\subseteq\linterpretations{\Sigma}{V}$ from state $\iget[state]{\I}$ in interpretation $\iget[const]{\I}$ is a strategy \m{\sigma\subseteq\igameplay[\alpha]{\I}} for Angel from $\iportray{\I}$ such that, for all strategies \m{\tau\subseteq\igameplay[\alpha]{\I}} for Demon from $\iportray{\I}$: Demon deadlocks or \m{\iplay[\sigma\oplus\tau]{\I}\in X}.
A \emph{winning strategy for Demon} for (Demon's) winning condition $X\subseteq\linterpretations{\Sigma}{V}$ from state $\iget[state]{\I}$ in interpretation $\iget[const]{\I}$ is a strategy \m{\tau\subseteq\igameplay[\alpha]{\I}} for Demon from $\iportray{\I}$ such that, for all strategies \m{\sigma\subseteq\igameplay[\alpha]{\I}} for Angel from $\iportray{\I}$: Angel deadlocks or \m{\iplay[\sigma\oplus\tau]{\I}\in X}.

The denotational modal semantics (\rref{sec:dGL-semantics}) is equivalent to the operational semantics:

\begin{theorem}[Equivalent semantics] \label{thm:equivalent-semantics}%
  The modal semantics of \dGL is equivalent to the operational game-tree semantics of \dGL, i.e.\
  for each hybrid game $\alpha$, each initial state $\iportray{\I}$ in each interpretation $\iget[const]{\I}$, and each winning condition $X\subseteq\linterpretations{\Sigma}{V}$:
  \begin{align*}
  \iportray{\I}\in\strategyfor[\alpha]{X} &\Longleftrightarrow
  \text{there is a winning strategy \(\sigma\subseteq\igameplay[\alpha]{\I}\)}~%
  \text{for Angel to achieve $X$ from $\iportray{\I}$}\\
  \iportray{\I}\in\dstrategyfor[\alpha]{\scomplement{X}} &\Longleftrightarrow
  \text{there is a winning strategy \(\tau\subseteq\igameplay[\alpha]{\I}\)}~%
  \text{for Demon to achieve $\scomplement{X}$ from $\iportray{\I}$}
  \end{align*}
\end{theorem}
\begin{proof}
Proceed by simultaneous induction on the structure of $\alpha$ and prove equivalence. As part of the equivalence proof, construct a winning strategy $\sigma$ achieving $X$ using that \(\iportray{\I}\in\strategyfor[\alpha]{X}\).
The simultaneous induction steps for $\dstrategyfor[\alpha]{\scomplement{X}}$ are simple dualities.
It is easy to see that Angel and Demon cannot both have a winning strategy from the same state $\iget[state]{\I}$ for complementary winning conditions $X$ and $\scomplement{X}$ in the same game \(\igameplay[\alpha]{\I}\).
\rref{thm:dGL-determined} implies \(\dstrategyfor[\alpha]{\scomplement{X}}=\scomplement{\strategyfor[\alpha]{X}}\).
\begin{enumerate}
\item \(\iportray{\I}\in\strategyfor[\pupdate{\pumod{x}{\theta}}]{X} \mbisubjunct \modif{\iget[state]{\I}}{x}{\ivaluation{\I}{\theta}} \in X\)
\(\mbisubjunct \iplay[\sigma\oplus\tau]{\I}=\iplay[\pupdate{\pumod{x}{\theta}}]{\I} = \modif{\iget[state]{\I}}{x}{\ivaluation{\I}{\theta}} \in X\),
using \(\sigma\mdefeq \{(\pupdate{\pumod{x}{\theta}})\} = \igameplay[\pupdate{\pumod{x}{\theta}}]{\I}\).
The converse direction follows, because the strategy $\sigma$ follows the only permitted strategy.

\item \(\iportray{\I}\in\strategyfor[\pevolvein{\D{x}=\genDE{x}}{\ivr}]{X} \mbisubjunct s=\varphi(0), 
      \varphi(r)\in X\)
      for some $r\in\reals$ and some (differentiable)
      \m{\varphi:[0,r]\to\linterpretations{\Sigma}{V}}
      such that
      \m{\D[t]{\,\varphi(t)(x)} (\zeta) =       %
      \ivaluation{\iconcat[state=\varphi(\zeta)]{\I}}{\theta}}
      and
      \m{\varphi(\zeta)\in\imodel{\I}{\ivr}}
      for all $\zeta\leq r$
      \(\mbisubjunct \iplay[\sigma\oplus\tau]{\I}=\iplay[\aevolvein{\D{x}=\theta}{\ivr}{r}]{\I} = \varphi(r) \in X\),
using \(\sigma\mdefeq \{(\aevolvein{\D{x}=\theta}{\ivr}{r})\} \subseteq \igameplay[\pevolvein{\D{x}=\genDE{x}}{\ivr}]{\I}\).
The converse direction follows, since this $\sigma$ has the only permitted form for a strategy and different values of $r$ that lead to $X$ are equally useful.

\item \(\iportray{\I}\in\strategyfor[\ptest{\ivr}]{X} = \imodel{\I}{\ivr}\cap X\)
\(\mbisubjunct \iplay[\sigma\oplus\tau]{\I}=\iplay[\ptest{\ivr}]{\I}=\iportray{\I}\in X\),
with \(\iget[state]{\I} \in \imodel{\I}{\ivr}\) using \(\sigma\mdefeq \{(\ptest{\ivr})\} = \igameplay[\ptest{\ivr}]{\I}\).
The converse direction uses that this $\sigma$ is the only permitted strategy and it deadlocks exactly if \(\iget[state]{\I} \not\in \imodel{\I}{\ivr}\).

\item \(\iportray{\I}\in\strategyfor[\pchoice{\alpha}{\beta}]{X} = \strategyfor[\alpha]{X}\cup\strategyfor[\beta]{X}\)
\(\mbisubjunct \iportray{\I}\in\strategyfor[\alpha]{X}\) or \(\iportray{\I}\in\strategyfor[\beta]{X}\).
By induction hypothesis, this is equivalent to:
there is a winning strategy \(\sigma_\alpha\subseteq\igameplay[\alpha]{\I}\) for Angel for $X$ from $\iportray{\I}$
or
there is a winning strategy \(\sigma_\beta\subseteq\igameplay[\beta]{\I}\) for Angel for $X$ from $\iportray{\I}$.
This is equivalent to \(\sigma\subseteq\igameplay[\pchoice{\alpha}{\beta}]{\I}\) being a winning strategy for Angel for $X$ from $\iportray{\I}$,
using \(\sigma\mdefeq \fulltree{\{(\aleft)\} \cup{}} \aleft\aappend\sigma_\alpha\) or \(\sigma\mdefeq \fulltree{\{(\aright)\} \cup{}} \aright\aappend\sigma_\beta\), respectively.

\item \(\iportray{\I}\in\strategyfor[\alpha;\beta]{X} = \strategyfor[\alpha]{\strategyfor[\beta]{X}}\).
By induction hypothesis, this is equivalent to the existence of a strategy \(\sigma_\alpha\subseteq\igameplay[\alpha]{\I}\) for Angel such that for all strategies \m{\tau\subseteq\igameplay[\alpha]{\I}} for Demon: \m{\iplay[\sigma_\alpha\oplus\tau]{\I}\in \strategyfor[\beta]{X}}.
By induction hypothesis, \m{\iplay[\sigma_\alpha\oplus\tau]{\I}\in \strategyfor[\beta]{X}} is equivalent to the existence of a winning strategy $\sigma_\tau$ for Angel (which depends on the state $\iplay[\sigma_\alpha\oplus\tau]{\I}$ that the previous $\alpha$ game led to) with winning condition $X$ from \(\iplay[\sigma_\alpha\oplus\tau]{\I}\).
This is equivalent to \(\sigma\subseteq\igameplay[\alpha;\beta]{\I}\) being a winning strategy for Angel for $X$ from $\iportray{\I}$,
using
\begin{equation}
\sigma\mdefeq \sigma_\alpha \cup \cupfold (\sigma_\alpha\oplus\tau)\aappend\sigma_\tau
\label{eq:equivalent-semantics-compose}
\end{equation}
The union is over all leaves \m{\sigma_\alpha\oplus\tau \in\leaf(\igameplay[\alpha]{\I})} for which the game is not won by a player yet.
Note that $\sigma$ is a winning strategy for $X$, because, for all plays for which the game is decided during $\alpha$, the strategy $\sigma_\alpha$ already wins the game. For the others, $\sigma_\tau$ wins the game from the respective state \m{\iplay[\sigma_\alpha\oplus\tau]{\I}} that (when $\alpha$ terminates) was reached by the actions \m{\sigma_\alpha\oplus\tau} according to the strategy $\tau$ that Demon was observed to have played during $\alpha$.
The converse direction uses that strategies do not depend on moves that have not been played yet and that all strategies can be factorized by prefixes of what has actually been played to be coerced into the form \rref{eq:equivalent-semantics-compose}.

\item
Both inclusions of the case $\prepeat{\alpha}$ are proved separately.
If $W$ denotes the set of states from which Angel has a winning strategy in $\igameplay[\prepeat{\alpha}]{\I}$ to achieve $X$, then
need to show that \m{\strategyfor[\prepeat{\alpha}]{X}=W}.
For \m{\strategyfor[\prepeat{\alpha}]{X}\subseteq W}, it is enough to show that $W$ is a pre-fixpoint, i.e.\ \(X\cup\strategyfor[\alpha]{W}\subseteq W\), because \m{\strategyfor[\prepeat{\alpha}]{X}} is the least (pre-)fixpoint.
Consider a \(\iget[state]{\I}\in X\cup\strategyfor[\alpha]{W}\).
If \(\iget[state]{\I}\in X\) then \(\iget[state]{\I}\in W\) with the winning strategy \(\sigma\mdefeq\{(\astop)\}\) for Angel to achieve $X$ in $\prepeat{\alpha}$ from $\iportray{\I}$.
Otherwise, \(\iget[state]{\I}\in \strategyfor[\alpha]{W}\) implies, by induction hypothesis, that there is a winning strategy $\sigma_\alpha\subseteq\igameplay[\alpha]{\I}$ for Angel in $\alpha$ to achieve $W$ from $\iportray{\I}$.
By definition of $W$, Angel has a winning strategy in $\igameplay[\prepeat{\alpha}]{\I}$ to achieve $X$ from all states reached after playing $\alpha$ from $\iportray{\I}$ according to $\sigma_\alpha$, i.e.\ \(\iplay[\sigma_\alpha\oplus\tau]{\I} \in W\) for all strategies $\tau$ of Demon.
As in the case for $\alpha;\beta$, composing $\sigma_\alpha$ with the respective (state-dependent) winning strategies $\sigma_\tau$ for all possible resulting states (which are all in $W$) corresponding to the respective possible strategies $\tau$ that Demon could play during the first $\alpha$, thus leads to a winning strategy of the form
\[\sigma\mdefeq \fulltree{\{(\arepeat)\} \cup{}} \arepeat\aappend\sigma_\alpha \cup \cupfold \arepeat\aappend(\sigma_\alpha\oplus\tau)\aappend\sigma_\tau\]
for Angel to achieve $X$ in $\prepeat{\alpha}$ from $\iportray{\I}$, where the union is over all leaves \m{\sigma_\alpha\oplus\tau \in\leaf(\igameplay[\alpha]{\I})} in all strategies $\tau$ of Demon for which the game is not won by a player yet during the first $\alpha$.

The converse inclusion \m{\strategyfor[\prepeat{\alpha}]{X}\supseteq W} is equivalent to \m{\scomplement{\strategyfor[\prepeat{\alpha}]{X}}\subseteq\scomplement{W}}.
For this, recall
\(\scomplement{\strategyfor[\prepeat{\alpha}]{X}}=\dstrategyfor[\prepeat{\alpha}]{\scomplement{X}} = \cupfold\{Z\subseteq\linterpretations{\Sigma}{V} \with Z\subseteq \scomplement{X}\cap\dstrategyfor[\alpha]{Z}\}\)
by \rref{thm:dGL-determined}.
Thus, since \(\scomplement{\strategyfor[\prepeat{\alpha}]{X}}\) is a greatest (post-)fixpoint, it is enough to show $Z\subseteq \scomplement{W}$ for all $Z$ with \(Z\subseteq\scomplement{X}\cap\dstrategyfor[\alpha]{Z}\).
Since, \(Z\subseteq\dstrategyfor[\alpha]{Z}\), Demon has a winning strategy in $\alpha$ to achieve $Z$ from all \(\iget[state]{\I}\in Z\), by induction hypothesis.
By composing the respective winning strategies for Demon, obtain a winning strategy $\tau$ for Demon to achieve $Z$ in $\prepeat{\alpha}$ for \emph{any} arbitrary number of repetitions that Angel chooses (recall that Angel cannot choose to repeat $\prepeat{\alpha}$ infinitely often to win).
Since \(Z\subseteq\scomplement{X}\), Angel cannot have a winning strategy to achieve $X$ in $\prepeat{\alpha}$ from any $\iget[state]{\I}\in Z$ by \rref{thm:dGL-determined}.
Thus, \(Z\subseteq \scomplement{W}\).

\item \(\iportray{\I}\in\strategyfor[\pdual{\alpha}]{X} = \scomplement{\strategyfor[\alpha]{\scomplement{X}}}\).
\(\mbisubjunct \iportray{\I}\not\in\strategyfor[\alpha]{\scomplement{X}}\).
By induction hypothesis, this is equivalent to:
there is no winning strategy \(\sigma\subseteq\igameplay[\alpha]{\I}\) for Angel winning $\scomplement{X}$ in $\alpha$ from $\iportray{\I}$.
Since \(\strategyfor[\pdual{\alpha}]{X} = \dstrategyfor[\alpha]{X}\) by \rref{thm:dGL-determined}, this is equivalent to: 
there is a winning strategy \(\tau\subseteq\igameplay[\alpha]{\I}\) for Demon winning $X$ in $\alpha$ from $\iportray{\I}$.
Since the nodes where Angel acts swap with the nodes where Demon acts when moving from $\alpha$ to $\pdual{\alpha}$, this is equivalent to:
there is a winning strategy \(\sigma\subseteq\igameplay[\pdual{\alpha}]{\I}\) for Angel winning $X$ in $\pdual{\alpha}$ from $\iportray{\I}$ using
\(\sigma \mdefeq \fulltree{\{(\adual)\} \cup \adual\aappend\tau \cup{}} \adual\aappend\tau\aappend\adual\).
The converse direction uses that all strategies permitted for $\pdual{\alpha}$ begin and end with $\adual$.
\qedhere
\end{enumerate}
\end{proof}

\section{Alternative Semantics} \label{app:alternative-semantics}

To elaborate why the \dGL semantics is both natural and general, this section briefly considers alternative choices for the semantics, focusing on the role of repetition in the context of hybrid games.
It turns out that alternative semantics require prior bounds of repetitions of ${<}\omega$ (\rref{app:advance-notice-semantics}) and $\omega$ (\rref{app:omega-semantics}), respectively.

\subsection{Advance Notice Semantics} \label{app:advance-notice-semantics}
One alternative semantics is the \emph{advance notice semantics} for $\prepeat{\alpha}$, which requires the players to announce the number of times that game $\alpha$ will be repeated when the loop begins.
The advance notice semantics defines \(\strategyfor[\prepeat{\alpha}]{X}\) as \(\cupfold_{n<\omega}\strategyfor[\alpha^n]{X}\) where $\alpha^{n+1}\mequiv\alpha^n;\alpha$ and $\alpha^0\mequiv\,\ptest{\ltrue}$ and defines \(\dstrategyfor[\prepeat{\alpha}]{X}\) as \(\capfold_{n<\omega}\dstrategyfor[\alpha^n]{X}\).
When playing $\prepeat{\alpha}$, Angel, thus, announces to Demon how many repetitions $n$ are going to be played when the game $\prepeat{\alpha}$ begins and Demon announces how often to repeat $\drepeat{\alpha}$.
This advance notice makes it easier for Demon to win loops $\prepeat{\alpha}$ and easier for Angel to win loops $\drepeat{\alpha}$, because the opponent announces an important feature of their strategy immediately as opposed to revealing whether or not to repeat the game once more one iteration at a time as in \rref{def:HG-semantics}.
Angel announces the number $n<\omega$ of repetitions when $\prepeat{\alpha}$ starts.

In hybrid systems, the advance notice semantics and the least fixpoint semantics are equivalent (\rref{lem:HP-Scott-continuous}), but the advance notice semantics and \dGL's least fixpoint semantics are different for hybrid games.
The following formula is valid in \dGL (see \rref{fig:cex-advance}), but would not be valid in the advance notice semantics:
\begin{equation}
  x=1\land a=1\limply\ddiamond{\prepeat{(\dchoice{(\pupdate{\pumod{x}{a}};\pupdate{\pumod{a}{0}})}{\pupdate{\pumod{x}{0}}})}}{x\neq1}
  \label{eq:advance-notice-ex}
\end{equation}
\begin{figure}[tbhp]
    \tikzstyle{box}+=[minimum width=0.5cm,minimum height=0.5cm]
  \begin{tikzpicture}[grow'=down]
    \tikzstyle{level 2}=[sibling distance=+22mm]
    \tikzstyle{level 3}=[sibling distance=+12mm]
    \tikzstyle{level 4}=[sibling distance=+7mm]
    \tikzstyle{level 5}=[sibling distance=+6mm]
    \tikzstyle{level 6}=[level distance=1.2cm]
    \node[diamond] (PO) at (0,0) {11}
      child[pdia] {node[box] {11}
        child[pbox] {node[diamond] (01) {01}
          child[pdia] {node[box] {01}
            child[pbox] {node[diamond,diawon] (back01) {01}}
            child[pbox] {node[diamond] {10}
              child[pdia] {node[box] {10}
                child[pbox] {node[diamond,diawon] {00}}
                child[pbox] {node[diamond,diawon] {00}}
                edge from parent node[action] {repeat}
              }
            child[pdia] {node[box,boxwon] {10} edge from parent node[action] {stop}}
          }
            edge from parent node[action] {repeat}
          }
          child[pdia] {node[box,diawonstrategy] {01} edge from parent node[action] {stop}}
        }
        child[pbox] {node[diamond] {10}
          child[pdia] {node[box] {10}
            child[pbox] {node[diamond,diawonstrategy] {00}}
            child[pbox] {node[diamond,diawonstrategy] {00}}
            edge from parent node[action] {repeat}
          }
          child[pdia] {node[box,boxwon] {10} edge from parent node[action] {stop}}
        }
        edge from parent node[action] {repeat}
      }
      child[pdia] {node[box,boxwon] {11}
        edge from parent node[action] {stop}
      }
      ;
      \draw[backedge] (01) to[bend left=80] (back01);
    \tikzstyle{action}=[straight action,left]
    \tikzstyle{level 1}=[sibling distance=+32mm]
    \tikzstyle{level 1}=[sibling distance=+18mm]
    \tikzstyle{level 2}=[sibling distance=+24mm]
    \tikzstyle{level 3}=[sibling distance=+12mm]
    \tikzstyle{level 4}=[sibling distance=+7mm]
    \tikzstyle{level 5}=[sibling distance=+6mm]
    \tikzstyle{level 6}=[level distance=1.2cm]
    \node[diamond] (PO) at (7.5,0) {11}
      child[pdia,level 3/.style={sibling distance=+12mm},level 4/.style={sibling distance=+6mm}] {node[box] (3) {11}
        child[pbox] {node[box] {01}
          child[pbox] {node[box] {01}
            child[pbox] {node[box,diawon] {01}}
            child[pbox] {node[box,boxwonstrategy] {10}}
          }
          child[pbox] {node[box] {10}
            child[pbox] {node[box,diawon] {00}}
            child[pbox] {node[box,diawon] {00}}
          }
        }
        child[pbox] {node[box] {10}
          child[pbox] {node[box] {00}
            child[pbox] {node[box,diawon] {00}}
            child[pbox] {node[box,diawon] {00}}
          }
          child[pbox] {node[box] {00}
            child[pbox] {node[box,diawon] {00}}
            child[pbox] {node[box,diawon] {00}}
          }
        }
        edge from parent node[action] {3}
      }
      child[pdia,level 2/.style={sibling distance=+12mm},level 3/.style={sibling distance=+6mm}] {node[box] (2) at ($(3)+(-3.3,0)$) {11}
        child[pbox] {node[box] {01}
          child[pbox] {node[box,diawon] {01}}
          child[pbox] {node[box,boxwonstrategy] {10}}
        }
        child[pbox] {node[box] {10}
          child[pbox] {node[box,diawon] {00}}
          child[pbox] {node[box,diawon] {00}}
        }
        edge from parent node[action] {2}
      }
      child[pdia,level 2/.style={sibling distance=+6mm}] {node[box] (1) at ($(2)+(-1.5,0)$) {11}
        child[pbox] {node[box,diawon] {01}}
        child[pbox] {node[box,boxwonstrategy] {10}}
        edge from parent node[action] {1}
      }
      child[pdia] {node[box,boxwonstrategy] (0) at ($(1)+(-1.2,0)$) {11}
        edge from parent node[action] {0}
      }
      ;
      \node[diamonda] at ($(3)+(0,0.8)$) {\dots};
  \end{tikzpicture}
  \caption{Game trees for \(x=1\land a=1\limply\ddiamond{\prepeat{\alpha}}{x\neq1}\) with game
  \m{\alpha\mequiv\dchoice{(\pupdate{\pumod{x}{a}};\pupdate{\pumod{a}{0}})}{\pupdate{\pumod{x}{0}}}} (notation: $x,a$).
   \textbf{(left)} valid in \dGL by strategy ``repeat once and repeat once more if $x=1$, then stop''
   \textbf{(right)} false in advance notice semantics by the strategy ``$n-1$ choices of \m{\pupdate{\pumod{x}{0}}} followed by \m{\pupdate{\pumod{x}{a}};\pupdate{\pumod{a}{0}}} once'', where $n$ is the number of repetitions Angel announced
  }
    \label{fig:cex-advance}
\end{figure}%

\noindent
If, in the advance notice semantics, Angel announces that she has chosen $n$ repetitions of the game, then Demon wins (for $a\neq0$) by choosing the \m{\pupdate{\pumod{x}{0}}} option $n-1$ times followed by one choice of \m{\pupdate{\pumod{x}{a}};\pupdate{\pumod{a}{0}}} in the last repetition.
This strategy would not work in the \dGL semantics, because Angel is free to decide whether to repeat $\prepeat{\alpha}$ after each repetition based on the resulting state of the game.

Conversely, the dual formula would be valid in the advance notice semantics but is not valid in \dGL:
\begin{equation*}
x=1\land a=1\limply\dbox{\prepeat{(\dchoice{(\pupdate{\pumod{x}{a}};\pupdate{\pumod{a}{0}})}{\pupdate{\pumod{x}{0}}})}}{x=1}
\label{eq:advance-noticeb-ex}
\end{equation*}
The \dGL semantics is more general, because it gives the player in charge of repetition more control as the state can be inspected before deciding on whether to repeat again.
Advance notice semantics, instead, only allows the choice of a fixed number of repetitions.
The advance notice games can be expressed easily in \dGL by having the players choose a counter $c$ before the loop that decreases to 0 during the repetition.
The advance notice semantics can be expressed in \dGL, e.g., for \rref{eq:advance-notice-ex} as
\begin{equation*}
\hspace{-0.3cm}
x=1\land a=1\limply\langle\pupdate{\pumod{c}{0}};\prepeat{\pupdate{\pumod{c}{c+1}}};\prepeat{((\dchoice{(\pupdate{\pumod{x}{a}};\pupdate{\pumod{a}{0}})}{\pupdate{\pumod{x}{0}}});\pupdate{\pumod{c}{c-1}})};\ptest{c=0}\rangle{x\neq1}
\end{equation*}

\subsection{$\omega$-Strategic Semantics} \label{app:omega-semantics}

Another alternative choice for the semantics would have been to allow only arbitrary finite iterations of the strategy function for computing the winning region by using the \emph{$\omega$-strategic semantics}, which defines \(\strategyfor[\prepeat{\alpha}]{X}\) as \(\inflop[\omega]{X} = \cupfold_{n<\omega}\inflop[n]{X}\) along with a corresponding definition for \(\dstrategyfor[\prepeat{\alpha}]{X}\).
Like the \dGL semantics, but quite unlike the advance notice semantics, the $\omega$-strategic semantics does not require Angel to disclose how often she is going to repeat when playing $\prepeat{\alpha}$.
Similarly, Demon does not have to announce how often to repeat when playing $\drepeat{\alpha}$.
Nevertheless, the semantics are different.
The $\omega$-strategic semantics would make the following valid \dGL formula invalid:
\begin{equation}
\ddiamond{\prepeat{(\pchoice{\pupdate{\pumod{x}{1}}; \devolve{\D{x}=1}}{\pupdate{\pumod{x}{x-1}}})}}{\, (0\leq x<1)}
\label{eq:omega-strategic-ex}
\end{equation}
By a simple variation of the argument in the proof of \rref{thm:dGL-closure-lower}, \m{\inflop[\omega]{[0,1)}=[0,\infty)}, because \m{\inflop[n]{[0,1)}=[0,n)} for all $n\in\naturals$.
Yet, this $\omega$-level of iteration of the strategy function for winning regions misses out on the perfectly reasonable winning strategy ``first choose \(\pupdate{\pumod{x}{1}}; \devolve{\D{x}=1}\) and then always choose \(\pupdate{\pumod{x}{x-1}}\) until stopping at \(0\leq x<1\)''.
The existence of this winning strategy is only found at the level \(\inflop[\omega+1]{[0,1)}=\strategyfor[\alpha]{[0,\infty)}=\reals\).
Even though any particular use of the winning strategy in game play uses only some finite number of repetitions of the loop, the argument why it will always work requires $>\omega$ many iterations of $\strategyfor[\alpha]{\cdot}$, because Demon can change $x$ to an arbitrarily big value, so that $\omega$ many iterations of $\strategyfor[\alpha]{\cdot}$ are needed to conclude that Angel has a winning strategy for all positive values of $x$.
There is no smaller upper bound on the number of iterations it takes Angel to win, in particular Angel cannot promise $\omega$ as a bound on the repetition count, which is what the $\omega$-semantics would require her to do.
But strategies do converge after $\omega+1$ iterations.
According to \rref{thm:dGL-closure-lower}, similar shortcomings would apply for a semantics that cuts winning region iteration of at higher transfinite ordinals below $\omega_1^{\text{HG}}$.

The \dGL semantics is also more general, because, by \rref{thm:dGL-closure-lower}, its closure ordinal is ${\geq}\dGLordinal$, in contrast to the $\omega$-semantics, which has closure ordinal $\omega$ by construction.
The same observation shows a fundamental difference between the \dGL semantics and the advance notice semantics, which has closure ordinal ${\leq}\omega$.

\section{Proof of Higher Closure Ordinals} \label{app:dGL-closure-lower}

This section illustrates that closure ordinals are not a simple function of the syntactic structure, because minor syntactic variations lead to vastly different closure ordinals.

\begin{proof}[of \rref{thm:dGL-closure-lower}]
In this proof, proceed in stages of increasing difficulty.
That the closure ordinal is ${\geq}\omega\cdot2$ has already been shown on p.\,\pageref{thm:dGL-closure-lower}.
Now prove the bounds ${\geq}\omega^2$ and finally ${\geq}\omega^\omega$.
\renewcommand{\inflop}[2][]{\inflopstrat[#1]{\alpha}{#2}}%
In order to see that the closure ordinal is at least $\omega^2$ even for a single nesting layer of dual and loop, follow a similar argument using more variables.
Consider the family of formulas (for some $N\in\naturals$) of the form
\[
\ddiamond{\prepeat{\big(
  \underbrace{
 \pchoice{\pchoice{\pchoice
  {\pupdate{\pumod{x_N}{x_N-1}};\devolve{\D{x_{N-1}}=1}}
  {\dots\newline}}
  {\pupdate{\pumod{x_2}{x_2-1}};\devolve{\D{x_1}=1}}}
  {\pupdate{\pumod{x_1}{x_1-1}}}
  }_\alpha
  \big)}}{\landfold_{i=1}^N x_i<0}
\]
The winning regions for this \dGL formula stabilizes after $\omega\cdot N$ iterations, because $\omega$ many iterations are necessary to show that \emph{all} $x_1$ can be reduced to $(-\infty,0)$ by choosing the last action sufficiently often, whereas another $\omega$ many iterations are needed to show that $x_2$ can then be reduced to $(-\infty,0)$ by choosing the second-to-last action sufficiently often, increasing $x_1$ arbitrarily under Demon's control, which can still be won because this adversarial increase in $x_1$ can be compensated for by the first part of the winning strategy.
The vector space of variables $(x_N,\dots,x_1)$ is used in that order.
It is easy to see that
\(\inflop[\omega]{(-\infty,0)^N}=\cupfold_{n<\omega}\inflop[n]{(-\infty,0)^N}=(-\infty,0)^{N-1}\times\reals\),
because \(\inflop[n+1]{(-\infty,0)}=(-\infty,0)^{N-1}\times(-\infty,n)\) holds for all $n\in\naturals,n$ by a simple inductive argument:
\begin{align*}
  \inflop[1]{(-\infty,0)^N} &= (-\infty,0)^N\\
  \inflop[n+1]{(-\infty,0)^N} &= (-\infty,0)^N\cup\strategyfor[\alpha]{\inflop[n]{(-\infty,0)^N}}
  \\&= (-\infty,0)^N\cup\strategyfor[\alpha]{(-\infty,0)^{N-1}\times(-\infty,n-1)}
  = (-\infty,0)^{N-1}\times(-\infty,n)
\end{align*}
Inductively,
\m{\inflop[\omega\cdot(k+1)]{(-\infty,0)^N}=\cupfold_{n<\omega}\inflop[\omega\cdot k+n]{(-\infty,0)^N}=(-\infty,0)^{N-k-1}\times\reals^{k+1}}
holds,
because \m{\inflop[\omega\cdot k+n+1]{(-\infty,0)}=(-\infty,0)^{N-k-1}\times(-\infty,n)\times\reals^k} holds for all $n\in\naturals$ by a simple inductive argument:
\begin{align*}
  \inflop[\omega\cdot k+n+1]{(-\infty,0)^N} &= (-\infty,0)^N\cup\strategyfor[\alpha]{\inflop[\omega\cdot k+n]{(-\infty,0)^N}}
  \\&= (-\infty,0)^N\cup \strategyfor[\alpha]{(-\infty,0)^{N-k-1}\times(-\infty,n-1)\times\reals^k}
  \\&= (-\infty,0)^{N-k-1}\times(-\infty,n)\times\reals^k
\end{align*}
Consequently,
\m{\strategyfor[\prepeat{\alpha}]{(-\infty,0)^N} = \inflop[\omega\cdot N]{(-\infty,0)^N} \neq \inflop[\omega\cdot(N-1)+n]{(-\infty,0)^N}}, which then makes $\omega\cdot N$ the closure ordinal for $\alpha$.
Since hybrid games $\alpha$ of the above form can be considered with arbitrarily big $N\in\naturals$, the common closure ordinal has to be ${\geq}\omega\cdot N$ for all $N\in\naturals$, i.e.\ it has to be ${\geq}\omega^2$.

In order to see that the closure ordinal is at least $\omega^\omega$, follow an argument expanding on the previous case.
Consider the family of formulas (for some $N\in\naturals$) of the form
\[
\ddiamond{\prepeat{\underbrace{\big(
 \pchoice{\pchoice{\pchoice
  {\ptest{x_{N{-}1}{<}0}{;}\devolve{\D{x_{N{-}1}}{=}1}{;}\pupdate{\pumod{x_N}{x_N{-}1}}}
  {\sdots}}
  {\ptest{x_1{<}0}{;}\devolve{\D{x_1}{=}1}}{;}\pupdate{\pumod{x_2}{x_2{-}1}}}
  {\pupdate{\pumod{x_1}{x_1{-}1}}}
 \big)}_{\alpha}
  }}{\landfold_{i=1}^N x_i{<}0}
\]
The winning region for this ``clockwork $\omega$'' formula stabilizes after $\omega^N$ iterations, $\omega$ many iterations are necessary to show that \emph{all} $x_1$ can be reduced to $(-\infty,0)$ by choosing the last action sufficiently often, whereas another $\omega$ many iterations are needed to show that $x_2$ can then be reduced to $(-\infty,0)$ by choosing the second-to-last action sufficiently often in case $x_1$ has already been reduced to $(-\infty,0)$.
Every time the second-to-last action is chosen, however, Demon increases $x_1$ arbitrarily, which again takes $\omega$ many steps of the last action to understand how $x_1$ can again be reduced to $(-\infty,0)$ before the second-to-last action can be chosen again to decrease $x_2$ further.
This phenomenon that $\omega$ many actions on $x_{i-1}$ are needed before $x_i$ can be decreased by 1 holds for all $i$ recursively.
Note that in any particular game play, Demon can only increase $x_i$ by some finite amount.
But Angel does not have a finite bound on that increment, so she will first have to convince herself that she has a winning strategy that could tolerate any change in $x_i$, which takes $\omega$ many iterations of the previous argument.

The vector space of variables $(x_N,\dots,x_1)$ is used in that order.
For $b_N,\dots,b_1\in\naturals\cup\{\infty\}$, use the short hand notation
\[
b_N\dots b_2 b_1 \mdefeq (-\infty,b_N)\times\dots\times(-\infty,b_2)\times(-\infty,b_1)
\]
and write $b_i^n$ for $(-\infty,b_i)^n$ in that context.
Let $\vec{b}=(b_N,\dots,b_1)$.
Then prove that \m{\mforall{n\in\naturals}{\mforallr{j\in\naturals,j>0}{}}}
\begin{align*}
  \inflop[\omega^j(n+1)]{b_N\sdots b_j \sdots b_1} &{=} {b_N\sdots  (b_{j+1}+n) \infty^j} &&\text{if}~b_N \sdots b_j<\infty, j>0
  \\
  \inflop[\omega^j(n+1)]{b_N\sdots b_{j+1}\infty^j} &{=} {b_N\sdots (b_{j+1}+n+1) \infty^j} &&\text{if}~b_N \sdots b_{j+1}{<}\infty,b_j{=}\infty{=}\sdots b_1
  \\
  \inflop[\omega^j(n+1)]{b_N\sdots b_{k+1}\infty^{k-j}\infty^j} &{=} {b_N\sdots (b_{k+1}{+}1)1^{k{-}j{-}1}(n{+}1) \infty^j}\hspace*{-8pt} %
  &&\text{if}~b_N \sdots b_{k+1}{<}\infty,b_k{=}\infty\ignore{= b_1},k{>}j
  \\&\quad \cup \vec{b} 
\end{align*}
by induction on the lexicographical order of $j$ and $n$.
Let \textcircled{1}, \textcircled{2},  \textcircled{3} denote the if conditions on the right, respectively.
Note that, in the case \textcircled{3}, there are some subordinate cases which do not need to be tracked in this analysis, because they are strategic dead ends.
IH is short for induction hypothesis.

The base case $j=0,n=0$ is vacuous for \textcircled{1} and can be checked easily for \textcircled{2}.
\begin{align*}
  \inflop[\omega^01]{b_N\sdots b_1\infty^0} &= \inflop[1]{b_N\sdots b_1} = {b_N\sdots (b_{1}+1)} = {b_N\sdots (b_{1}+1) \infty^0}\\
  \inflop[\omega^0(n+1)]{b_N\sdots b_1\infty^0} &= \vec{b}\cup\inflop{\inflop[n]{b_N\sdots b_1}} = \vec{b}\cup\inflop{b_N\sdots (b_{1}+n)} = {b_N\sdots (b_{1}+n+1)}
\end{align*}
For \textcircled{3}, the case $j=0$ holds only after an extra offset $k$, however:
\begin{align*}
  \inflop[1]{b_N\sdots b_{k+1}\infty^k} &= \vec{b}\cup {b_N\sdots (b_{k+1}+1)0\infty^{k-1}}\\
  \inflop[n+1]{b_N\sdots b_{k+1}\infty^k} &= \inflop[n]{b_N\sdots b_{k+1}\infty^k} \cup {b_N\sdots (b_{k+1}+1)1^n0\infty^{k-n-1}} \quad\text{for}~n<k\\
  \inflop[k+n+1]{b_N\sdots b_{k+1}\infty^k} &= \inflop[k+n]{b_N\sdots b_{k+1}\infty^k} \cup {b_N\sdots (b_{k+1}+1)1^{k-1}(n+1)}
\end{align*}
Instead, prove base case $j=1,n=0$, as the extra offset $k$ has been overcome at $\omega$:
\begin{align*}
  \inflop[\omega^11]{b_N\sdots b_1} &= \cupfold_{n<\omega} \inflop[\omega^0(n+1)]{b_N\sdots b_1\infty^0}
  = \cupfold_{n<\omega} {b_N\sdots (b_1{+}n{+}1)} = {b_N\sdots b_2 \infty} &&\text{if \textcircled{1}}\\
  \inflop[\omega^11]{b_N\sdots b_2\infty} &= \cupfold_{n<\omega} \inflop[\omega^0(n+1)]{b_N\sdots b_2\infty^1}
  =  b_N\sdots (b_2{+}1)\infty &&\text{if \textcircled{2}}\\
  \inflop[\omega^11]{b_N\sdots b_{k{+}1}\infty^k} &= \cupfold_{n<\omega} \inflop[\omega^0(n{+}1)]{b_N\sdots b_{k{+}1}\infty^k} = \cupfold_{n<\omega} {b_N\sdots (b_{k{+}1}{+}1)1^{k-1}(n{+}1)} \cup \vec{b}
  \hspace*{-0.6cm}
  \\& = {b_N\sdots (b_{k+1}+1)1^{k-1}\infty} \cup \vec{b}
  &&\text{if \textcircled{3}}
\end{align*}
In case \textcircled{3}, there are some subordinate cases $\cup\vec{b}$ coming from mixed occurrences $b_N\sdots(b_{k+1}+1)^i0\infty^{k-i-1}$, but do not need to be tracked, because they are strategic dead ends.
By construction of $\alpha$, no counter can be changed without resetting all smaller variables to 0 first as indicated.

$j\curvearrowright j+1,n=0$: For the step from $j$ to $j+1$ prove the case $n=0$ as follows.
\begin{align*}
  & \inflop[\omega^{j+1}\cdot(0+1)]{b_N\sdots b_j \sdots b_1} 
  = \inflop[\omega^j\cdot\omega]{b_N\sdots b_j \sdots b_1}
  = \cupfold_{n<\omega}\inflop[\omega^j\cdot(n+1)]{b_N\sdots b_j \sdots b_1}\\
  &
  \stackrel{IH}=
  \begin{cases}
  \cupfold_{n<\omega} {b_N\sdots (b_{j+1}+n) \infty^j} &\text{if \textcircled{1}}\\
  \cupfold_{n<\omega} {b_N\sdots (b_{j+1}+n+1) \infty^j} &\text{if \textcircled{2}}\\
  \cupfold_{n<\omega} {b_N\sdots (b_{k+1}+1)1^{k-j-1}(n+1) \infty^j \cup \vec{b}} &\text{if \textcircled{3}}
  \end{cases}
  \\
  &
  \stackrel{IH}=
  \begin{cases}
  {b_N\sdots b_{j+2} \infty^{j+1}} &\text{if}~b_N,\sdots,b_j<\infty\\
  {b_N\sdots b_{j+2} \infty^{j+1}} &\text{if}~b_N,\sdots,b_{j+1}<\infty\\
  {b_N\sdots (b_{j+2}+1) \infty^{j+1}} &\text{if}~b_N,\sdots,b_{j+2}<\infty,b_{j+1}=\infty,k=j+1\\
  {b_N\sdots (b_{k+1}+1)1^{k-j-2}1\infty^{j+1}} \cup \vec{b} %
  &\text{if}~b_N,\sdots,b_{k+1}<\infty,b_k=\infty,k>j+1
  \end{cases}
\end{align*}
$n\curvearrowright n+1$:
Within a level $j$, prove the step from $n$ to $n+1$ as follows.
If $n=0$, then \(\inflop[\omega^j(n+1)]{b_N\sdots b_j \sdots b_1}=\inflop[\omega^j]{b_N\sdots b_j \sdots b_1}\) already has the property by induction hypothesis.
Otherwise $n>0$, which allows to conclude:
\begin{align*}
  & \inflop[\omega^j(n+1)]{b_N\sdots b_j \sdots b_1} 
  = \inflop[\omega^jn+\omega^j]{b_N\sdots b_j \sdots b_1} 
  \stackrel{\text{\rref{lem:inflop-inductive-homomorphic}}}{=} \inflop[\omega^j]{\inflop[\omega^jn]{b_N\sdots b_j \sdots b_1}}\\
  &
  \stackrel{IH}=
  \begin{cases}
  \inflop[\omega^j]{b_N\sdots (b_{j+1}+n-1) \infty^j} &\text{if \textcircled{1}}\\ %
  \inflop[\omega^j]{b_N\sdots (b_{j+1}+n) \infty^j} &\text{if \textcircled{2}}\\
  \inflop[\omega^j]{b_N\sdots (b_{k+1}+1)1^{k-j-1}n \infty^j \cup \vec{b}
  } 
  &\text{if \textcircled{3}}
  \end{cases}
  \\
  &
  \stackrel{IH}=
  \begin{cases}
  {b_N\sdots (b_j+n) \infty^j} &\text{if \textcircled{1}}\\
  {b_N\sdots (b_j+n+1) \infty^j} &\text{if \textcircled{2}}\\
  {b_N\sdots (b_{k+1}+1)1^{k-j-1}(n+1) \infty^j} \cup \vec{b} %
  &\text{if \textcircled{3}}
  \end{cases}
\end{align*}
Consequently,
\m{\strategyfor[\prepeat{\alpha}]{(-\infty,0)^N} = \inflop[\omega^N]{(-\infty,0)^N} =\reals^N\neq \inflop[\omega^{N-1}\cdot n]{(-\infty,0)^N}} for all $n\in\naturals$,
which makes $\omega^N$ the closure ordinal for $\alpha$.
Since hybrid games $\alpha$ of the above form can be considered with arbitrarily big $N\in\naturals$, the common closure ordinal has to be $\geq\omega^N$ for all $N\in\naturals$, i.e.\ it has to be $\geq\omega^\omega$.
\end{proof}

\begin{acks}
I thank Stephen Brookes, Frank Pfenning, James Cummings, and Erik Zawadzki for helpful discussions and I appreciate the helpful comments by the reviewers.
\end{acks}

\bibliographystyle{ACM-Reference-Format-Journals}
\bibliography{dGL}

\received{August 2014}{February 2015}{July 2015}

\end{document}